\documentclass{article}

\usepackage[a4paper]{geometry}

\usepackage{times}

\usepackage{url}

\usepackage{algorithm}
\usepackage{algpseudocode}
\usepackage{booktabs, colortbl, diagbox}
\usepackage{amsmath}
\usepackage{amssymb}
\usepackage{amsfonts}
\usepackage{xspace}
\usepackage{enumitem} \setlist[itemize]{noitemsep} \setlist[enumerate]{noitemsep}
\usepackage{color}
\usepackage{graphicx}
\usepackage{multirow}
\usepackage{authblk}
\usepackage{amsthm}
\usepackage{caption} \usepackage{subcaption}

\usepackage{bussproofs}
\usepackage{cmll}
\usepackage{balance}

\usepackage{wasysym}

\newtheorem{theorem}{Theorem}

\newtheorem{proposition}[theorem]{Proposition}
\newtheorem{lemma}[theorem]{Lemma}

\newtheorem{definition}[theorem]{Definition}
\newtheorem{example}[theorem]{Example}

\newcolumntype{C}[1]{>{\centering\let\newline\\\arraybackslash\hspace{0pt}}m{#1}}

\newcommand\pcase[1]{}

\newcommand{\logica}{{\sf LOG}\xspace}
\newcommand{\iNE}{{\sf NE}\xspace}

\renewcommand{\phi}{\varphi}
\renewcommand{\implies}{\multimap}
\newcommand{\lineg}{\mathop \sim}

\newcommand{\ch}{\mathsf{ch}}

\newcommand{\out}{\mathsf{out}}
\newcommand{\redis}{\mathsf{redis}}

\newcommand{\var}[1]{\mathsf{#1}}

\newcommand{\arrow}{\vartriangleright}

\newcommand\diNE{{\square}}
\newcommand\paNE{{\blacksquare}}
\newcommand\dipaNE{{\blacksquare}}

\newcommand\transparency{0.90}

\title{
Individual Resource Games and Resource Redistributions
}
\author{Nicolas Troquard}

\affil{
  The KRDB Research Centre for Knowledge and Data\\
  Faculty of Computer Science\\
  Free University of Bozen-Bolzano\\
  Piazza Domenicani, 3\\
  I-39100 Bozen-Bolzano BZ, Italy\\
\texttt{Nicolas.Troquard@unibz.it}\\
}
\date{}

\begin{document}

\maketitle

\begin{abstract}
  To introduce agent-based technologies in real-world systems, one needs to acknowledge that the agents often have limited access to resources. They have to seek after resource objectives and compete for those resources.
  
    We introduce a class of resource games where resources and preferences are specified with the language of a resource-sensitive logic.
    The agents are endowed with a bag of resources and try to achieve a resource objective.
    For each agent, an action consists in making available a part of their endowed resources.
    All the resources made available can be used towards the agents' objectives.

    We study three decision problems, the first of which is deciding whether an action profile is a Nash equilibrium: when all the agents have chosen an action, it is a Nash Equilibrium if no agent has an incentive to change their action unilaterally.
    
    When dealing with resources, interesting questions arise as to whether some equilibria can be eliminated or constructed by a central authority by redistributing the available resources among the agents.
  In our economies, division of property in divorce law exemplifies how a central authority can redistribute the resources of individuals, and why they would desire to do so.
  We thus study two related decision problems:
  \begin{itemize}
  \item rational elimination: given an action profile's outcome, can the endowed resources be redistributed so that it is not the outcome of a Nash equilibrium.
  \item rational construction: given an action profile's outcome, can the endowed resources be redistributed so that it is the outcome of a Nash equilibrium.
  \end{itemize}
  
  Among other results, we prove that all three problems are $\mathsf{PSPACE}$-complete when the resources are described in the very expressive language of the propositional multiplicative and additive Linear Logic.

  We also identify a new modest fragment of Linear Logic that we call MULT, suitable to represent multisets and reason about the inclusion and equality of bags of resources.
  We show that when the resources are described in MULT, the problem of deciding whether a profile is a Nash equilibrium is in $\mathsf{PTIME}$.
\end{abstract}

\section{Introduction}

Agents, or players, are entities capable of action and trying to reach their goals.
In the physical or cyber world, these agents have limited access to resources. They have to seek after resource objectives and compete for those resources.

This paper makes use of resource-sensitive logics, Linear Logic~\cite{girard1987} specifically, to model and solve problems of rational agents interacting in a resource-aware environment.
We use Linear Logic to define and reason about a new class of non-cooperative games~\cite{Osborne1994}. Every Linear Logic formula represents a resource. In these games, each agent is endowed with a bag of resources, and has an objective to achieve by transforming the resources made available from the agents' endowed resources.
Can we decide whether the resources made available by the agents constitute a Nash equilibrium, that is, whether it is locally optimal under individual strategic considerations?
If a local optimal is not desirable, could an arbitrator redistribute the resources in the endowments among the agent so that is not a Nash equilibrium anymore, thus eliminating it?
To the contrary, if an outcome is desirable, could an arbitrator redistribute the resources so that is becomes the outcome of a Nash equilibrium, thus constructing it?
In this paper we are going to address the computational complexity of the decision problems corresponding to these questions.

As we study the computational complexity of answering these questions about resource-sensitive game theoretical interactions, we will be particularly interested in a few varying parameters:
\begin{itemize}
\item What kind of preferences the agents have?
  \begin{itemize}
  \item Do they only care about reaching their resource objectives? (dichotomous)
  \item Do they also care about how much resource the consume? (parsimonious)
  \end{itemize}
\item What is the exact language for talking about resources, and what is the complexity of reasoning about resources in this language?
\item Which resource-sensitive logic exactly is used to reason about the resources?
  \begin{itemize}
  \item Can resources be disposed of freely during reasoning? (affine reasoning)
  \item Must all resources be accounted for during reasoning? (linear reasoning)
  \end{itemize}
\end{itemize}

This paper is putting together:
\begin{enumerate}
\item \textbf{Linear Logic}, which enables the specification of resources and
  the reasoning about them.
\item \textbf{Game theory and Nash equilibria}, which give us a guideline to
  characterize normatively good outcome in games whose actions and
  preferences are defined in terms of the resources expressed in
  Linear Logic.
\item \textbf{Computational complexity}, which helps us towards an algorithmic
  treatment of our resource games. It is intimately affected by the
  precise Linear Logic used to represent the resources.
\end{enumerate}

The models and the algorithms presented here can be used as analytical tools at the disposition of actors and policy makers, for instance in interconnected economies~\cite{ambaketal1994interconnection,brock1995interconnection}. They can serve at gauging the possible strategic behaviours of the actors and of their competitors, and at identifying possible issues of resource scarcity in a commons.

Our games are reminiscent of notable models existing in the literature.
They share the logic-based approach of Boolean games~\cite{Harrenstein:2001:BG:1028128.1028159,DBLP:conf/ecai/BonzonLLZ06}.
In Boolean games, each player controls a set of Boolean variables and produces truth values which can be used without restriction towards the Boolean goals. As such, resources proper are absent from Boolean games.
Our games also share the resource-sensitiveness of congestion games~\cite{Rosenthal1973}.
In congestion games, the players choose a set of resources (e.g., edges to travel in a graph), and their utility depends on the cost (e.g., delays) of using the shared resources, which depends on the number of players travelling them.
Despite some apparent similarities, they are rather superficial. One thing should be obvious: the resources in congestion games are limited to basic resources and lack a rich specification language of resources like the one of resource-sensitive logics.

Using resource-sensitive logic languages to represent goods that are
transformed and exchanged between agents owes to previous
work, e.g., \cite{DBLP:conf/atal/HarlandW02,DBLP:conf/kr/PorelloE10,DBLP:conf/ecai/PorelloE10},
in multiagent systems and computational social
choice~\cite{DBLP:books/cu/FKW2014, Brandt:2016:HCS:3033138}.

A short version of this paper appeared as~\cite{Tr16ijcai}.

\paragraph{Logic: exploiting resource-sensitive logics.}
In this paper, we study games of resources that are aimed at representing the strategic interactions between rational agents %~\cite{Osborne1994}
where some combinations of resources replace the
abstract notions of action and preferences. In these games,
players are endowed with some resources and have preferences upon some
resources to be available after the game is played. Players' actions
also consist in making available some of the resources they are
endowed with.

We propose a class of games of resources that exploits the formalisms and reasoning methods coming from the
literature in knowledge representation and computational logics,
namely resource-sensitive logics: e.g., Linear Logic, Separation
Logic, BI Logic~\cite{girard1987,Reynolds:2002:SLL:645683.664578,DBLP:journals/bsl/OHearnP99}. The languages of these logics allow a
fine-grained description of resources, processes, and their harmonious
combinations. In computer science, they have been quite successful at
modeling systems for multi-party access and modification of shared
structures, by allocation and deallocation of resources. 
The resources used in this paper are not based on a trivial and na\"ive set theory. Instead, they are based on rich logical languages, supported by elaborate reasoning features.

A resource is represented by one formula of a
resource-sensitive logic \logica. More specifically, we assume here that
\logica is some propositional variant of Linear Logic.
We provide an informal presentation of the resource interpretation of Linear Logic in Section~\ref{sec:resources-ll} so that the conceptual aspects of the paper can be grasped without a great understanding of Linear Logic.

\paragraph{Game theory: individual resource games.}

We will consider individual resource games defined formally in
Section~\ref{sec:irg-dp}. Each player $i$ of a game will be endowed
with a \emph{multiset} of resources $\epsilon_i$. An action for
Player~$i$ will be to contribute a subset of
$\epsilon_i$. An \emph{(action) profile} specifies a contribution for every player.
An \emph{outcome} will be a context consisting of a \emph{multiset} of resources resulting from a profile.
Then, each player $i$ has a goal $\gamma_i$, which is a resource, represented by \emph{one} formula of \logica.
An outcome $X$ satisfies the goal of Player $i$ if there is a proof of $X \vdash
\gamma_i$ in the logic \logica.
This will mean that the resources in $X$ \emph{can} be consumed so as to produce $\gamma_i$.\footnote{Indeed, $X \vdash \gamma_i$ indicates that the resources $X$ are sufficient to produce $\gamma_i$, and $X \vdash \gamma_j$ indicates that the resources $X$ are sufficient to produce $\gamma_j$. It may be however that the resources $X$ are not sufficient to produce $\gamma_i$ and $\gamma_j$ simultaneously.}

\begin{quote}\itshape
Intuitively, we can imagine a game taking place around a table. Each
player has an objective to create some resource. Each player has also
a bag of resources. To play, each player chooses to take some resources
(possibly none) from their respective bags and put them on the table
in front of them. The outcome is the collection of resources on the
table after every player has chosen. A player is satisfied if we
\emph{can} transform the resources on the table so as to produce her
goal. It is a Nash equilibrium when no player has an incentive to take
back any resources she put on the table, or to add more resources from
her bag.
\end{quote}
What should be an incentive to take back or to add resources? We will study these games of resources with two kinds of preferences.
We will first consider, in Section~\ref{sec:dicho-pref}, preferences over outcomes that are \emph{dichotomous}.
We can thus initially say that Player~$i$ prefers an outcome $X$ over an outcome
$Y$ iff $X \vdash \gamma_i$ and $Y \not\vdash \gamma_i$. Some formal
results will lead us to define in Section~\ref{sec:parsi-pref}, \emph{parsimonious} preferences, a finer notion of preference where $i$ may be
qualitatively indifferent between $X$ and $Y$, but still prefer $X$
over $Y$ because $i$'s contribution is strictly less in $X$ than in
$Y$.

\paragraph{Algorithms and complexity: solving problems.}
We will study three decision problems defined also in
Section~\ref{sec:irg-dp}, the first of which is deciding whether an
action profile is a Nash equilibrium.
A Nash equilibrium is, under strategic considerations, a local optimal.
A situation in which every agent has picked an action is a Nash equilibrium when no agent has an incentive to change their mind.
A variant of this example, with one additional player, will be formalized later in Section~\ref{sec:inter-econ}.

\begin{example}\label{example:telecom1}
  In a local telecom industry, anti-trust laws forbid a priori cooperation, and regulations oblige the companies to accept traffic from each other.
  (These telecom companies operate in an interconnected economy~\cite{ambaketal1994interconnection,brock1995interconnection}.)
  Consider two competing telecommunication companies. 
  Company~$A$ manages a 3G network of comprised capacity~$3$ (bundled as capacities~$1$, and $2$). 
  Company~$B$ manages a 4G network of capacity~$3$ (bundled as capacities~$1$, and $2$).
  Company~$A$ need to offer their customers 3G at capacity~$2$ and 4G at capacity~$1$. Company~$B$ need to offer their customers 3G at capacity $2$ and 4G at capacity~$2$.

Activating a network at some capacity has a cost. 
%Serving under-capacitated network to customers may yield to various technical and economic failures. 
Companies can privately activate and deactivate networks on the fly. 
What are the possible equilibria?

There are two Nash equilibria. First, there is the one where Company~$A$ provides a bundle of two 3G antennas and Company~$B$ provides a bundle of two 4G antennas. Both companies can achieve their goal, and none has an incentive to reduce their contribution as they would not satisfy their goal anymore. Second, there is the one where both Company~$A$ and Company~$B$ contribute nothing. None of them has an incentive to change their contribution since they would not be able to achieve their goal on their own.

%% \begin{verbatim}
%% Game: ["(1, Player: Name: Company A. 
%% Endowment: {3G, ('3G', '3G')}. Goal: {3G, 3G, 4G}.)", 
%% "(2, Player: Name: Company B. 
%% Endowment: {4G, ('4G', '4G')}. Goal: {3G, 3G, 4G, 4G}.)"]
%% **** pct
%% Nash: Profile: [(1, '{}'), (2, '{}')]
%% Nash: Profile: [(1, "{('3G', '3G')}"), (2, "{('4G', '4G')}")]
%% \end{verbatim}

\end{example}

When dealing with resources,
interesting questions arise as to whether some equilibria can be
eliminated or constructed by a central authority by redistributing the
available resources among the players~\cite{harrenstein2015aamas}.
In the tradition of social mechanism design, redistribution schemes can be used by a central authority to enforce some behavior, either by disincentivizing a behavior or incentivizing a behavior.
Formal frameworks dealing with redistribution schemes and economic policies have been studied \cite{Endriss:2011:DIB:2030470.2030482,Levit:2013:TSB:2484920.2484952,naumovmodal}.

Some profiles that are not equilibria can have desirable outcomes. Some equilibria can have outcomes that are undesirable.
Desirability must here be understood from the point of view of a system designer.
A system designer can redistribute the resources of the players in a game so as to steer the interaction to or away from a particular outcome.
\begin{quote}\itshape
  A redistribution consists in reallocating the resources endowed to the players. To every redistribution corresponds a new game where the players maintain their objectives, but their possible actions have changed. If $G^\epsilon$ is the original game, and $\epsilon'$ is a redistribution of the endowment function $\epsilon$, then $G^{\epsilon'}$ is a new game.
\end{quote}
We will investigate how resource distribution schemes can contribute to eliminate undesirable game equilibria, and construct desirable game equilibria.
They are a form of redistribution of wealth, which consists in wealth being transferred from some individuals to others. In our economies, it exists in the form of social mechanisms such as taxation, public services, and confiscation. Division of property and division of debt in divorce law are good imagery of what a designer can do in the mechanisms we propose in this paper. This example will be formalized later in Section~\ref{sec:divorce}.
\begin{example}\label{ex:divorce}
  Ann and Bernard, a couple of bakers, have filed for divorce. Ann is officially the tenant of the business premises of the bakery. Bernard is the owner of the baking equipment. He also owns enough flour to make bread for two years. 
  Ann would like to be able to keep the means of production, and being able to make bread for one year.
  Bernard wants to keep the shop. In this context, if Ann and Bernard are parsimonious, the outcome is very likely to be the one where Ann does not use the shop and Bernard does not use the breadmaking equipment and the flour. It is the only equilibrium. Neither of them satisfy their objective.

  However, an arbitrator can redistribute their endowments. He can give the equipment and half the flour to Ann, and give the shop to Bernard. Doing so, the outcome where Ann and Bernard do not use any of their endowment can be eliminated. Moreover, a new outcome equilibrium can be constructed where both satisfy their objectives.
\end{example}

We will thus look at two decision problems related to Nash equilibria: \emph{rational elimination} and \emph{rational construction} of Nash equilibria.
\begin{quote}\itshape
In a game $G^{\epsilon}$, a profile can be rationally eliminated from a game if there exists a redistribution $\epsilon'$ of $\epsilon$ such that there is no profile with the same outcome which is a Nash equilibrium in $G^{\epsilon'}$. A profile can be rationally constructed if there exists a redistribution $\epsilon'$ such that there is a profile in $G^{\epsilon'}$ with the same outcome, which a Nash equilibrium. 
\end{quote}

\paragraph{Outline.} We make a brief presentation of Linear Logic in Section~\ref{sec:resources-ll}. We explain how the language can be used to capture a variety of resources which we will put to use in the remainder of the paper.
We present individual resource games formally in Section~\ref{sec:irg-dp}. We also introduce precisely the decision problems {\sf NASH EQUILIBRIUM}, {\sf RATIONAL ELIMINATION}, and {\sf RATIONAL CONSTRUCTION}. 
We will use and study two kinds of preferences over action profiles. 
We define dichotomous preferences in Section~\ref{sec:dicho-pref}. We study all three decision problems. We propose general algorithms and general complexity results depending on the complexity of sequent provability in \logica, and on whether \logica admits the weakening rule or not (that is, whether \logica is linear or affine). 
We do the same for parsimonious preferences in Section~\ref{sec:parsi-pref}. 
We also illustrate the decision problems with a few small examples.
We present more thorough examples in Section~\ref{sec:examples}. In particular, we formalize Example~\ref{ex:divorce} in Section~\ref{sec:divorce}, and a variant Example~\ref{example:telecom1} in Section~\ref{sec:inter-econ}, and we illustrate the findings of this paper on them.
Some concluding remarks are offered in Section~\ref{sec:conclusions}.

\medskip

We provide a technical appendix. Specifically, Appendix~\ref{sec:rules} presents the sequent rules of the biggest fragment of Linear Logic used in the paper.
Appendix~\ref{sec:complexity} briefly summarizes some elements of computational complexity that can be useful to the reader.

\section{Resources and Linear Logic}\label{sec:resources-ll}
One contribution of this paper is to show that resource-sensitive logics are a useful tool for studying the formal aspects of resources in game theoretical settings.
Another contribution is to demonstrate that it is possible to obtain rather general results for a large class of games of resources depending on the formal properties of the logic \logica we start with.
This offers the opportunity to tailor a game to the needs of a certain application without changing the framework.
We can indeed choose any sensible fragment of a resource-sensitive logic.

We will work with some fragments of Linear Logic~\cite{girard1987}.
The conceptual aspects of the paper can be grasped without a great understanding of Linear Logic, but the technical results will draw upon the proof theory and its rules presented in the Appendix~\ref{sec:rules}. A basic understanding of logic is thus necessary to follow the proofs in general, and some intuitions about the resource interpretation of Linear Logic can hopefully contribute to make reading through the remainder of this paper less dull.

\subsection{Formulas and sequents}
A good introduction to Linear Logic and its variants is
\cite{Troelstra1992}. We will use logics defined on the language of
propositional Linear Logic. The classical tautology splits into the
additive~$\top$ and the multiplicative~$\mathbf{1}$. The classical
falsum splits into the additive~$\mathbf{0}$ and the multiplicative~$\bot$. The additive conjunction and disjunction are respectively~$\with$ and~$\oplus$. The multiplicative conjunction and disjunction are
respectively~$\parr$ and~$\otimes$. The linear implication is $A \multimap B$ and combines with the multiplicative conjunction such that $(A
\otimes (A \multimap B)) \multimap B$ is a valid principle. The linear
negation is $\lineg A$.

MLL is the multiplicative fragment, whose language is formalized by the grammar 
$A ::= \mathbf{1} | \bot | p | \lineg A |\linebreak A \parr A | 
                                A \otimes A | A \multimap A$, 
where $p$ is an atomic formula. It only contains the multiplicative connectives.
MALL is the fragment with both
additive and multiplicative operators
$A ::= \top | \mathbf{0} | \mathbf{1} | \bot | p | 
\lineg A |
A \parr A | A \otimes A | A \multimap A | 
A \with A | A \oplus A$.

\medskip

We now introduce some terminology and notations.
A \emph{sequent} is a statement $ \Gamma \vdash \Delta$ where $\Gamma$ and $\Delta$ are finite multisets of occurrences of formulas of $\logica$. Often, we can conveniently write a multiset $\{A_1, \ldots, A_n\}$ as the list of formulas $A_1, \ldots, A_n$. 
Also, we use the notation $\Gamma^* = \bigotimes_{A \in \Gamma} A$ and $\emptyset^* = \mathbf{1}$. 
An \emph{intuitionistic sequent} is a sequent $\Gamma \vdash A$ with only one formula to the right.
Sequent provability will play an important part in the technical work of the paper.
A sequent $\Gamma \vdash \Delta$ is \emph{provable} in \logica if
there exists a linear proof using the rules of the logic \logica.
Intuitively, $\Gamma \vdash \Delta$ being provable means that the resources in $\Gamma$ can be transformed into either of the resources in $\Delta$.
If a sequent $\Gamma \vdash \Delta$ is not provable, we can write $\Gamma \not\vdash \Delta$, although we will also often simply write ``not $\Gamma \vdash \Delta$''. Section~\ref{sec:seq-prov-def-compl-summary} summarizes the computational complexity characterizations of a few fragments of Linear Logic in terms sequent provability.

In the individual resource games introduced in this paper, the action of a player~$i$ consists in making available a multiset $C_i$ of formulas/resources. The outcome of an action is the multiset union of all the individual actions: $\Gamma = \biguplus_i C_i$.\footnote{We use $\biguplus$ for the multiset union, and $\bigcup$ for the set union.} The goal of a player is a formula/resource $\gamma$. To decide whether the profile with outcome $\Gamma$ satisfies the goal $\gamma$ of a player, we will evaluate the provability of the (intuitionistic) sequent $\Gamma \vdash \gamma$.
%

%
%The logic is \emph{affine} when it admits the structural rule of weakening $(W)$. We shall see that the weakening rule has striking consequences on the algorithmic solutions to problems we solve in this paper.
The logic captured by all the rules in the Appendix~\ref{sec:rules} is Affine MALL.

A rule that is \emph{not} part of the calculus is the structural rule of contraction. One rule of contraction (left contraction) says that if something can be proved with two occurrences of $A$, then it can be proved with only one occurrence. Symbolically,
\[
\AxiomC{$\Gamma, A,A \vdash \Delta$} \UnaryInfC{$\Gamma, A \vdash \Delta$} \DisplayProof \enspace .
\]
This is prohibited in every resource-sensitive logic. Integrating it into Linear Logic, one consequence would be that $A \vdash A \otimes A$. If we interpret formulas as resources---as we do---contraction would be a license to duplicate resources at will. (See~\cite{restall94thesis} for a detailed account of logics without contraction.)

\medskip

\begin{table}[t]
\centering
\bgroup
\def\arraystretch{1.2}
\begin{tabular}{|lcl|}\hline
  \cellcolor[gray]{\transparency}{$\lineg \lineg A$}  & \cellcolor[gray]{\transparency}{$\dashv\vdash$} & \cellcolor[gray]{\transparency}{$A$}\\  
  $\lineg (A \with B)$ & $\dashv\vdash$ & $(\lineg A) \oplus (\lineg B)$\\  
  \cellcolor[gray]{\transparency}{$A \parr B$} & \cellcolor[gray]{\transparency}{$\dashv\vdash$} & \cellcolor[gray]{\transparency}{$(\lineg A) \multimap B$}\\
  $\lineg (A \otimes B)$ & $\dashv\vdash$ & $(\lineg A) \parr (\lineg B)$\\
  \cellcolor[gray]{\transparency}{$A \parr \bot$} & \cellcolor[gray]{\transparency}{$\dashv\vdash$} & \cellcolor[gray]{\transparency}{$A$}\\  
  $A \otimes \mathbf{1}$ & $\dashv\vdash$ & $A$\\
  \cellcolor[gray]{\transparency}{$A \with \top$} & \cellcolor[gray]{\transparency}{$\dashv\vdash$} & \cellcolor[gray]{\transparency}{$A$}\\
  $A \oplus \mathbf{0}$ & $\dashv\vdash$ & $A$\\  
  \cellcolor[gray]{\transparency}{$\mathbf{0}$} & \cellcolor[gray]{\transparency}{$\dashv\vdash$} & \cellcolor[gray]{\transparency}{$\lineg \top$}\\
  $\mathbf{\bot}$ & $\dashv\vdash$ & $\lineg \mathbf{1}$\\\hline
\end{tabular}
\egroup
\caption{\label{tab:equiv-ll}  Remarkable relationship between the Linear Logic connectives. The symbol $\dashv\vdash$ indicates provability in both directions.}
\end{table}

We must concede that some of the connectives of MLL and MALL do not have an intuitive interpretation in terms of resources, in and of themselves. This is the case of the multiplicative and the additive falsums ($\bot$, $\mathbf{0}$), and of the somehow infamous multiplicative disjunction $\parr$. Fortunately, we do not need them to enjoy the full expressivity of Linear Logic. To see that, Table~\ref{tab:equiv-ll} shows how the connectives interact. From it, it is clear that we can as well make without some language redundancy. The resource-interpretable language of MLL is
\[A ::= \mathbf{1} \mid p \mid \lineg A \mid A \otimes A \mid A \multimap A \enspace ,\]
and the resource-interpretable language of MALL is 
\[A ::= \top \mid \mathbf{1} \mid p \mid 
\lineg A \mid A \otimes A \mid A \multimap A \mid 
A \with A \mid A \oplus A \enspace .\]
It suffices to see the other connectives as definitions, following the equivalences of Table~\ref{tab:equiv-ll}. We define $\bot = \lineg \mathbf{1}$, $\mathbf{0} = \lineg \top$, and $A \parr B = (\lineg A) \multimap B$.

\subsection{Resources as propositions} 
A resource captured by a proposition of Linear Logic, can be atomic
like one mole of hydrogen $\var{H}$ or one mole of oxygen $\var{O}$. It can be a
simultaneous combination of resources, e.g., $\var{O} \otimes \var{O}$ being two moles
of oxygen. A resource can be a process transforming resources, e.g.,
$\var{H_2O} \otimes \var{H_2O} \multimap \var{H_2} \otimes \var{H_2} \otimes O_2$ would be the
well known chemical reaction of electrolysis. It consumes two moles of
water to produce two moles of dihydrogen and one mole of dioxygen.
Working harmoniously with resources and resource transformation
processes with this meticulous control over their combination is made
possible using resource-sensitive logics. In a game where a player is
endowed with $2q$ moles of water and a player is endowed with $q$
processes of electrolysis, it is possible to consume these resources
and produce $2q$ moles of hydrogen gas and $q$ of oxygen gas. But not
more!

In Section~\ref{sec:alan-fish}, we will illustrate our games with an
example using chemical reactions.
But for the time being, we explain in more details how the refined operators of Linear Logic can be used to formalize and grasp a variety of resources.
Table~\ref{tab:readings-formulas} reports possible readings of the connectives.

\begin{table}[ht]
\centering
\bgroup
\def\arraystretch{1.2}
  \begin{tabular}{|ll|}\hline
    \cellcolor[gray]{\transparency}{$A \otimes B$} & \cellcolor[gray]{\transparency}{$A$ and $B$ simultaneously}\\
    $A \with B$ & a deterministic choice between $A$ and $B$; not both\\
    \cellcolor[gray]{\transparency}{$A \oplus B$} & \cellcolor[gray]{\transparency}{$A$ or $B$ non-deterministically; not both}\\
    $A \multimap B$ & $A$ is sufficient to produce $B$ (losing $A$ in the process)\\
    \cellcolor[gray]{\transparency}{$\mathbf{1}$} & \cellcolor[gray]{\transparency}{vacuous resource}\\
    $\top$ & some resource\\\hline
  \end{tabular}
  \egroup
  \caption{\label{tab:readings-formulas} Possible resource interpretations of formulas.}
\end{table}
Now, whether the occurrence of a resource indicates a consumption or a production of the resources depends on where a formula appears in the sequent. The sequent of Linear Logic \[A \vdash B\] can be read as
\[\text{``if you give $A$ you can receive $B$''} \enspace.\]
Hence, as it should be expected, we give the resources at the left of the sequent, and receive the resources at the right of the sequent. Table~\ref{tab:readings-sequents} reports possible readings of the sequents. 
\begin{table}[ht]
\centering
\bgroup
\def\arraystretch{1.2}
\begin{tabular}{|lp{0.7\linewidth}|}\hline
    \cellcolor[gray]{\transparency}{$\Gamma \vdash A \otimes B$} & \cellcolor[gray]{\transparency}{receive $A$ and $B$ simultaneously}\\
    $\Gamma \vdash A \with B$ & choose whether to receive $A$ or $B$; you can't receive both\\
    \cellcolor[gray]{\transparency}{$\Gamma \vdash A \oplus B$} & \cellcolor[gray]{\transparency}{receive $A$ or $B$; you don't choose; you won't receive both}\\
    $\Gamma \vdash A \multimap B$ & receive a resource that can be used in such a way that, if you give $A$, you receive $B$ (losing $A$ in the process)\\
    \cellcolor[gray]{\transparency}{$A \otimes B \vdash \Delta$} & \cellcolor[gray]{\transparency}{give $A$ and $B$ simultaneously}\\
    $A \with B \vdash \Delta$ & choose whether to give $A$ or $B$; you don't give both\\
    \cellcolor[gray]{\transparency}{$A \oplus B \vdash \Delta$} & \cellcolor[gray]{\transparency}{give $A$ or $B$; you don't choose; you don't give both}\\
    $A \multimap B \vdash \Delta$ & give a resource that can be used in such a way that, if you give $A$, you receive $B$ (losing $A$ in the process)\\\hline
\end{tabular}
\egroup
  \caption{\label{tab:readings-sequents} Possible resource interpretations of sequents.}
\end{table}
The linear negation allows one to switch the give/receive mode. The sequent $A \vdash \lineg B$ represents ``give $A$ and $B$, and receive nothing''. The sequent $A, \lineg B \vdash \bot$ represents ``give $A$ and receive $B$''.

\begin{example}
A few items can be obtained from vending machine in exchange of money. For instance, giving $\var{\$1}$ you can choose to receive a chocolate bar or a soft-drink. This is captured by \[\var{\$1} \vdash \var{chocobar} \with \var{drink} \enspace .\] Also, giving $\var{\$0.8}$ you can receive $2$ packs of gum. This is captured by:
\[\$0.8 \vdash \var{gum} \otimes \var{gum}\enspace .\]
\end{example}
In the previous example, the formula $\var{chocobar} \with \var{drink}$ denotes a deliberative choice between $\var{chocobar}$ and $\var{drink}$. One \emph{and} the other can be obtained from $\var{\$1}$, but not both. This is significantly different from $\var{\$1} \vdash \var{chocobar} \oplus \var{drink}$ which denotes something more akin to the classical disjunction: $\var{chocobar}$ or $\var{drink}$ can be obtained from $\var{\$1}$. But for all we know, it might be impossible to actually get one or to get the other, and we don't get to decide.
\begin{example}
We can represent a simple act of gambling. The sequent
\[  \var{\$1} \vdash (\var{\$1}\otimes \var{\$1}) \oplus \mathbf{1} \] 
captures the fact that you can give $\var{\$1}$ to receive $\$2$ or nothing (the vacuous resource); but you don't choose what you get.
\end{example}
The next example uses most of the resource-interpretable connectives.
\begin{example}
We can capture the fact that $\var{\$17}$ get you a menu:
\[\var{\$17} \vdash \var{menu} \enspace .\]
The menu consists of a main dish, a side dish, and a dessert:
\[\var{menu} \vdash \var{dish} \otimes \var{side} \otimes \var{dessert} \enspace .\]
As main dish, you can choose between fish and meat:
\[
\var{dish} \vdash \var{fish} \with \var{meat} \enspace .
\]
The side dish depends on the season; you don't choose; it is either aubergine, or parsnip with leek, or asparagus:
\[\var{side} \vdash \var{aubergine} \oplus (\var{parsnip} \otimes \var{leek}) \oplus \var{asparagus} \enspace .\]
Finally, as dessert, you choose between the strudel and the chocolate tart. Moreover, you choose whether to have ice cream for $\var{\$1}$ extra, or to have no extra (the vacuous resource).
\[\var{dessert} \vdash (\var{strudel} \with \var{chocotart}) \otimes ((\var{\$1} \multimap \var{icecream}) \with \mathbf{1}) \enspace .
\]
\end{example}
We have not illustrated the additive unit $\top$ yet. The next example hints at the upcoming formalization of Example~\ref{ex:divorce} in Section~\ref{sec:divorce}.
\begin{example}
  We can formalize the function of the whole baking equipment (mixer, oven, etc) as the resource transformation process $\var{flour} \multimap \var{bread}$. That is, the equipment transforms flour into bread. (Arguably ignoring that we would also need water and electricity. For simplicity, water and electricity could here be considered resources that are provably equivalent to the vacuous resource $\mathbf{1}$.)
  The sequent
  \[\var{flour}, \var{flour}, \var{flour} \multimap \var{bread} \vdash \var{bread} \otimes \top\]
  indicates that with two `tokens' of flour and the breakmaking equipment, one can make bread, and \emph{some} resources will remain in excess, viz., $\var{flour}$.
\end{example}
The additive unit $\top$ has some connection with the relationship between linear and affine reasoning that we now discuss briefly.

\subsection{Linear vs.\ affine reasoning and preferences}
Weakening (rules $(W)$ in the Appendix~\ref{sec:rules}) in the logic \logica can play a crucial role in the satisfaction of the goals of the players. It will also have striking consequences for the algorithmic solutions of the decision problems that we study in this paper.

In the context of resource-sensitive logic, one rule of weakening (left weakening) says that if something can be obtained from a set of resources then it can also be obtained from more resources. Symbolically,
\[
\AxiomC{$\Gamma \vdash \Delta$}
\UnaryInfC{$\Gamma, A \vdash \Delta$} \DisplayProof
\]
Weakening gives a monotonic flavor to the process of deduction in the logic.
%Weakening says that if something is deducible in a situation $\Gamma$, it will be deducible in every superset of $\Gamma$.
Following the terminology in Linear Logic, a logic \logica admitting weakening will be referred to as \emph{affine} and a logic \logica without weakening will just be referred to as \emph{linear}.

Despite the fact the Affine Logic admits more inference rules than Linear Logic, the unit $\top$ allows one to simulate the reasoning in Affine Logic with the provability of Linear Logic. Indeed, the sequent $\Gamma \vdash A$ is provable in a logic \logica with the rule of weakening iff the sequent $\Gamma \vdash A \otimes \top$ is provable the logic \logica without using weakening.

In the affine case, $A, B \vdash A$ is a provable sequent. If a player has a goal $\gamma = A$, then she will find her objective satisfied with an outcome
$\{A,B\}$. In the linear case, we have in general $A, B \not\vdash A$
(unless $B$ is a vacuous resource equivalent to $\mathbf{1}$). A player with a goal $\gamma = A$ will not be satisfied with an outcome $\{A,B\}$ as she wants $A$ and nothing more. If she \emph{is} indeed indifferent to leftover resources, her goal can be expressed as $\gamma = A \otimes \top$, when \logica is linear.

Affine logic should be used when extra resources can be disposed of freely. That is, when we can assume that a player satisfied with an outcome would be satisfied with a more sizeable outcome. As we will see in Section~\ref{sec:parsi-pref}, this does not prevent players to behave more parsimoniously when they can.

\subsection{Sequent provability and some complexity characterizations}\label{sec:seq-prov-def-compl-summary}
Given a sequent in a fragment \logica of Linear Logic, the problem of \emph{sequent provability} (or \emph{provability} for short) asks whether the sequent is provable from the sequent rules for \logica. When convenient, we write ``\logica is in $\mathsf{C}$'' when the problem of sequent provability in the logic \logica is in the complexity class $\mathsf{C}$.

Before moving to the technical part of this paper,
we quickly summarize the complexity of sequent provability in some
fragments and variants of Linear Logic that could be used as the
\logica parameter in our analysis resource games.\footnote{See Appendix~\ref{sec:complexity} for some elements of complexity that will be useful in the proofs in this paper.}
The results of this
paper will be applicable to every fragment mentioned here. MALL is
$\mathsf{PSPACE}$-complete; MLL is $\mathsf{NP}$-complete; Affine MLL
is $\mathsf{NP}$-complete; Affine MALL is $\mathsf{PSPACE}$-complete;
Intuitionistic MALL is $\mathsf{PSPACE}$-complete; Intuitionistic MLL
is $\mathsf{NP}$-complete. %See \cite{LMSS93,K94}.
Remarkably, and unlike classical logic, these fragments of Linear
Logic behave well computationally also in the first-order
case. First-Order MLL is $\mathsf{NP}$-complete and First-Order MALL
is $\mathsf{NEXPTIME}$-complete. See \cite{LMSS93,K94}.

We will also consider the weaker fragment that we call MULT:
\[A ::= \mathbf{1} \mid p \mid A \otimes A \enspace .\]
\begin{proposition}\label{prop:mult-easy}
  Sequent provability in Intuitionistic Affine and Intuitionistic Linear MULT is in $\mathsf{PTIME}$.
  \begin{proof}
    Linear MULT is captured by the rules (ax), (cut), (E), ($\otimes$R), ($\otimes$L), ($\mathbf{1}$L), and ($\mathbf{1}$R). Affine MULT also requires $(W)$.
    To check whether the Intuitionistic sequent $\Gamma \vdash A$ is provable, it suffices to check whether $\flat^\bullet(\Gamma) \supseteq \flat(A)$ in the case of Affine MULT or $\flat^\bullet(\Gamma) = \flat(A)$ in the case of Linear MULT, where the flattening functions $\flat$ and $\flat^\bullet$ are defined as follows:%\footnote{We use the symbol $\uplus$ as the multiset union.}
    \begin{itemize}
    \item $\flat(\mathbf{1}) = \emptyset$
    \item $\flat(p) = \{p\}$
    \item $\flat(A\otimes B) = \flat(A) \uplus \flat(B)$
    \item $\flat^\bullet(\emptyset) = \emptyset$
    \item $\flat^\bullet(\{A\} \uplus \Delta) = \flat(A) \uplus \flat^\bullet(\Delta)$
    \end{itemize}
    Both multiset inclusion and multiset equality can be performed in linear time in the number of elements in the sets.
  \end{proof}
\end{proposition}

\section{Individual resource games and decision problems}
\label{sec:irg-dp}

We formally define our models of individual resource games.\footnote{Individual resource games were called 
  \emph{ideal} resource games in~\cite{Tr16ijcai}.} 
\begin{definition}
An \emph{individual resource game} (IRG) is a tuple $G = (N,\gamma_1, \ldots, \gamma_n, \epsilon_1, \ldots, \epsilon_n)$ where:
\begin{itemize}
\item $N = \{1, \ldots, n\}$ is a finite set of players;
\item $\gamma_i$ is a formula of \logica ($i$'s goal, or objective);
\item $\epsilon_i$ is a finite multiset of formulas of \logica ($i$'s endowment).
\end{itemize}
\end{definition}
Let $G = (N,\gamma_1, \ldots, \gamma_n, \epsilon_1, \ldots,
\epsilon_n)$, we define: the set of possible actions of $i$ as the set
of multisets $\ch_i(G) = \{C \mid C \subseteq \epsilon_i \}$, and the set of \emph{profiles} in $G$ as
$\ch(G) = \prod_{i \in N} \ch_i(G)$. When $P= (C_1, \ldots , C_k) \in
\ch(G)$ and $1 \leq i \leq k$, then $P_{-i} = (C_1, \ldots ,
C_{i-1},C_{i+1}, \ldots, C_k)$. That is, $P_{-i}$ denotes $P$ without
player $i$'s contribution. The \emph{outcome} of a profile $P = (C_1, \ldots , C_n)$
is given by the multiset of resources $\out(P) = \biguplus_{1\leq i
  \leq n} C_i$.

We will define ``$i$ strongly prefers $P$ over $P'$'' in due time, reflecting
dichotomous preferences first (Section~\ref{sec:dicho-pref}) and
parsimonious preferences second (Section~\ref{sec:parsi-pref}).
\begin{definition}Let $G = (N,\gamma_1, \ldots, \gamma_n, \epsilon_1, \ldots, \epsilon_n)$. A profile $P \in \ch(G)$ is a \emph{Nash equilibrium} iff for all $i \in N$ and for all $C_i \in \ch_i(G)$, we have that $i$ does not strongly prefer $(P_{-i}, C_i)$ over $P$. 
\end{definition}
Let us note $NE(G)$ the set of Nash equilibria in $\ch(G)$.

\medskip
A basic decision problem is the one of determining whether a choice profile is a Nash equilibrium.

\noindent\begin{minipage}{\linewidth}
\medskip
\hrule height 1pt
\smallskip
\textsf{NASH EQUILIBRIUM (NE)}
\smallskip
\hrule
\begin{description}
\item[{\bf (in)}] An individual resource game $G$ and $P \in \ch(G)$.
\item[{\bf (out)}] $P \in NE(G)$?
\end{description}
%\smallskip
\hrule% height 1pt
\medskip
\end{minipage}

Some profiles that are not equilibria can have desirable
outcomes. Some equilibria can have outcomes that are
undesirable. Hence, it is interesting to investigate how resource
distribution schemes influence how undesirable game equilibria can be
eliminated and how desirable game equilibria can be constructed.

In the tradition of social mechanism design, redistribution schemes
can be used by a central authority to enforce some behavior, either
by disincentivizing a behavior or incentivizing a behavior.

\medskip
We will study redistribution schemes in individual resource games. Let $\epsilon$ be an endowment function such that for every player $i$ we have $\epsilon(i) = \epsilon_i$, a multiset of formulas of \logica. A redistribution scheme of $\epsilon$ is an endowment function $\epsilon'$ such that \[\biguplus_{i \in N} \epsilon(i) = \biguplus_{i \in N} \epsilon'(i) \enspace .\]
We note $\redis(\epsilon)$ the set of redistributions of the endowment function $\epsilon$.

Given the individual resource game $G^{\epsilon} = (N, \gamma_1, \ldots, \gamma_n, \epsilon(1),
\ldots, \epsilon(n))$ we can apply a redistribution scheme where we modify the endowment function $\epsilon$ into $\epsilon'$. We thus obtain the individual resource game 
$G^{\epsilon'} = (N, \gamma_1, \ldots, \gamma_n, \epsilon'(1), \ldots, \epsilon'(n))$.

We will investigate two decision problems inspired
by~\cite{harrenstein2015aamas}, which are related to resource
redistributions. We will look at whether the outcome of a resource
game can be rationally eliminated. That is whether there is a resource
redistribution such that no Nash equilibrium of the new resource game
yields this outcome.

\noindent\begin{minipage}{\linewidth}
\medskip
\hrule height 1pt
\smallskip
\textsf{RATIONAL ELIMINATION (RE)}
\smallskip
\hrule
\begin{description}
\item[{\bf (in)}] An individual resource game $G^\epsilon$ and $P \in \ch(G^\epsilon)$.
\item[{\bf (out)}] Is there a redistribution $\epsilon'$ of $\epsilon$ such that
  for all $P' \in \ch(G^{\epsilon'})$, if $\out(P') = \out(P)$ then
  $P' \not\in NE(G^{\epsilon'})$?
\end{description}
%\smallskip
\hrule% height 1pt
\medskip
\end{minipage}

Conversely, we will look at whether the outcome of a resource game can be rationally constructed. That is whether there is a resource redistribution such that the outcome is the outcome of some Nash equilibrium in the new resource game.

\noindent\begin{minipage}{\linewidth}
\medskip
\hrule height 1pt
\smallskip
\noindent\textsf{RATIONAL CONSTRUCTION (RC)}
\smallskip
\hrule
\begin{description}
\item[{\bf (in)}] An individual resource game $G^\epsilon$ and $P \in \ch(G^\epsilon)$.
\item[{\bf (out)}] Is there a redistribution $\epsilon'$ of $\epsilon$ such that
  there is $P' \in \ch(G^{\epsilon'})$ where $\out(P') = \out(P)$ and $P' \in
  NE(G^{\epsilon'})$?
\end{description}
%\smallskip
\hrule
\medskip
\end{minipage}

Note that being a game equilibrium is without ambiguity a property of
profile. However, after a redistribution of resources in an individual
resource game, the space of actions and the space of profiles
change. Thus, elimination and construction are more about the
\emph{outcomes} of profiles. Section~\ref{sec:elim:illustr} and
Section~\ref{sec:illustr-pasimony} will illustrate these decision
problems in due time.

%%% DICHO PREFERENCE
\section{Dichotomous preferences}
\label{sec:dicho-pref}

Let $G = (N, \gamma_1, \ldots, \gamma_n, \epsilon_1, \ldots,
\epsilon_n)$ be an individual resource game.
Player~$i$, whose goal is $\gamma_i$, realizes her objectives in a profile $P$ when $\out(P) \vdash \gamma_i$. That is, the resources in $\out(P)$ can be transformed into a shareable resource $\gamma_i$.
For $P \in \ch(G)$ and $Q \in
\ch(G)$, we say that player $i \in N$ (dichotomously) strongly prefers
$P$ over $Q$ (noted $Q \prec_i P$) iff $\out(P) \vdash \gamma_i$ and
not $\out(Q) \vdash \gamma_i$.

%% (I.e., $\out((P_{-i}, C_i)) \vdash \gamma_i \text{ implies } \out(P) \vdash \gamma_i$.)

%% OLD prop:compl-dicho-pref WITHOUT HARDNESS
%% \begin{proposition}\label{prop:compl-dicho-pref}
%% Let $G = (N, \gamma_1, \ldots, \gamma_n, \epsilon_1, \ldots, \epsilon_n)$ be an individual resource game, two profiles $P \in \ch(G)$ and $Q \in
%% \ch(G)$, and a player $i \in N$. When sequent validity in \logica is in $\mathsf{NP}$, the statement $Q \prec_i P$ is a $\mathsf{NP \land coNP} = \mathsf{B\var{H_2}}$ predicate. When sequent validity in \logica is in $\mathsf{PSPACE}$, the statement $Q \prec_i P$ is a $\mathsf{PSPACE}$ predicate.

%% \begin{proof}
%% The corresponding language is $L = \{ (P,Q) \mid Q \prec_i P \} = L_1
%% \cap L_2$ with $L_1 = \{ (P,Q) \mid \out(P) \vdash \gamma_i \}$, and
%% $L_2 = \{ (P,Q) \mid \text{ not } \out(Q) \vdash \gamma_i\}$. In
%% particular, when \logica is in $\mathsf{NP}$, we clearly have that
%% $L_1$ is a $\mathsf{NP}$ language and $L_2$ is a $\mathsf{coNP}$ language.
%% \end{proof}
%% \end{proposition}

\begin{proposition}\label{prop:compl-dicho-pref}
Let $G = (N, \gamma_1, \ldots, \gamma_n, \epsilon_1, \ldots, \epsilon_n)$ be an individual resource game, two profiles $P \in \ch(G)$ and $Q \in
\ch(G)$, and a player $i \in N$.
The problem of deciding whether $Q \prec_i P$ is: in $\mathsf{PTIME}$ when provability in \logica is in $\mathsf{PTIME}$.
It is $\mathsf{NP \land coNP} = \mathsf{B\var{H_2}}$-complete when provability in \logica is $\mathsf{NP}$-complete.
It is $\mathsf{PSPACE}$-complete when provability in \logica is $\mathsf{PSPACE}$-complete.
\begin{proof}
The language corresponding to the problem is $L = \{ (P,Q) \mid Q \prec_i P \} = L_1 \cap L_2$ with $L_1 = \{ (P,Q) \mid \out(P) \vdash \gamma_i \}$, and $L_2 = \{ (P,Q) \mid \text{ not } \out(Q) \vdash \gamma_i\}$.
In particular, when the problem of provability in \logica is in $\mathsf{NP}$, we clearly have that $L_1$ is a $\mathsf{NP}$ language and $L_2$ is a $\mathsf{coNP}$ language.

For hardness, we consider a newly fabricated decision problem that we call \textsf{PROV-NONPROV}. The problem \textsf{PROV-NONPROV} takes in input two sequents of \logica $\Gamma_1 \vdash \Delta_1$ and $\Gamma_2 \vdash \Delta_2$, and outputs true iff $\Gamma_1 \vdash \Delta_1$ is provable and $\Gamma_2 \vdash \Delta_2$ is not provable. It is easy to see that if \logica is $\mathsf{NP}$-complete, then \textsf{PROV-NONPROV} is $\mathsf{B\var{H_2}}$-complete, and if \logica is $\mathsf{PSPACE}$-complete, then \textsf{PROV-NONPROV} is $\mathsf{PSPACE}$-complete.

We propose a reduction of \textsf{PROV-NONPROV} into the problem of deciding whether in an individual resource game, a profile is strongly preferred to another profile by a player.

Let $\Gamma_1 \vdash \Delta_1$ and $\Gamma_2 \vdash \Delta_2$ be two sequents of \logica.
We can prove using $\bot L$, $\bot R$, (cut) and (E) that $\Gamma \vdash \Delta$ iff $\Gamma \vdash \Delta, \bot$.
Thus, we have $\Gamma_1 \vdash \Delta_1$ iff $\Gamma_1, \lineg \Delta_1 \vdash \bot$, and we have $\Gamma_2 \vdash \Delta_2$ iff $\Gamma_2, \lineg \Delta_2 \vdash \bot$.\footnote{For $\Gamma = \{A_1, \ldots, A_k\}$ we note $\lineg \Gamma$ the set  $\{\lineg A_1, \ldots, \lineg A_k\}$.}
Now we construct the game $G = (\{1\}, \gamma_1 = \bot, \epsilon_1 = \Gamma_1 \uplus \lineg \Delta_1 \uplus \Gamma_2 \uplus \lineg \Delta_2)$.
It is now easy to see that \textsf{PROV-NONPROV} instantiated with $\Gamma_1 \vdash \Delta_1$ and $\Gamma_2 \vdash \Delta_2$ returns true iff Player~$1$ strongly prefers $(\Gamma_1 \uplus \lineg\Delta_1)$ over $(\Gamma_2 \uplus \lineg\Delta_2)$ in $G$.
\end{proof}
\end{proposition}

\subsection{Finding Nash equilibria}

We study the complexity of the problem {\sf NASH EQUILIBRIUM} with dichotomous preferences.

\subsubsection{Hardness}

We are about to prove the lower bound of the complexity {\sf NE} with dichotomous preferences. Before we do so, observe that  by applying the rules $L\lineg$ and $R\lineg$,
\[A_1, \ldots, A_n \vdash B_1, \ldots, B_m \text{~~~iff~~~} A_1, \ldots, A_n,  \lineg B_2 , \ldots,  \lineg B_m \vdash B_1\]
is immediate. Hence, we can, without loss of generality, consider only the intuitionistic sequents of \logica in the many-to-one reductions of this paper.

\begin{proposition}\label{prop:iNE-hardness}
\iNE is as hard as the problem of checking sequent provability in \logica, even when there is only one player.
\end{proposition}
\begin{proof}

%% By applying the rules $L\lineg$ and $R\lineg$, it is immediate that
%% \[\gamma_1, \ldots, \gamma_n \vdash \delta_1, \ldots, \delta_m \text{ iff }
%% \gamma_1, \ldots, \gamma_n, \lineg\delta_2 , \ldots, \lineg\delta_m
%% \vdash \delta_1.\] Thus we can w.l.o.g.~consider only the intuitionistic
%% sequents of \logica in the following reduction, 

  We reduce the problem of sequent provability for the logic \logica. W.l.o.g., we only consider intuitionistic sequents.
Let $\Gamma \vdash \delta$ be the intuitionistic sequent where $\Gamma$
is an arbitrary multiset of formulas of \logica and $\delta$ is an
arbitrary formula.

We can construct the individual resource game $G$ such that $G =
(\{1\},\delta, \Gamma \cup \{\delta\})$. $G$ is thus the one-player individual
resource game where Player~$1$'s goal is to achieve $\delta$, and
Player~$1$ is endowed with $\Gamma \cup \{\delta\}$ (this is a set union but we
could have chosen the endowment $\Gamma \uplus \{\delta\}$ as well). A
profile in $G$ is a choice of Player~$1$, that is, a subset $C_1$ of
$\Gamma \cup \{\delta\}$. In this case for any profile $P$ in $G$,
$\out(P) = P$.

We show that $\Gamma \vdash \delta$ iff $\Gamma \in NE(G)$.

From left to right, suppose that $\Gamma \vdash \delta$. We need to
show that $\Gamma \in NE(G)$. That is, for all $C_1 \subseteq \Gamma
\cup \{\delta\}$, if $C_1 \vdash \delta$ then $\Gamma \vdash
\delta$. Since we supposed $\Gamma \vdash \delta$, this is trivially
true.

From right to left, suppose that $\Gamma \in NE(G)$. This means that
for all $C_1 \subseteq \Gamma \cup \{\delta\}$, if $C_1 \vdash \delta$
then $\Gamma \vdash \delta$. Let in particular $C_1 =
\{\delta\}$. Indeed, $C_1 \subseteq \Gamma \cup \{\delta\}$. Moreover,
by (ax) we have $\delta \vdash \delta$. Hence, $\Gamma \vdash \delta$
follows.
\end{proof}

\subsubsection{Algorithms}

\pcase{Linear case.}
To establish an upper bound on the complexity of \iNE let us
first outline an algorithm for solving its complement. That is,
checking whether a profile is \emph{not} a Nash equilibrium. 
Let $P \in \ch(G)$ be a profile. To determine whether $P \not\in NE(G)$,
we can employ a simple non-deterministic algorithm, showed as Algorithm~\ref{GA-co-iNE}.

\begin{algorithm}
\caption{General algorithm for {\sf co-NE}}
\label{GA-co-iNE}
\begin{algorithmic}[1]
\\non-deterministically guess $(i,C'_i) \in N \times \ch_i(G)$.
\\return $P \prec_i (P_{-i},C'_i)$.
\end{algorithmic}
\end{algorithm}

\begin{proposition}\label{prop:NP-Pi2P}\label{prop:ine-pspace-easy}
If the problem of provability in \logica is in $\mathsf{PTIME}$ then \iNE is in $\mathsf{coNP}$.
If the problem of provability in \logica is in $\mathsf{NP}$ then \iNE is in $\mathsf{coNP^{B\var{H_2}}}$ and indeed in $\mathsf{\Pi_2^p}$.
If the problem of provability in \logica is in $\mathsf{PSPACE}$ then \iNE is in $\mathsf{PSPACE}$.
\end{proposition}
\begin{proof}
Consider Algorithm~\ref{GA-co-iNE}. If sequent provability in \logica is in $\mathsf{NP}$, we can check $P \prec_i
(P_{-i},C'_i)$ in $\mathsf{B\var{H_2}}$ (Proposition~\ref{prop:compl-dicho-pref}).
Thus we can check whether $P \not\in NE(G)$ in
$\mathsf{NP^{B\var{H_2}}}$. Finally, we can solve \iNE in
$\mathsf{coNP^{B\var{H_2}}}$.  It is the case that $\mathsf{B\var{H_2}} \subseteq
\mathsf{\Delta_2^p}$, and also that $\mathsf{NP^{\Delta_2^p}} =
\mathsf{\Sigma_2^p}$ so we can solve \iNE in
$\mathsf{\Pi_2^p}$. The proofs for the cases of sequent provability in
$\mathsf{PTIME}$ and
$\mathsf{PSPACE}$ proceed with similar considerations about Algorithm~\ref{GA-co-iNE}.
\end{proof}

\pcase{Affine case.}
%
%% \begin{proposition}\label{prop:monotony}
%% Let an individual resource game $G = (N,\linebreak \gamma_1, \ldots, \gamma_n,
%% \epsilon_1, \ldots, \epsilon_n)$. When \logica is affine, for
%% every profile $P$, and for every $i \in N$: if $\out(P) \vdash
%% \gamma_i$ then $\out((P_{-i}, \epsilon_i) \vdash \gamma_i$.
%% \end{proposition}
%
%% \begin{proof}
%% Suppose $\out(P) \vdash \gamma_i$ and observe that $\out(P) \subseteq
%% \out((P_{-i}, \epsilon_i))$. Thus, $\out((P_{-i}, \epsilon_i) \vdash
%% \gamma_i$ follows by weakening.
%% \end{proof}
%% \begin{proposition}\label{prop:ne-weakening}
%% Let an individual resource game $G = (N,\linebreak \gamma_1, \ldots, \gamma_n,
%% \epsilon_1, \ldots, \epsilon_n)$. When \logica is affine, $P
%% \in NE(G)$ iff 
%% $\forall i \in N: \out((P_{-i}, \epsilon_i)) \vdash \gamma_i 
%% \text{ implies } \out(P) \vdash \gamma_i$.
%% \end{proposition}
Affine logic admits the rule of weakening $(W)$, which allows one to discard resources. 
In this setting, if a player can achieve her goal with the resources $\Gamma$, 
she can as well achieve her goal with the resources $\Gamma \cup \{A\}$.
A consequence is the following lemma, which will have a significant impact 
on the computational complexity of \iNE.
\begin{lemma}\label{lemma:ne-dich-affin}
Let $G = (N,\gamma_1, \ldots, \gamma_n, \epsilon_1, \ldots,
\epsilon_n)$ be an individual resource game. When \logica is affine, $P
\not \in NE(G)$ iff $\exists i \in N: P \prec_i (P_{-i}, \epsilon_i)$.
\end{lemma}

\begin{proof}
Suppose $P \not \in NE(G)$. There is $i \in N$ and $C_i \in \ch_i(G)$
s.t.\ $P \prec_i (P_{-i}, C_i)$. By definition, $\out((P_{-i}, C_i))
\vdash \gamma_i$ and $\out(P) \not\vdash \gamma_i$. We have $C_i
\subseteq \epsilon_i$, so by applying weakening $(W)$ with every
instance of formulas in $\epsilon_i \setminus C_i$, we can prove that
$\out((P_{-i}, \epsilon_i)) \vdash \gamma_i$. We thus have that there
is $i \in N$ s.t.\ $P \prec_i (P_{-i}, \epsilon_i)$. The other way
around is immediate from the definition of Nash equilibria.
%% Suppose that there is $i \in N$ s.t.\ $P \prec_i (P_{-i},
%% \epsilon_i)$. So, there is $i \in N$ and $C_i \in \ch_i(G)$ such that
%% $P \prec_i (P_{-i}, C_i)$. So $P \not\in NE(G)$.
\end{proof}
It means that, in a profile, if no player has an incentive to deviate by making available their whole 
endowment, then the profile is a Nash equilibrium.
The very profile where all the players make available their whole endowment is trivially such a profile.
The next proposition follows immediately:
\begin{proposition}\label{prop:nonempty}
Let $G = (N, \gamma_1, \ldots,
\gamma_n, \epsilon_1, \ldots, \epsilon_n)$ be
an individual resource game. When \logica is affine:
$NE(G) \not= \emptyset$ and $(\epsilon_1, \ldots, \epsilon_n) \in
NE(G)$.
\end{proposition}
Lemma~\ref{lemma:ne-dich-affin} also helps us to establish the following
result.
\begin{proposition}\label{prop:ine-dicho-weakening-easy}
  When \logica is affine, if the problem of sequent provability in
  \logica is in $\mathsf{PTIME}$ then \iNE is in $\mathsf{PTIME}$.
  If the problem of sequent provability in
\logica is in $\mathsf{NP}$ then \iNE is in
$\mathsf{P^{NP||}}$.
If the problem of sequent provability in
\logica is in $\mathsf{PSPACE}$ then \iNE is in
$\mathsf{PSPACE}$.
\end{proposition}
\begin{proof}
  Let $G = (N, \gamma_1, \ldots, \gamma_n, \epsilon_1, \ldots, \epsilon_n)$ be an individual resource game and let $P \in \ch(G)$ be a profile.
  One can check whether $P \in NE(G)$ with Algorithm~\ref{A:iNE-dicho-affine}.

\begin{algorithm}
\caption{Algorithm for \iNE with dichotomous preferences and affine \logica}
\label{A:iNE-dicho-affine}
\begin{algorithmic}[1]
\\for each $i \in N$ do:
\\\indent\indent if ($\out(P) \vdash \gamma_i$):
\\\indent\indent\indent continue;
\\\indent\indent else if ($\out((P_{-i}, \epsilon_i)) \vdash \gamma_i$):
\\\indent\indent\indent return false.% [$P \not\in NE(G)$]
\\return true.% [$P \in NE(G)$]
\end{algorithmic}
\end{algorithm}

For correctness, note that the instructions of the lines $2-4$ are
equivalent to a test of whether $\out(P) \not\vdash \gamma_i$ and
$\out((P_{-i}, \epsilon_i)) \vdash \gamma_i$, that is, $P \prec_i
(P_{-i}, \epsilon_i)$. Lemma~\ref{lemma:ne-dich-affin} ensures that
exactly when there is an $i \in N$ such that $P \prec_i (P_{-i},
\epsilon_i)$ we can conclude that $P$ is not a Nash equilibrium.

Suppose sequent provability in \logica is in $\mathsf{NP}$.
%% The algorithm can be simulated by a deterministic oracle Turing
%% machine in polynomial time with at most $2n$ calls to a $\mathsf{NP}$
%% oracle. $P \in NE(G)$ is thus a $\mathsf{P^{NP[\text{$2n$}]}}$
%% predicate.
The algorithm can be simulated by a deterministic oracle Turing
machine in polynomial time with $2n$ non-adaptive queries to an $\mathsf{NP}$
oracle. Indeed, $P \in NE(G)$ is thus a $\mathsf{P^{NP||[\text{$2n$}]}}$
predicate. The problem is in $\mathsf{P^{NP||}}$.
When sequent provability in \logica is in $\mathsf{PTIME}$ (resp., $\mathsf{PSPACE}$), the algorithm runs in polynomial time (resp., polynomial space).
\end{proof}

\subsection{Elimination}
\label{sec:elim:illustr}

A very simple illustration of \textsf{RATIONAL ELIMINATION} is given by the individual resource game $G^\epsilon = 
 (\{1,2\}, \gamma_1 = B, \gamma_2 = A, \{A\}, \{B\})$. There are two players. Player~$1$ wants $B$ but is endowed with $\{A\}$, while Player~$2$ wants $A$ but is endowed with $\{B\}$. The game $G^\epsilon$ can be represented as on Figure~\ref{fig:elim-Gepsilon}. (We indicate the realized objectives assuming that \logica is affine.)

\begin{figure}[ht]
\begin{center}
\bgroup
\def\arraystretch{1.5}
        \begin{tabular}{|c|*{2}{>{$}C{25mm}<{$}}|}
\hline
           \cellcolor[gray]{\transparency}{\diagbox[linewidth=0.2pt, width=\dimexpr \textwidth/10+2\tabcolsep\relax, height=0.8cm]{$1$}{$2$}}
   & \cellcolor[gray]{\transparency}{\emptyset} & \cellcolor[gray]{\transparency}{\{B\}}\\
\hline
          \cellcolor[gray]{\transparency}{$ \emptyset $} & \emptyset\diNE & \{B\}\diNE: \gamma_1\\
          \cellcolor[gray]{\transparency}{$ \{A\} $} & \{A\}\diNE: \gamma_2 & \{A, B\}\diNE: \gamma_1, \gamma_2\\
\hline
         \end{tabular}
\egroup
\end{center}
\caption{\label{fig:elim-Gepsilon} The game $G^\epsilon$. $\gamma_1$ and $\gamma_2$ indicate that Player~$1$ and Player~$2$ have their goals satisfied, assuming \logica is affine. The symbol ${}\diNE$ denotes a Nash equilibrium.}
\end{figure}

One can readily check that all profiles are Nash equilibria. However,
the profile $(\{A\}, \{B\})$ is more `socially desirable' than the
others since it satisfies both players' goal.

A centralized authority could effectively eliminate the others by
redistributing the resources present in $G^{\epsilon}$ so as to obtain
$G^{\epsilon'} = (\{1,2\}, \gamma_1 = B, \gamma_2 = A, \{B\}, \{A\})$.
The game $G^{\epsilon'}$ can be represented as on Figure~\ref{fig:elim-Gepsilonprime}.

\begin{figure}[ht]
\begin{center}
\bgroup
\def\arraystretch{1.5}
        \begin{tabular}{|c|*{2}{>{$}C{25mm}<{$}}|}
\hline
           \cellcolor[gray]{\transparency}{\diagbox[linewidth=0.2pt, width=\dimexpr \textwidth/10+2\tabcolsep\relax, height=0.8cm]{$1$}{$2$}}
   & \cellcolor[gray]{\transparency}{\emptyset} & \cellcolor[gray]{\transparency}{\{A\}}\\
\hline
          \cellcolor[gray]{\transparency}{$ \emptyset $} & \emptyset & \{A\}: \gamma_2\\
          \cellcolor[gray]{\transparency}{$ \{B\} $} & \{B\}: \gamma_1 & \{A, B\}\diNE: \gamma_1, \gamma_2\\
\hline
         \end{tabular}
\egroup
\end{center}
\caption{\label{fig:elim-Gepsilonprime} The game $G^{\epsilon'}$.}
\end{figure}
The only Nash equilibrium is now the one with outcome $\{B, A\}$.

\subsubsection{Algorithms}

As a consequence of Proposition~\ref{prop:nonempty}, we already know that:
\begin{proposition}
Let $G = (N, \gamma_1, \ldots, \gamma_n, \epsilon_1, \ldots,
\epsilon_n)$ be an individual resource game.  When \logica is affine, the
profile $P$ such that $\out(P) = \biguplus_j \epsilon_j$ is not
rationally eliminable.
\end{proposition}

This is very specific to the affine case (and dichotomous
preferences), and even then, it is of course not true of all Nash
equilibria. To decide whether some outcome is rationally eliminable,
one na\"ive approach consists in trying all possible redistributions
and check whether the outcome is a Nash equilibrium in the resulting
individual resource game. Instead, we are going to exploit a pleasant
property, analogous to~\cite[Corollary~$4$]{harrenstein2015aamas}.

Let $G^\epsilon = 
(N, \gamma_1, \ldots, \gamma_n, \epsilon(1), \ldots, \epsilon(n))$
be an individual resource game. For each player $i \in N$, we define $G^{[\epsilon\arrow i]}$ where $[\epsilon\arrow i]$ is the redistribution of $\epsilon$ where all resources are assigned to $i$, that is:
\[
[\epsilon\arrow i](j) = 
\begin{cases}
\biguplus_{k \in N} \epsilon(k) & \text{when } j = i\\
\emptyset & \text{otherwise.}
\end{cases}\]
Because there is only one active player in $G^{[\epsilon\arrow i]}$, we will sometimes write a profile of $G^{[\epsilon\arrow i]}$ as $(C_i)$ with $C_i \in \ch_i(G^{[\epsilon\arrow i]})$ instead of $(\emptyset, \ldots, \emptyset, C_i, \emptyset, \ldots, \emptyset)$, by abuse of notation.

\begin{lemma}\label{lemma:rat-elimination}
Let $G^\epsilon$ be an individual resource game  and $P \in \ch(G^\epsilon)$. $P$ is rationally eliminable iff there is a player $i \in N$ and a profile $Q \in \ch(G^{[\epsilon\arrow i]})$, such that $\out(Q) = \out(P)$ and $Q \not\in NE(G^{[\epsilon\arrow i]})$.
\end{lemma}
\begin{proof}
From right to left. Suppose $Q \not\in NE(G^{[\epsilon\arrow i]})$ for
some $i \in N$.  Let also $P \in \ch(G^\epsilon)$ be a profile and assume $\out(P) = \out(Q)$. When there is at most one player with a non-empty endowment, as in $[\epsilon\arrow i]$, there is a one-to-one correspondence between the set of profiles and the set of outcomes. Thus, there is one and only one profile in 
$G^{[\epsilon\arrow i]}$ with outcome $\out(P)$ and it is $Q$. So there is a redistribution of $\epsilon$, namely $[\epsilon\arrow i]$, such that for all profiles $Q \in \ch(G^{[\epsilon\arrow i]})$ with outcome $\out(P)$, we have $Q \not\in NE(G^{[\epsilon\arrow i]})$. So $P$ is rationally eliminable.

From left to right. Suppose that $P$ is rationally eliminable. Thus,
there is a redistribution $\epsilon'$ of $\epsilon$ such that for all
$P' \in \ch(G^{\epsilon'})$, if $\out(P') = \out(P)$ then $P' \not\in
NE(G^{\epsilon'})$. So let $R \in \ch(G^{\epsilon'})$ be an arbitrary profile 
with $\out(R) = \out(P)$. By assumption, we have that $R \not\in
NE(G^{\epsilon'})$. By definition of Nash equilibria, this means that there is $i \in N$ and $C_i' \in \ch_i(G^{\epsilon'})$ such that $R \prec_i (R_{-i},C_i')$.
Now consider the game $G^{[\epsilon\arrow i]}$. We have $\out(R) \in \ch_i(G^{[\epsilon\arrow i]})$ and $\out((R_{-i},C_i')) \in \ch_i(G^{[\epsilon\arrow i]})$. 
Let the profile $R^1 \in \ch(G^{[\epsilon\arrow i]})$ with $R^1_i = \out(R)$ and $R^1_j = \emptyset$ when $j \not = i$. Let $R^2 \in \ch(G^{[\epsilon\arrow i]})$ be the profile with $R^2_i = \out((R_{-i},C_i'))$ and $R^2_j = \emptyset$ when $j \not = i$. Since, $R \prec_i (R_{-i},C_i')$, we also have $R^1 \prec_i R^2$. So $R^1 \not\in NE(G^{[\epsilon\arrow i]})$. The profile $R^1$ is the only profile of $G^{[\epsilon\arrow i]}$ with outcome $\out(P)$. So we can conclude.
\end{proof}

%% \begin{lemma}\label{lemma:rat-elimination}
%% Let an individual resource game $G^\epsilon$ and $P \in \ch(G^\epsilon)$. $P$ is rationaly eliminable iff there is $Q \in \ch(G^{[\epsilon\arrow i]})$, $\out(Q) = \out(P)$, and $Q \not\in NE(G^{[\epsilon\arrow i]})$ for some $i \in N$.

%% \begin{proof}
%% From right to left. Suppose $P \not\in NE(G^{[\epsilon\arrow i]})$ for
%% some $i \in N$.  So, there is a redistribution, namely
%% $[\epsilon\arrow i]$, of $\epsilon$ such that $P \not\in
%% NE(G^{[\epsilon\arrow i]})$. By definition, $P$ is rationaly
%% eliminable.

%% From left to right. Suppose that $P$ is rationaly eliminable. Thus,
%% there is a redistribution $\epsilon'$ of $\epsilon$ such that $P
%% \not\in NE(G^{\epsilon'})$. By definition of Nash equilibria, it means
%% that there is $i \in N$ and there is $C'_i \in \ch_i(G^{\epsilon'})$
%% such that $P \prec_i (P_{-i}, C'_i)$. Now consider the game
%% $G^{[\epsilon\arrow i]}$. We have $\out(P) \in \ch_i(G^{[\epsilon\arrow i]})$ and
%% $\out((P_{-i}, C'_i)) \in \ch_i(G^{[\epsilon\arrow i]})$. Since $P \prec_i
%%   (P_{-i}, C'_i)$, we have that $P \not \in NE(G^{[\epsilon\arrow i]})$.
%% \end{proof}
%% \end{lemma}

\pcase{Linear case.} We establish an upper bound on the complexity of {\sf RE} 
when \logica does not admit the weakening rule.

\begin{proposition}\label{prop:easyness-rat-elim-linear}
  When \logica is linear, {\sf RE} is in $\mathsf{NP}$ when provability in \logica is in $\mathsf{PTIME}$, in $\mathsf{NP^{B\var{H_2}}}$ and
  indeed in $\mathsf{\Sigma_2^p}$ when \logica is in $\mathsf{NP}$,
  and in $\mathsf{PSPACE}$ when \logica is in $\mathsf{PSPACE}$.
\end{proposition}
\begin{proof}
Let $P \in \ch(G^\epsilon)$ be a profile. To determine
whether $P$ is rationally eliminable,
we can use Algorithm~\ref{GA-RE}.

\begin{algorithm}
\caption{General algorithm for {\sf RE}}
\label{GA-RE}
\begin{algorithmic}[1]
\\non-deterministically guess $(i,C'_i) \in N \times \ch_i(G^{[\epsilon\arrow i]})$.\\return $P \prec_i (P_{-i},C'_i)$.
\end{algorithmic}
\end{algorithm}
Straightforwardly, it guesses a player $i$ and a deviation in the game
$G^{[\epsilon\arrow i]}$ for Player~$i$ from the profile $(\out(P))
\in \ch(G^{[\epsilon\arrow i]})$, and checks whether Player~$i$ has an
incentive to do this deviation.
%% Straightforwardly, it checks whether a player~$i$ has an incentive to
%% deviate in the game $G^{[\epsilon\arrow i]}$ from a profile $Q \in
%% \ch(G^{[\epsilon\arrow i]})$ with $\out(Q) = \out(P)$.  
By Lemma~\ref{lemma:rat-elimination}, if such a player and deviation
exist and only if they exist, the profile $P$ is rationally eliminable
in $G^\epsilon$. So the algorithm is correct. It can of course be
simulated by a non-deterministic oracle Turing machine with one call
to an oracle for $P \prec_i
(P_{-i},C'_i)$. Proposition~\ref{prop:compl-dicho-pref} informs us of a
containing class of this oracle.
\end{proof}

\pcase{Affine case.}
When \logica admits the weakening rule, we can propose a surprisingly simple algorithm, 
which takes advantage of both Lemma~\ref{lemma:ne-dich-affin} and Lemma~\ref{lemma:rat-elimination}.
\begin{proposition}\label{prop:easyness-rat-elim-affine}
  When \logica is affine, {\sf RE} is in $\mathsf{PTIME}$ when provability in \logica is in $\mathsf{PTIME}$, in
  $\mathsf{P^{NP||}}$
  %$\mathsf{\Delta_2^p}$
  when \logica is in $\mathsf{NP}$, and in $\mathsf{PSPACE}$ when \logica is in $\mathsf{PSPACE}$.
\end{proposition}
\begin{proof}

  Let $G = (N,\gamma_1, \ldots, \gamma_n, \epsilon_1, \ldots, \epsilon_n)$ be an individual resource game and let $P \in \ch(G)$ be a profile.
  Consider Algorithm~\ref{DA-RE}.

\begin{algorithm}
\caption{Algorithm for {\sf RE} with dichotomous preferences and affine \logica}
\label{DA-RE}
\begin{algorithmic}[1]
\\for each $i \in N$ do:
\\\indent\indent if ($P \prec_i ([\epsilon\arrow i](i))$):\label{line:DA-RE-test-prec}
%\\\indent\indent if ($P \not \vdash \gamma_i$ and $[\epsilon\arrow i](i) \vdash \gamma_i$):
\\\indent\indent\indent return true.
\\return false.
\end{algorithmic}
\end{algorithm}
%% It checks whether a player~$i$ has an incentive to deviate in the game $G^{[\epsilon\arrow i]}$ from a profile $Q \in \ch(G^{[\epsilon\arrow i]})$, with $\out(Q) = \out(P)$, to the profile $([\epsilon\arrow i](i)) \in \ch(G^{[\epsilon\arrow i]})$.

%% The algorithm is sound and complete. Indeed, by
%% Lemma~\ref{lemma:ne-dich-affin}, we know that $P \prec_i
%% ([\epsilon\arrow i](i))$ iff $\out(P) \not\in NE(G^{[\epsilon\arrow
%%     i](i)})$. By Lemma~\ref{lemma:rat-elimination}, there is $i \in N$
%% where $\out(P) \not\in NE(G^{[\epsilon\arrow i](i)})$ iff $P$ is
%% eliminable in $G$.

The algorithm is correct.
Indeed, by Lemma~\ref{lemma:rat-elimination}, $P$ is eliminable in $G$ iff there is $i \in N$ where $(\out(P)) \not\in NE(G^{[\epsilon\arrow i]})$.
By Lemma~\ref{lemma:ne-dich-affin}, we know that $(\out(P)) \not\in NE(G^{[\epsilon\arrow i]})$ iff $P \prec_i ([\epsilon\arrow i](i))$.
%It can be simulated by a deterministic oracle Turing machine in polynomial time with at most $n$ calls to an oracle whose complexity is stated in Proposition~\ref{prop:compl-dicho-pref}.
Notice that the test of line~\ref{line:DA-RE-test-prec} is equivalent to $P \not \vdash \gamma_i$ and $[\epsilon\arrow i](i) \vdash \gamma_i$.
Thus, it can be simulated by a deterministic oracle Turing machine in polynomial time with at most $2n$ non-adaptive queries to an oracle for the problem of sequent provability.
When the problem of sequent provability in \logica is in $\mathsf{NP}$ it yields a complexity of $\mathsf{P^{NP||}}$.
%The results follow from the facts that $\mathsf{P}^\mathsf{\Delta_2^p} = \mathsf{\Delta_2^p}$ and $\mathsf{P}^\mathsf{PSPACE} = \mathsf{PSPACE}$.
When it is in $\mathsf{PTIME}$ (resp., $\mathsf{PSPACE}$), it yields a complexity of $\mathsf{PTIME}$ (resp., $\mathsf{PSPACE}$).
\end{proof}

\subsubsection{Hardness}

The linear and affine cases both use the same proof strategy which we
present at once.
\begin{proposition}\label{prop:hard-rat-elim-dicho}
{\sf RE} is as hard as the problem of checking
sequent \emph{non-provability} in \logica.
\end{proposition}
\begin{proof}
Let $\Gamma \vdash \delta$ be an arbitrary intuitionistic sequent.
Let $\phi = \Gamma^* \multimap \delta$. (Remember that $\Gamma^* = \bigotimes_{A \in \Gamma} A$.) Let $G^\epsilon = (\{1,2\}, \phi, \mathbf{1}, \emptyset, \{\phi\})$ be an individual resource game. So, we have $\epsilon_1 = \emptyset$ and $\epsilon_2 = \{\phi\}$. There is only one other distinct redistribution $\epsilon'$ of $\epsilon$ where $\epsilon'_1 = \{\phi\}$ and $\epsilon'_2 = \emptyset$. It is the case that $\redis(\epsilon) = \{\epsilon, \epsilon'\}$. Let $G^{\epsilon'} = (\{1,2\}, \phi, \mathbf{1}, \{\phi\}, \emptyset)$ be the individual resource game resulting from the redistribution $\epsilon'$. Both games are represented on Figure~\ref{fig:hard-redis}.

\begin{figure}[ht]
\centering
\bgroup
\def\arraystretch{1.5}
\begin{subfigure}{.5\textwidth}\centering
\begin{tabular}{|c|*{2}{>{$}C{10mm}<{$}}|}
\hline
\cellcolor[gray]{\transparency}{\diagbox[linewidth=0.2pt, width=\dimexpr \textwidth/10+2\tabcolsep\relax, height=0.8cm]{$1$}{$2$}}
   & \cellcolor[gray]{\transparency}{\emptyset} & \cellcolor[gray]{\transparency}{\{\phi\}}\\
\hline
\cellcolor[gray]{\transparency}{$ \emptyset $} & \emptyset & \{\phi\}\\
\hline
\end{tabular}
\subcaption{\label{fig:hard-redis-Gepsilon} $G^{\epsilon}$.}
\end{subfigure}%
\begin{subfigure}{.5\textwidth}\centering
\begin{tabular}{|c|*{1}{>{$}C{10mm}<{$}}|}
\hline
           \cellcolor[gray]{\transparency}{\diagbox[linewidth=0.2pt, width=\dimexpr \textwidth/10+2\tabcolsep\relax, height=0.8cm]{$1 $}{$2 $}}
   & \cellcolor[gray]{\transparency}{\emptyset}\\
\hline
\cellcolor[gray]{\transparency}{$ \emptyset $} & \emptyset\\
\cellcolor[gray]{\transparency}{$ \{\phi\} $} & \{\phi\}\\
\hline
\end{tabular}
\subcaption{\label{fig:hard-redis-Gepsilon-prime} $G^{\epsilon'}$.}
\end{subfigure}
\egroup
\caption{\label{fig:hard-redis} Games $G^{\epsilon}$ and $G^{\epsilon'}$. The profile $(\emptyset,\emptyset)$ is a Nash equilibrium in $G^{\epsilon}$. The profile $(\emptyset,\emptyset)$ is a Nash equilibrium in $G^{\epsilon'}$ iff $\Gamma \vdash \delta$. (The profile $(\{\phi\}, \emptyset)$ is a Nash equilibrium in $G^{\epsilon'}$. Depending on whether $\Gamma \vdash \delta$ and whether \logica is linear or affine, $(\emptyset, \{\phi\})$ may or may not be Nash equilibria in $G^{\epsilon}$. This is inconsequential for the reduction in the proof of Proposition~\ref{prop:hard-rat-elim-dicho}.)}
\end{figure}

We show that both in the case of linear and of affine logics, we have $\Gamma \not\vdash \delta$ iff $(\emptyset, \emptyset)$ is rationally eliminable in $G^\epsilon$.

We first show that
\begin{equation}\label{eq:deduction}
\Gamma \vdash \delta \text{ iff } \emptyset \vdash \phi \enspace .
\end{equation}
From left to right, suppose $\Gamma \vdash \delta$. By applying ($\otimes$L) enough times we obtain $\Gamma^* \vdash \delta$. Then we obtain $\emptyset \vdash \Gamma^*  \multimap \delta$ using ($\multimap$R). From right to left, suppose $\emptyset \vdash \Gamma^* \multimap \delta$. With (ax) and $\otimes$R we can show $\Gamma \vdash \Gamma^*$. Using $\otimes$R on the sequents $\Gamma \vdash \Gamma^*$ and $\emptyset \vdash \Gamma^* \multimap \delta$ we obtain
\begin{equation}\label{eq:deduction-a}
  \Gamma \vdash \Gamma^* \otimes \Gamma^* \multimap \delta \enspace .
\end{equation}
Without assumption we can also show
\begin{equation}\label{eq:deduction-b}
\Gamma^* \otimes \Gamma^* \multimap \delta \vdash \delta \enspace ,
\end{equation}
using the rules (ax), ($\multimap$L), and ($\otimes$L). We conclude that $\Gamma \vdash \delta$ using (cut) on the sequents~\ref{eq:deduction-a} and~\ref{eq:deduction-b}.

\medskip

We can proceed. Suppose $\Gamma \not\vdash \delta$. We show that $(\emptyset, \emptyset)$ is not a Nash equilibrium in $G^{\epsilon'}$. Since $\Gamma \not\vdash \delta$, we also have $\emptyset \not\vdash \phi$ (by Equation~\ref{eq:deduction}). On the other hand, using (ax), we have $\{\phi\} \vdash \phi$. So in the profile $(\emptyset,\emptyset)$, Player~$1$ has an incentive to deviate to the profile $(\{\phi\},\emptyset)$. So $(\emptyset,\emptyset)$ is not a Nash equilibrium in $G^{\epsilon'}$.

Suppose $\Gamma \vdash \delta$. We show that $(\emptyset, \emptyset)$ is a Nash equilibrium both in $G^{\epsilon}$ and in $G^{\epsilon'}$.

In $G^{\epsilon}$. We have $\emptyset \vdash \mathbf{1}$ from $\textbf{1}$R, so Player~$2$ has no incentive to deviate from the profile $(\emptyset,\emptyset)$ in $G^{\epsilon}$. Moreover, Player~$1$ is dummy in $G^{\epsilon}$. So $(\emptyset,\emptyset)$ is a Nash equilibrium in $G^{\epsilon}$.

In $G^{\epsilon'}$. Since $\Gamma \vdash \delta$, we also have $\emptyset \vdash \phi$ (by Equation~\ref{eq:deduction}), so Player~$1$ has no incentive to deviate from the profile $(\emptyset,\emptyset)$ in $G^{\epsilon'}$. Moreover, Player~$2$ is dummy in $G^{\epsilon'}$. So $(\emptyset,\emptyset)$ is a Nash equilibrium in $G^{\epsilon'}$.

\end{proof}

\subsection{Construction}

For elimination, Lemma~\ref{lemma:rat-elimination} provided a remarkable necessary and sufficient condition for the rational eliminability of a profile. For the rational constructibility of a profile, we can only indicatively provide sufficient conditions. Let $G =  (N,\gamma_1, \ldots, \gamma_n, \epsilon_1, \ldots, \epsilon_n)$ be an IRG, and let $P \in \ch(G)$ be a profile in $G$. If there is a player $i \in N$ such that $\out(P) \vdash \gamma_i$, then $P$ can be rationally constructed by redistributing all the resources to Player~$i$. Also, if there is a player $i \in N$ such that $\biguplus_{k \in N} \epsilon_k \not\vdash \gamma_i \otimes \top$, then $P$ can be rationally constructed by redistributing all the resources to Player~$i$.

We tackle the complexity of {\sf RATIONAL CONSTRUCTION} with dichotomous preferences.

\subsubsection{Hardness}
We prove a lower bound of the problem {\sf RC} in presence of dichotomous preferences.
\begin{proposition}\label{prop:hard-constr-dicho}
{\sf RC} is as hard as the problem of checking
sequent provability in \logica.
\end{proposition}
\begin{proof}
Let $\phi = \Gamma^* \multimap \delta$ and $G = (\{1\}, \phi,
\epsilon_1 = \{\phi\})$. We can see that $(\emptyset) \in NE(G)$ iff
$\emptyset \vdash \phi$, that is $\Gamma \vdash \delta$. As
$\redis(\epsilon) = \{\epsilon\}$, we conclude that: for every sequent
$\Gamma \vdash \delta$, $(\emptyset)$ is rationally constructible in
$G$ iff $\Gamma \vdash \delta$ is provable.
\end{proof}

\subsubsection{Algorithms}\label{sec:algo:dicho-constr}

Let $G^\epsilon$ be an individual resource game, and let $P \in \ch(G^\epsilon)$.
To decide whether the profile $P$ can be rationally constructed we can use
Algorithm~\ref{alg:RAT-CONSTR}. This algorithm will serve for all
cases of rational construction in this paper.

\begin{algorithm}
\caption{General algorithm for {\sf RC}}
\label{alg:RAT-CONSTR}
\begin{algorithmic}[1]
\\non-deterministically guess $(\epsilon', P') \in \redis(\epsilon)
  \times \ch(G^{\epsilon'})$.
\\return $\out(P') = \out(P)$ and $P' \in NE(G^{\epsilon'})$.
\end{algorithmic}
\end{algorithm}

The algorithmic analysis is rather simple: we use the problem {\sf NE} as a blackbox, for which complexity upper bounds have been established in Proposition~\ref{prop:NP-Pi2P} and Proposition~\ref{prop:ine-dicho-weakening-easy}.
\pcase{Linear case.}
\begin{proposition}\label{prop:rc-dicho-linear}
When \logica is in $\mathsf{PTIME}$, {\sf RC} is in
$\mathsf{\Sigma_2^p}$.
  When \logica is in $\mathsf{NP}$, {\sf RC} is in
$\mathsf{\Sigma_3^p}$.  When \logica is in $\mathsf{PSPACE}$, {\sf
  RC} is in $\mathsf{PSPACE}$.
\end{proposition}
\begin{proof}
  When \logica is in $\mathsf{PTIME}$, from Proposition~\ref{prop:NP-Pi2P}, we
know that the test of line~$2$ is in $\mathsf{coNP}$.  So {\sf
  RC} is in $\mathsf{NP^{coNP}} =
\mathsf{\Sigma_2^p}$.
Similarly, when \logica is in $\mathsf{NP}$, from Proposition~\ref{prop:NP-Pi2P}, we
know that the test of line~$2$ is in $\mathsf{\Pi_2^p}$.  So {\sf
  RC} is in $\mathsf{NP^{\Pi_2^p}} =
\mathsf{\Sigma_3^p}$. The case for \logica in $\mathsf{PSPACE}$ is
analogous.
\end{proof}

\pcase{Affine case.}
Again, an affine \logica seems to bring some relative algorithmic ease.
\begin{proposition}\label{prop:rc-dicho-affine}
  If \logica is affine, when provability in \logica is in $\mathsf{PTIME}$, then {\sf RC} is in $\mathsf{NP}$.
  When \logica is in $\mathsf{NP}$, {\sf
  RC} is in $\mathsf{\Sigma_2^p}$.  When \logica is in
$\mathsf{PSPACE}$, {\sf RC} is in $\mathsf{PSPACE}$.
\end{proposition}
\begin{proof}
The proof is similar to the one of Proposition~\ref{prop:rc-dicho-linear}, using
the result of Proposition~\ref{prop:ine-dicho-weakening-easy} and, for the case of $\mathsf{NP}$ the fact that $\mathsf{NP^{P^{NP||}}} \subseteq \mathsf{NP^{\Delta_2^p}} = \mathsf{\Sigma_2^p}$.
\end{proof}

\section{Parsimonious preferences}
\label{sec:parsi-pref}

Weakening $(W)$ is sometimes a desirable property of \logica and of
our preferences of resources. However, it has the untoward consequence of incentivizing
players to spend all their resources in individual resource games with
dichotomous preferences. This is well exemplified for instance by
Proposition~\ref{prop:nonempty}.

We can teach our players parsimony by attaching to them finer
preferences that take into account the realization of their objective,
but also the optimality of their contribution.

In an individual resource game $G = (N,\gamma_1, \ldots, \gamma_n,
\epsilon_1, \ldots, \epsilon_n)$, we now say that player $i \in N$
(parsimoniously) strongly prefers $P \in \ch(G)$ over $Q \in \ch(G)$
(noted $Q \prec_i P$) iff one of the following conditions is
satisfied:
\begin{enumerate}
\item not $\out(P) \vdash \gamma_i$ and not $\out(Q) \vdash \gamma_i$
  and $P_i \subset Q_i$;
\item $\out(P) \vdash \gamma_i$ and not $\out(Q) \vdash \gamma_i$;
\item $\out(P) \vdash \gamma_i$ and $\out(Q) \vdash \gamma_i$ and $P_i
  \subset Q_i$.
\end{enumerate}
%%($P_i$ (resp.\ $Q_i$) is $i$'s contribution in $P$ (resp.\ $Q$).) 
Similar preferences have been called pseudo-dichotomous in the literature.%were used over coalitional structures in~\cite{Dunne201020}.

We recognise that the second condition corresponds to profile $P$ being dichotomously strongly preferred by Player~$i$ to profile $Q$. The following proposition is a simple consequence.
\begin{lemma}\label{lem:d-p-pref}
If Player~$i$ dichotomously strongly prefers $P$ over $Q$ then Player~$i$ parsimoniously strongly prefers $P$ over $Q$.
\end{lemma}
This has another immediate consequence on Nash equilibria.
\begin{lemma}\label{lem:NE-dicho-parsi}
If a profile $P$ is a Nash equilibrium in presence of parsimonious preferences, then $P$ is a Nash equilibrium in presence of dichotomous preferences.
\end{lemma}
\begin{proof}
Let $\prec_i^d$ (resp., $\prec_i^p$) denote Player~$i$'s parsimonious (resp., dichotomous) preferences; Let $NE_d(G)$ (resp., $NE_p(G)$) denote the set of Nash equilibria in $G$ when considering dichotomous (resp., parsimonious) preferences. Now suppose that $P \in NE_p(G)$. That is, for every $i \in N$ and for every $C_i \in \ch_i(G)$ we have not $P \prec_i^p (C_i, P_{-i})$, and by Lemma~\ref{lem:d-p-pref}, we have not $P \prec_i^d (C_i, P_{-i})$. So $P \in NE_d(G)$.
\end{proof}
Lemma~\ref{lem:NE-dicho-parsi} indicates that every Nash equilibrium in presence of parsimonious preferences is also a Nash equilibrium in presence of dichotomous preferences. The next proposition, which will help us later to prove some hardness result, says that the other way around holds when the profile is the one where every player plays the empty set of resources.
\begin{lemma}\label{lem:NE-dicho-parsi-empty}
The profile $(\emptyset,\ldots,\emptyset)$ is a Nash equilibrium in presence of parsimonious preferences iff it is a Nash equilibrium in presence of dichotomous preferences.
\end{lemma}
\begin{proof}
  Left to right is a consequence of Lemma~\ref{lem:NE-dicho-parsi}. For right to left, assume $(\emptyset,\ldots,\emptyset)$ is in $NE_d(G)$. With parsimonious preferences, the only incentive to deviate from a Nash equilibrium  in presence of dichotomous preferences, would be to play a smaller multiset of resources. This is impossible in $(\emptyset,\ldots,\emptyset)$.
\end{proof}

We now address the complexity of the decision problem of deciding whether a player parsimoniously strongly prefers a profile over another profile.
\begin{proposition}\label{prop:pars-pref-compl} Let $G = (N, \gamma_1,
  \ldots, \gamma_n, \epsilon_1, \ldots, \epsilon_n)$ be an individual
  resource game. Let also $P \in \ch(G)$ and $Q \in \ch(G)$ be two
  profiles, and $i \in N$ be a player.
  The problem of deciding whether $Q \prec_i P$ is: in $\mathsf{PTIME}$ when provability in \logica is in $\mathsf{PTIME}$.
It is in $\mathsf{P^{NP||[2]}}$ when provability in \logica is $\mathsf{NP}$-complete.
It is in $\mathsf{PSPACE}$ when provability in \logica is $\mathsf{PSPACE}$-complete.
\end{proposition}
\begin{proof}
  First, we can evaluate $P_i \subseteq Q_i$ efficiently. We store the
  result in the Boolean variable $v_{\subseteq}$.
 
  We can then perform two non-adaptive queries to an oracle to solve
  sequent validity in \logica on $\out(P) \vdash \gamma_i$ and
  on $\out(Q) \vdash \gamma_i$, and store the results in the Boolean
  variables $v_P$ and $v_Q$ respectively. The formula 
  $((\lnot v_p \land \lnot v_q \land v_{\subseteq})
  \lor (v_p \land \lnot v_q)
  \lor (v_p \land v_q \land v_{\subseteq}))$ is true iff $Q \prec_i P$.

  This yields a correct algorithm for deciding $Q \prec_i P$ in
  $\mathsf{PTIME}$ when \logica is in $\mathsf{PTIME}$, in
  $\mathsf{P^{NP||[2]}}$ when \logica is in $\mathsf{NP}$, and in
  $\mathsf{PSPACE}$ when \logica is in $\mathsf{PSPACE}$.
\end{proof}

To compare the complexity of dichotomous and parsimonious preferences, remember from Proposition~\ref{prop:compl-dicho-pref} that when \logica is in $\mathsf{NP}$, the same problem for dichotomous preferences is in $\mathsf{B\var{H_2}}$. From~\cite{kobler1987} we know that $\mathsf{P^{NP||[1]}} \subseteq \mathsf{B\var{H_2}} \subseteq \mathsf{P^{NP||[2]}}$. It is not known whether these inclusions are strict.

\subsection{Illustration of redistribution and parsimony}\label{sec:illustr-pasimony}
Consider again the individual resource game of Section~\ref{sec:elim:illustr}. (Unless stated otherwise, suppose we are in the affine case.)
%% Consider game the individual resource game 
%%  $G^\epsilon = 
%%  (\{1,2\}, \gamma_1 = B, \gamma_2 = A, \{A\}, \{B\})$. There are two players. Player~$1$ wants $B$ but is endowed with $\{A\}$, while Player~$2$ wants $A$ but is endowed with $\{B\}$. The game $G^\epsilon$ can be represented as follows. (We assume that \logica is affine for the exemple.)
%% \[
%% \begin{array}{c|c|c}
%%  \Downarrow  \ch_1 \mid \ch_2 \Rightarrow  & \emptyset          & \{B\}\\
%% \hline
%% \emptyset & \emptyset          & \{B\}: \gamma_1\\
%% \{A\}     &     \{A\}: \gamma_2    & \{A, B\}: \gamma_1, \gamma_2
%% \end{array}
%% \]
With parsimonious preferences, we have $NE(G) = \{
(\emptyset, \emptyset)\}$.  The profile $(\{A\}, \{B\})$ is not a
Nash equilibrium as it was with dichotomous preferences. It would be
more desirable from a social welfare point of view than any other
outcome (it satisfies both players), but the players would nonetheless
not be individually rational by choosing it. They have indeed no
bearing upon the outcome that satisfies them and thus are rational in
withholding their resources.

Nonetheless, like in the case of dichotomous preference, we can
effectively eliminate the current Nash equilibrium in $G^{\epsilon}$
\emph{and} construct the Nash equilibrium yielding $\{A, B\}$ by
redistributing the resources present in $G^{\epsilon}$ so as to obtain
$G^{\epsilon'} = (\{1,2\}, \gamma_1 = B, \gamma_2 = A, \{B\}, \{A\})$.
%% The game $G^{\epsilon'}$ can be represented as follows.
%% \[
%% \begin{array}{c|c|c}
%% \Downarrow  \ch_1 \mid \ch_2 \Rightarrow  & \emptyset          & \{A\}\\
%% \hline
%% \emptyset & \emptyset          & \{A\}: \gamma_2\\
%% \{B\}     &     \{B\}: \gamma_1    & \{B, A\}: \gamma_1, \gamma_2
%% \end{array}
%% \]
The only Nash equilibrium is now $(\{B\}, \{A\})$.

\medskip

Unlike dichotomous preferences, parsimonious preferences do not ensure the existence of a Nash equilibrium in the affine case.
Consider the individual resource game 
 $H^\epsilon =
 (\{1,2\}, \gamma_1 = A, \gamma_2 = A \otimes A, \{A\}, \{A\})$. There are two players. The game $H^\epsilon$ can be represented as on Figure~\ref{fig:Hepsilon}.

\begin{figure}[ht]
\begin{center}
\bgroup
\def\arraystretch{1.5}
        \begin{tabular}{|c|*{2}{>{$}C{25mm}<{$}}|}
\hline
           \cellcolor[gray]{\transparency}{\diagbox[linewidth=0.2pt, width=\dimexpr \textwidth/10+2\tabcolsep\relax, height=0.8cm]{$1$}{$2$}}
   & \cellcolor[gray]{\transparency}{\emptyset} & \cellcolor[gray]{\transparency}{\{A\}}\\
\hline
          \cellcolor[gray]{\transparency}{$ \emptyset $} & \emptyset & \{A\}: \gamma_1\\
          \cellcolor[gray]{\transparency}{$ \{A\} $} & \{A\}: \gamma_1 & \{A, A\}: \gamma_1, \gamma_2\\
\hline
         \end{tabular}
\egroup
\end{center}
\caption{\label{fig:Hepsilon} The game $H^\epsilon$. There is no Nash equilibrium under parsimonious preferences.}
\end{figure}

The game $H^\epsilon$ has no Nash equilibrium: At $(\emptyset,\emptyset)$, Player~$1$ does not realize her objective, but she can deviate and play $\{A\}$ to satisfy it. At $(\{A\},\emptyset)$, Player~$2$ has an incentive to deviate and play $\{A\}$ to realize her objective. At $(\{A\},\{A\})$ Player~$1$ has an incentive to deviate and play $\emptyset$. (In the affine case this is because she can still satisfy her objective by contributing less. In the linear case, this is because she can satisfy her objective while she does not before deviating.) At $(\emptyset,\{A\})$, Player~$2$ does not satisfy her objective and thus has an incentive to deviate to play $\emptyset$.

However, we can construct the Nash equilibrium yielding $\{A, A\}$. Let $\epsilon'$ be the redistribution of $\epsilon$ such that $\epsilon'(2) = \{A,A\}$ and $\epsilon'(1) = \emptyset$. We obtain the game depicted on Figure~\ref{fig:Hepsilonprime}.

\begin{figure}[ht]
\begin{center}
\bgroup
\def\arraystretch{1.5}
%        \begin{tabular}{|c|*{3}{>{$}C{22mm}<{$}}|}
        \begin{tabular}{|c|>{$}C{15mm}<{$}|>{$}C{15mm}<{$}|>{$}C{25mm}<{$}|}
\hline
           \cellcolor[gray]{\transparency}{\diagbox[linewidth=0.2pt, width=\dimexpr \textwidth/10+2\tabcolsep\relax, height=0.8cm]{$1$}{$2$}}
   & \cellcolor[gray]{\transparency}{\emptyset} & \cellcolor[gray]{\transparency}{\{A\}} & \cellcolor[gray]{\transparency}{\{A, A\}}\\
\hline
          \cellcolor[gray]{\transparency}{$ \emptyset $} & \emptyset & \{A\}: \gamma_1 & \{A, A\}\paNE{}: \gamma_1, \gamma_2\\
\hline
         \end{tabular}
\egroup
\end{center}
\caption{\label{fig:Hepsilonprime} The game $H^{\epsilon'}$. The symbol ${}\paNE$ denotes a Nash equilibrium.}
\end{figure}

In $H^{\epsilon'}$, by assigning all the resources to Player~$2$, the
profile $(\emptyset, \{A,A\})$ is a Nash equilibrium and the only one.
In affine logics, both players satisfy their objectives, but only
Player~$2$ does when the logic is linear.

\subsection{Finding Nash equilibria}% with parsimonious preferences}

We study the complexity of {\sf NASH EQUILIBRIUM} with parsimonious preferences.

\subsubsection{Hardness}

We are now getting used to many-to-one reductions from sequent (non-)provability. It was a fruitful problem in presence of dichotomous preference, and it will remain one in presence of parsimonious preferences. We prove a complexity lower bound for the problem of {\sf NE} in presence of parsimonious preferences.
\begin{proposition}\label{prop:iNE-hardness-parsimonious} 
The problem {\sf NE} is as hard as the problem of
checking sequent \emph{non-provability} in \logica, even when there is only
one player.
\end{proposition}
\begin{proof}
As before, we consider w.l.o.g.~only the intuitionistic sequents of
\logica in the following reduction.

Let $\Gamma \vdash \delta$ be an intuitionistic sequent of \logica.
We define $\phi = \Gamma^* \multimap \delta$.
We can construct the individual resource game $G$ such that $G =
(\{1\},\phi, \{ \phi\})$.
In $G$, Player~1 has exactly two choices: $\ch_i(G) = \{\emptyset, \{\phi\}\}$.

We show that $\Gamma \vdash \delta$ iff $\phi \not\in NE(G)$.

Suppose $(\{\phi\}) \not\in NE(G)$. So $(\{\phi\}) \prec_1 (\emptyset)$. Since by (ax) $\phi \vdash \phi$ (the profile $(\{\phi\})$ satisfies Player~$1$'s objectives) and $\emptyset \subset \{\phi\}$ (Player~$1$'s contribution is strictly less in the profile $(\emptyset)$ than it is in $(\{\phi\})$), it must be that $\emptyset \vdash \phi$. We infer $\Gamma \vdash \delta$, as we did in part of the proof of Proposition~\ref{prop:hard-rat-elim-dicho}.% (by $\multimap$L, $\otimes$L, (cut)).

Suppose $\Gamma \vdash \delta$. We obtain $\Gamma^* \vdash \delta$ by using ($\otimes$L) enough times, and we deduce $\vdash \phi$ with ($\multimap$R). We thus have $\emptyset \vdash \phi$ and $\emptyset \subset \{\phi\}$. So $(\{\phi\}) \prec_1 (\emptyset)$ and $(\{\phi\}) \not\in NE(G)$.
\end{proof}

\subsubsection{Algorithms}

\pcase{Linear case.}
In the individual resource game $G = (N, \gamma_1, \ldots, \gamma_n, \epsilon_1, \ldots, \epsilon_n)$, we can use Algorithm~\ref{GA-co-iNE} to check whether a profile $P \not\in NE(G)$, even for parsimonious preferences.
We have a result analogous to Proposition~\ref{prop:NP-Pi2P} for parsimonious preferences.

\begin{proposition}\label{prop:ine-easy-NP-pars-linear}\label{prop:ine-easy-pspace-pars}
If the problem of sequent provability in \logica is in $\mathsf{PTIME}$ then \iNE is in $\mathsf{coNP}$.  
If the problem of sequent provability in \logica is in $\mathsf{NP}$ then \iNE is in $\mathsf{\Pi_2^p}$. If the problem of sequent provability in
\logica is in $\mathsf{PSPACE}$ then \iNE is in $\mathsf{PSPACE}$.

\begin{proof}We use Proposition~\ref{prop:pars-pref-compl} and, for the case of $\mathsf{NP}$, the fact that
  $\mathsf{coNP^{{P}^{NP||[2]}[1]}}  
\subseteq \mathsf{coNP^{{P}^{NP}}}
= \mathsf{coNP^{\Delta_2^p}}
= \mathsf{co\Sigma_2^p}
= \mathsf{\Pi_2^p}$.
\end{proof}
\end{proposition}

\pcase{Affine case.}
When \logica is affine, we can do better than using Algorithm~\ref{GA-co-iNE}.
We first state a technical lemma which is analogous to Lemma~\ref{lemma:ne-dich-affin}.
  \begin{lemma}\label{lemma:ne-parsi-affin}
    Let $G = (N, \gamma_1, \ldots, \gamma_n, \epsilon_1, \ldots,
    \epsilon_n)$ be an individual resource game. When \logica is affine, $P
    \not \in NE(G)$ iff $\exists i \in N:$ s.t.\ either:
    \begin{enumerate}
    \item $\out(P) \not\vdash \gamma_i$ and $P_i \not = \emptyset$;
    \item $\out(P) \not\vdash \gamma_i$ and
      $\out((P_{-i},\epsilon_i)) \vdash \gamma_i$;
    \item $\out(P) \vdash \gamma_i$ and $\exists A \in P_i$:
      $\out((P_{-i}, P_i \setminus \{A\})) \vdash \gamma_i$.
    \end{enumerate}
    \begin{proof}
      Right to left is immediate. From left to right, suppose $P \not
      \in NE(G)$. So there exists $i \in N$ and $C_i \in \ch_i(G)$
      such that $P\prec_i (P_{-i}, C_i)$. There are three
      cases to consider:
\begin{enumerate}
\item not $\out((P_{-i}, C_i)) \vdash \gamma_i$ and not $\out(P)
  \vdash \gamma_i$ and $C_i \subset P_i$;
\item $\out((P_{-i}, C_i)) \vdash \gamma_i$ and not $\out(P) \vdash
  \gamma_i$;
\item $\out((P_{-i}, C_i)) \vdash \gamma_i$ and $\out(P) \vdash
  \gamma_i$ and $C_i \subset P_i$.
\end{enumerate}
Suppose~(1) is the case. It implies that there is $C_i \subset P_i$
and thus that $P_i \not = \emptyset$.  Suppose~(2) is the case.
We essentially use the same argument as the one used in the proof of
Lemma~\ref{lemma:ne-dich-affin}.
We have $\out((P_{-i}, C_i)) \vdash \gamma_i$. By applying weakening
$(|\epsilon_i| - |C_i|)$ times, we easily obtain that $\out((P_{-i},
\epsilon_i)) \vdash \gamma_i$. Suppose~(3) is the case. We thus have
$\out((P_{-i}, C_i)) \vdash \gamma_i$ with $C_i \subset P_i$. Take a
formula $A \in P_i \setminus C_i$. Then, by applying weakening $(|P_i|
- |C_i| - 1)$ times, we easily obtain that $\out((P_{-i}, P_i
\setminus \{A\})) \vdash \gamma_i$.
    \end{proof}
  \end{lemma}

Algorithm~\ref{alg:PA-iNE} can then be used to check whether $P \in NE(G)$.\footnote{Algorithm~\ref{alg:PA-iNE} corrects an omission in~\cite[Algo.~$5$]{Tr16ijcai} by adding ``if ($P_i \not = \emptyset$): return false'' lines $9$ and $10$.}

\begin{algorithm}
\caption{Algorithm for \iNE with parsimonious preferences and affine \logica}
\label{alg:PA-iNE}
\begin{algorithmic}[1]
\\for each $i \in N$ do:
\\\indent\indent if ($\out(P) \vdash \gamma_i$) : \{
\\\indent\indent\indent for each $A \in P_i$ do:
\\\indent\indent\indent\indent if ($\out((P_{-i},P_i\setminus\{A\})) \vdash \gamma_i$):
\\\indent\indent\indent\indent\indent return false.% [$P \not\in NE(G)$]
\\\indent\indent \} else \{
\\\indent\indent\indent if ($\out((P_{-i},\epsilon_i)) \vdash \gamma_i$):
\\\indent\indent\indent\indent return false.% [$P \not\in NE(G)$]
\\\indent\indent\indent if ($P_i \not = \emptyset$):
\\\indent\indent\indent\indent return false.  % [$P \not\in NE(G)$] 
\\\indent \indent \}
\\return true.% [$P \in NE(G)$]
\end{algorithmic}
\end{algorithm}

\begin{proposition}\label{prop:ine-nondicho-NP}
  When \logica is affine, if the problem of sequent provability in \logica is in $\mathsf{PTIME}$, then \iNE is in $\mathsf{PTIME}$.
  If the problem of sequent provability in \logica is in $\mathsf{NP}$, then \iNE is in $\mathsf{P^{NP||}}$.
%$\mathsf{\Delta_2^p}$.
If the problem of sequent provability in \logica is in $\mathsf{PSPACE}$, then \iNE is in $\mathsf{PSPACE}$.
\end{proposition}
\begin{proof}
Lemma~\ref{lemma:ne-parsi-affin} justifies the correctness of Algorithm~\ref{alg:PA-iNE}.
%When sequent validity in \logica is in $\mathsf{NP}$, the algorithm can be simulated by a deterministic oracle Turing machine in polynomial time with at most $\Sigma_{i\in N} (1 + |P_i|)$ calls to a $\mathsf{NP}$ oracle.
%Thus, $P \in NE(G)$ is a $\mathsf{P^{NP[\text{$\Sigma_i (1 + |P_i|)$}]}}$ predicate.
The algorithm can be simulated by a deterministic oracle Turing
  machine in polynomial time with less than $\Sigma_{i\in N} (1 + |P_i|)$ non-adaptive queries to an oracle for sequent provability in \logica.
When the complexity of sequent provability in \logica is in $\mathsf{NP}$ it yields a complexity of $\mathsf{P^{NP||}}$.
\end{proof}

\subsection{Elimination}% with parsimonious preferences}

We study the complexity of {\sf RATIONAL ELIMINATION} with parsimonious preferences.

\subsubsection{Algorithms}

Lemma~\ref{lemma:rat-elimination} also holds for parsimonious
preferences. It is easy to see that the proof carries over.

\pcase{Linear case.}
Algorithm~\ref{GA-RE} can still be used in the case of parsimonious preferences
because Lemma~\ref{lemma:rat-elimination} is still granted.
We thus have the analog to Proposition~\ref{prop:easyness-rat-elim-linear} for
parsimonious preferences.
\begin{proposition}\label{prop:easyness-rat-elim-linear-pars}
  When \logica is linear, {\sf RE} is in $\mathsf{NP}$ when sequent provability in \logica is in $\mathsf{PTIME}$, in $\mathsf{\Sigma_2^p}$ when \logica is in $\mathsf{NP}$, and in $\mathsf{PSPACE}$ when \logica is in $\mathsf{PSPACE}$.
  \begin{proof}We use Proposition~\ref{prop:pars-pref-compl} and, in the case of $\mathsf{NP}$, the fact that $\mathsf{NP^{{P}^{NP||[2]}[1]}}
\subseteq \mathsf{NP^{{P}^{NP}}}
= \mathsf{NP^{\Delta_2^p}}
= \mathsf{\Sigma_2^p}$.
  \end{proof}
\end{proposition}

\pcase{Affine case.} Let $G = (N,\gamma_1, \ldots, \gamma_n, \epsilon_1, \ldots, \epsilon_n)$ be an individual resource game and let $P \in \ch(G)$ be a profile.
We can use Algorithm~\ref{alg:PA-RE} to check whether a profile $P \in \ch(G)$ is rationally eliminable.

\begin{algorithm}
\caption{Algorithm for {\sf RE} with parsimonious preferences and affine \logica}
\label{alg:PA-RE}
\begin{algorithmic}[1]
\\for each $i \in N$ do:
\\\indent\indent if ($(\out(P)) \prec_i ([\epsilon\arrow i](i))$):\label{line-PA-RE-test1}
\\\indent\indent\indent return true.
\\\indent\indent for each $A \in \out(P)$:
\\\indent\indent \indent if ($(\out(P))  \prec_i (\out(P)\setminus \{A\})$):\label{line-PA-RE-test2}
\\\indent\indent \indent \indent return true.
\\return false.
\end{algorithmic}
\end{algorithm}

%% \begin{algorithm}
%% \caption{Algorithm for {\sf RE} with parsimonious preferences and affine \logica}
%% \label{alg:PA-RE-alt2}
%% \begin{algorithmic}[1]
%%   \\for each $i \in N$ do:
%%   \\\indent\indent if ($(\out(P)) \not\in NE(G^{[\epsilon\arrow i]})$):
%%   \\\indent\indent\indent return true.
%%   \\return false.
%% \end{algorithmic}
%% \end{algorithm}

\begin{proposition}\label{prop:easyness-rat-elim-affine-pars}
  When \logica is affine, {\sf RE} is in
  $\mathsf{PTIME}$ when provability in \logica is in $\mathsf{PTIME}$.
  It is in $\mathsf{P^{NP||}}$
  when \logica is in $\mathsf{NP}$. It is in $\mathsf{PSPACE}$ when \logica is in $\mathsf{PSPACE}$.

\begin{proof}
%   Algorithm~\ref{alg:PA-RE-alt2} returns true if and only iff the
%   profile $(\out(P)) \in G^{[\epsilon\arrow i]}$ is not a Nash
%   equilibrium for some agent $i$.
%
    Lemma~\ref{lemma:rat-elimination} which still holds with
    parsimonious preferences ensures that it is enough to
    consider the redistributions $[\epsilon\arrow i]$ for some player
    $i$.
    Algorithm~\ref{alg:PA-RE}, then checks for each of these redistributions whether Player~$i$ has an incentive to deviate in the game $G^{[\epsilon\arrow i]}$ from the profile $(\out(P)) \in \ch(G^{[\epsilon\arrow i]})$ to any one of $([\epsilon\arrow i](i)) \in \ch(G^{[\epsilon\arrow i]})$ and $(\out(P) \setminus \{A\}) \in \ch(G^{[\epsilon\arrow i]})$ for some $A \in \out(P)$. It is weakening $(W)$ that justifies that it is enough to consider these profiles, because $X \not \vdash \gamma_i$ implies $Y \not \vdash \gamma_i$ for any couple of multisets $Y \subseteq X$.
The correctness of Algorithm~\ref{alg:PA-RE} follows.
    %
    %% Moreover, every player $j \not = i$ is a dummy player in
    %% $G^{[\epsilon\arrow i]}$. Thus, for every $P' \in
    %% \ch(G^{[\epsilon\arrow i]})$, if $\out(P') = \out(P)$, then $P'$
    %% is in $NE(G^{[\epsilon\arrow i]})$ iff $(\out(P))$ is in
    %% $NE(G^{[\epsilon\arrow i]})$. The correctness of
    %% Algorithm~\ref{alg:PA-RE-alt} follows.

    %% The algorithm can be simulated by a deterministic oracle Turing
    %% machine in polynomial time with at most $n$ calls to an oracle
    %% deciding \iNE. If the problem of sequent validity of \logica is in
    %% $\mathsf{NP}$, we know from Proposition~\ref{prop:ine-nondicho-NP}
    %% that \iNE is in $\mathsf{\Delta_2^p}$. The fact that
    %% $\mathsf{P^{\Delta_2^p}} = \mathsf{\Delta_2^p}$ yields the
    %% result. With the same reasoning, it is in $\mathsf{PSPACE}$ when
    %% the problem of sequent validity of \logica is in
    %% $\mathsf{PSPACE}$.

    %% If $f$ is the number of formulas in the outcome $P$, at most $n(f+1)$ profiles will be compared to $P$.
    
The tests of line~\ref{line-PA-RE-test1} and line~\ref{line-PA-RE-test2} only involve the following instances of the sequent provability decision problem: 
$(\out(P)) \vdash \gamma_i$ and
$([\epsilon\arrow i](i)) \vdash \gamma_i$ for very Player~$i \in N$, and 
$(\out(P)\setminus \{A\}) \vdash \gamma_i$, for every Player~$i \in N$ and every formula $A \in \out(P)$.
The algorithm can thus be simulated by a deterministic oracle Turing machine in polynomial time with at most $|N|(|\out(P)|+2)$ non-adaptive calls to an oracle for sequent provability.
\end{proof}
\end{proposition}
%% Remark that Algorithm~\ref{alg:PA-RE-alt2} is in fact correct for
%% solving \textsf{RATIONAL ELIMINATION} no matter what are the
%% assumptions on \logica. Even though it is deterministic, it is not
%% necessarily better. For instance if \logica is linear and its sequent
%% validity problem is in $\mathsf{NP}$, then using
%% Proposition~\ref{prop:ine-easy-NP-pars-linear},
%% Algorithm~\ref{alg:PA-RE-alt2} yields a complexity of
%% $\mathsf{P}^\mathsf{\Pi_2^p} = \mathsf{\Delta_3^p}$. We do better with
%% Algorithm~\ref{GA-RE} (Proposition~\ref{prop:easyness-rat-elim-linear}).

\subsubsection{Hardness}

After Lemma~\ref{lem:NE-dicho-parsi-empty} and the proof of Proposition~\ref{prop:hard-rat-elim-dicho}, the following proposition does not come as a surprise.
\begin{proposition}\label{prop:hardness-elimination-parsimonious}
{\sf RE} is as hard as the problem of checking
sequent \emph{non-provability} in \logica.
\end{proposition}
\begin{proof}
Let $\Gamma \vdash \delta$ be an arbitrary intuitionistic sequent. We construct the same game as in the proof of Proposition~\ref{prop:hard-rat-elim-dicho}.
Let $\phi = \Gamma^* \multimap \delta$. Let $G^\epsilon = (\{1,2\}, \phi, \mathbf{1}, \emptyset, \{\phi\})$. 

In the proof of Proposition~\ref{prop:hard-rat-elim-dicho},
we showed that, in presence of dichotomous preferences, both in the case of linear and of affine logics, we have $\Gamma \not\vdash \delta$ iff $(\emptyset, \emptyset)$ is rationally eliminable in $G^\epsilon$.

Now with Lemma~\ref{lem:NE-dicho-parsi-empty}, we know that $(\emptyset, \emptyset)$ is a Nash equilibrium in presence of dichotomous preferences iff it is a Nash equilibrium in presence of parsimonious preferences (both in $G^{\epsilon}$ and $G^{\epsilon'}$, and no matter if \logica is linear or affine, or if $\Gamma \vdash \delta$ or $\Gamma \not\vdash \delta$).

Hence, we have $\Gamma \not\vdash \delta$ iff $(\emptyset, \emptyset)$ is rationally eliminable in $G^\epsilon$, also in presence of parsimonious preferences.
\end{proof}

%% \begin{proposition}\label{prop:hardness-elimination-parsimonious-2}

%% {\sf RATIONAL ELIMINATION} is as hard as the problem of checking
%% sequent validity in \logica.
%% \begin{proof}
%% Considering the same games as in the proof of Proposition~\ref{prop:hardness-elimination-parsimonious} we can see that
%% $\Gamma \vdash \delta$ iff $(\{\phi\}, \emptyset)$ is rationally eliminable in $G^{\epsilon'}$.
%% \end{proof}
%% \end{proposition}

\subsection{Construction}% with parsimonious preferences}

Finally, we tackle the complexity of {\sf RATIONAL CONSTRUCTION} with parsimonious preferences.

\subsubsection{Hardness}
We establish a complexity lower bound for the problem of {\sf RC} in presence of parsimonious preferences.
\begin{proposition}\label{prop:RC-parsim-hard}
{\sf RC} is as hard as the problem of checking
sequent \emph{non-provability} in \logica.
\end{proposition}
\begin{proof}
Consider the games in the proof of Proposition~\ref{prop:hardness-elimination-parsimonious}. We can see that both for linear and affine logics we have that
$\Gamma \not\vdash \delta$ iff $(\{\phi\}, \emptyset)$ can be rationally constructed in $G^{\epsilon'}$.
\end{proof}

\subsubsection{Algorithms}

Our algorithmic analysis is very similar to the analysis we made when the preferences are dichotomous in Section~\ref{sec:algo:dicho-constr}. Let $G^\epsilon$ be an individual resource game and $P \in \ch(G^\epsilon)$. To decide whether $P$ can be rationally constructed we can reuse Algorithm~\ref{alg:RAT-CONSTR}.

Again, we use the problem {\sf NE} as a blackbox, for which complexity upper bounds have been established in Proposition~\ref{prop:ine-easy-NP-pars-linear} and Proposition~\ref{prop:ine-nondicho-NP}.
\pcase{Linear case.}
\begin{proposition}\label{prop:rc-parsi-linear}
  When sequent provability in \logica is in $\mathsf{PTIME}$, {\sf RC} is in
  $\mathsf{\Sigma_2^p}$.
  When \logica is in $\mathsf{NP}$, {\sf RC} is in
$\mathsf{\Sigma_3^p}$.  When \logica is in $\mathsf{PSPACE}$, {\sf
  RC} is in $\mathsf{PSPACE}$.
\end{proposition}
\begin{proof}
The proof is similar to the one of Proposition~\ref{prop:rc-dicho-linear},
using the result of Proposition~\ref{prop:ine-easy-NP-pars-linear}.
\end{proof}

The next proposition also comes without surprise.
\pcase{Affine case.}
\begin{proposition}\label{prop:rc-parsi-affine}
  If \logica is affine, when \logica is in $\mathsf{PTIME}$, {\sf
  RC} is in $\mathsf{NP}$.
  When \logica is in $\mathsf{NP}$, {\sf
  RC} is in $\mathsf{\Sigma_2^p}$.  When \logica is in
$\mathsf{PSPACE}$, {\sf RC} is in $\mathsf{PSPACE}$.
\end{proposition}
\begin{proof}
The proof is similar to the one of Proposition~\ref{prop:rc-dicho-affine},
using the result of Proposition~\ref{prop:ine-nondicho-NP}.
\end{proof}

\section{Examples}
\label{sec:examples}

We present more thorough examples. They involve several resources and objectives that are modeled with a variety of logical operands. We take the opportunity to present fully the important formal proofs of the realized objectives.

We start with a toy example, simple but rich enough, upon which we can demonstrate all the frameworks and problems addressed in the paper.

Then, we formally study the divorce arbitration scenario of Example~\ref{ex:divorce}, as well as a three-player variant of the scenario of interconnected economies from Example~\ref{example:telecom1}.

\subsection{Alan and the fish}\label{sec:alan-fish}

We first introduce the resources involved and how they are built in
the logical language.
\begin{itemize}
\item Basic resources:
\begin{itemize}
\item one mole of dioxygen: $\var{O_2}$
\item one mole of dihydrogen: $\var{H_2}$
\item one mole of water: $\var{H_2O}$
\item one `token' of thirst: $\var{thirst}$
\end{itemize}
\item Anti-resources can be captured via the linear negation:
\begin{itemize}
\item one thirst quencher: $\lineg \var{thirst}$
\end{itemize}
\item Resource transformation processes:
\begin{itemize}
\item one process of electrolysis: $\var{elec} = \var{H_2O} \otimes \var{H_2O} \multimap \var{H_2} \otimes \var{H_2} \otimes \var{O_2}$
\item one process of drinking water: $\var{drink} = \var{H_2O} \multimap \lineg \var{thirst}$
\end{itemize}
\end{itemize}

\paragraph{Game definition.}
Let $G^\epsilon_{\mathit{af}} = (\{a,f\},\gamma_a, \gamma_f, \epsilon_a,
\epsilon_f)$ be the individual resource game with two players, Alan $a$ and
the Fish $f$. The fish wants one mole of dioxygen: $\gamma_f =
\var{O_2}$. Alan wants one mole of dioxygen for his fish and wants to
quench his thirst: $\gamma_a = \var{O_2} \otimes \lineg \var{thirst}$.

In the game $G^\epsilon_{\mathit{af}}$, Alan is endowed with $\epsilon_a = \{\var{drink}, \var{elec}\}$. He can
drink once and can electrolyse water once. The fish is endowed with
three tokens of water $\epsilon_f = \{\var{H_2O},\var{H_2O},\var{H_2O}\}$.

\medskip
We suppose that \logica is affine. For this example, we will consider both cases of dichotomous and parsimonious preferences.
\medskip

As we did before, we will represent a Nash equilibrium under dichotomous preferences with the symbol~${}\diNE$, and under parsimonious preferences with the symbol~${}\paNE$. By Lemma~\ref{lem:NE-dicho-parsi}, the latter implies the former. Then, when a profile is a Nash equilibrium under both dichotomous and parsimonious preferences we will use the symbol~${}\dipaNE$.
The game $G^\epsilon_{\mathit{af}}$ and the realized objectives can be depicted as on Figure~\ref{fig:Gaf-epsilon}.

\begin{figure}[ht]
\begin{center}
\resizebox{\textwidth}{!}{
\bgroup
\def\arraystretch{1.5}
        \begin{tabular}{|c|*{4}{>{$}c<{$}}|}
\hline
           \cellcolor[gray]{\transparency}{\diagbox[linewidth=0.2pt, width=\dimexpr \textwidth/10+2\tabcolsep\relax, height=0.8cm]{$a$}{$f$}}
   & \cellcolor[gray]{\transparency}{\emptyset} & \cellcolor[gray]{\transparency}{\{\var{H_2O}\}} & \cellcolor[gray]{\transparency}{\{\var{H_2O}, \var{H_2O}\}} & \cellcolor[gray]{\transparency}{\{\var{H_2O}, \var{H_2O}, \var{H_2O}\}}\\
\hline
\cellcolor[gray]{\transparency}{$\emptyset$} &  \emptyset\dipaNE &  \{\var{H_2O}\}\diNE      & \{\var{H_2O},\var{H_2O}\}\diNE     & \{\var{H_2O},\var{H_2O},\var{H_2O}\} \\
\cellcolor[gray]{\transparency}{$\{\var{drink}\}$}     &  \{\var{drink}\}\diNE     &  \{\var{drink},\var{H_2O}\}\diNE    & \{\var{drink},\var{H_2O},\var{H_2O}\}\diNE   & \{\var{drink},\var{H_2O},\var{H_2O},\var{H_2O}\}\\
\cellcolor[gray]{\transparency}{$\{\var{elec}\}$}     &  \{\var{elec}\}     &  \{\var{elec},\var{H_2O}\}    & \{\var{elec},\var{H_2O},\var{H_2O}\}\diNE:\gamma_f   & \{\var{elec},\var{H_2O},\var{H_2O},\var{H_2O}\}:\gamma_f  \\
\cellcolor[gray]{\transparency}{$\{\var{drink}, \var{elec}\}$}  &  \{\var{drink}, \var{elec}\}  &  \{\var{drink},\var{elec},\var{H_2O}\}  & \{\var{drink},\var{elec},\var{H_2O},\var{H_2O}\}\diNE:\gamma_f & 
\{\var{drink},\var{elec},\var{H_2O},\var{H_2O},\var{H_2O}\}\diNE:\gamma_a, \gamma_f\\
\hline
         \end{tabular}
\egroup
}
\end{center}
\caption{\label{fig:Gaf-epsilon} The game $G^{\epsilon}_{\mathit{af}}$. Alan plays rows, and the fish plays columns. \logica is affine. The symbol ${}\diNE$ marks the Nash equilibria under dichotomous preferences. The symbol ${}\dipaNE$ marks the profiles that are also Nash equilibria under both dichotomous and parsimonious preferences.}
\end{figure}
Appendix~\ref{sec:proofs-realized-obj} provides the detailed proofs of the realized objectives.

\paragraph{Dichotomous preferences: eliminations of bad equilibria.}
If the preferences are dichotomous, there are plenty Nash equilibria in
$G^\epsilon_{\mathit{af}}$. They are:
$(\emptyset, \emptyset)$, $(\emptyset, \{\var{H_2O}\})$, $(\emptyset, \{\var{H_2O}, \var{H_2O}\})$,
$(\{\var{drink}\}, \emptyset)$, $(\{\var{drink}\},\linebreak \{\var{H_2O}\})$, $(\{\var{drink}\}, \{\var{H_2O}, \var{H_2O}\})$,
$(\{\var{elec}\}, \{\var{H_2O},\var{H_2O}\})$, $(\{\var{drink},\var{elec}\}, \{\var{H_2O},\var{H_2O}\})$,
and $(\{\var{drink},\var{elec}\},\linebreak \{\var{H_2O},\var{H_2O},\var{H_2O}\})$.

However, only the profile $(\{\var{drink},\var{elec}\}, \{\var{H_2O},\var{H_2O},\var{H_2O}\})$, whose outcome is
$\{\var{drink},\var{H_2O},\var{H_2O},\var{H_2O},\linebreak \var{elec}\}$, satisfies the objectives of both players.
It would thus be desirable to eliminate the other profiles. To do so,
let $\epsilon'$ be the endowment such that $\epsilon'_a =
\{\var{drink},\var{elec},\var{H_2O},\var{H_2O},\var{H_2O}\}$ and $\epsilon'_f = \emptyset$. The game
$G^{\epsilon'}_{\mathit{af}}$ and the realized objectives can be (partially)
depicted as on Figure~\ref{fig:Gaf-epsilonprime}.

\begin{figure}[ht]
\begin{center}
\bgroup
\def\arraystretch{1.5}
        \begin{tabular}{|c|*{1}{>{$}c<{$}}|}
\hline
           \cellcolor[gray]{\transparency}{\diagbox[linewidth=0.2pt, width=\dimexpr \textwidth/4+2\tabcolsep\relax, height=0.8cm]{$a$}{$f$}}
   & \cellcolor[gray]{\transparency}{\emptyset}\\
\hline
  \cellcolor[gray]{\transparency}{$\emptyset$} & \emptyset\\
  \cellcolor[gray]{\transparency}{$\{\var{drink}\}$} & \{\var{drink}\}\\
  \cellcolor[gray]{\transparency}{$\vdots$} & \vdots\\
  \cellcolor[gray]{\transparency}{$\{\var{elec},\var{H_2O}, \var{H_2O}\}$} & \{\var{elec},\var{H_2O}, \var{H_2O}\}:\gamma_f\\
  \cellcolor[gray]{\transparency}{$\vdots$} & \vdots\\
  \cellcolor[gray]{\transparency}{$\{\var{elec},\var{H_2O}, \var{H_2O}, \var{H_2O}\}$} & \{\var{elec},\var{H_2O}, \var{H_2O}, \var{H_2O}\}:\gamma_f\\
  \cellcolor[gray]{\transparency}{$\vdots$} & \vdots\\
  \cellcolor[gray]{\transparency}{$\{\var{drink},\var{elec},\var{H_2O}, \var{H_2O}\}$} & \{\var{drink},\var{elec},\var{H_2O}, \var{H_2O}\}:\gamma_f\\
  \cellcolor[gray]{\transparency}{$\{\var{drink},\var{elec},\var{H_2O}, \var{H_2O}, \var{H_2O}\}$} & \{\var{drink},\var{elec},\var{H_2O}, \var{H_2O}, \var{H_2O}\}\dipaNE:\gamma_a, \gamma_f\\
\hline
         \end{tabular}
\egroup
\end{center}
\caption{\label{fig:Gaf-epsilonprime} The game $G^{\epsilon'}_{\mathit{af}}$.}
\end{figure}

It is readily seen that in $G^{\epsilon'}_{\mathit{af}}$, when preferences are
dichotomous, only the profile\linebreak $(\{\var{drink},\var{elec},\var{H_2O},\var{H_2O},\var{H_2O}\}, \emptyset)$
whose outcome is $\{\var{drink},\var{H_2O},\var{H_2O},\var{H_2O},\var{elec}\}$, is a Nash equilibrium.

\paragraph{Parsimonious preferences: construction of a good equilibrium.}
If the preferences are parsimonious, the profile $(\emptyset,
\emptyset)$ is a Nash equilibrium in the game $G^{\epsilon}_{\mathit{af}}$,
and is the only one. One can
nonetheless redistribute the resources so as to construct an
equilibrium where Alan and the fish both realize their
objectives. That is, one can construct the profile $(\{\var{drink}, \var{elec}\},
\{\var{H_2O},\var{H_2O},\var{H_2O}\})$.
To do so, let $\epsilon''$ be the endowment such that $\epsilon''_a =
\{\var{drink},\var{H_2O},\var{H_2O},\var{H_2O}\}$ and $\epsilon''_f = \{\var{elec}\}$. The game
$G^{\epsilon''}_{\mathit{af}}$ and the realized objectives can be depicted as
on Figure~\ref{fig:Gaf-epsilonsecond}.

\begin{figure}[ht]
\begin{center}
\bgroup
\def\arraystretch{1.5}
        \begin{tabular}{|c|*{2}{>{$}c<{$}}|}
\hline
           \cellcolor[gray]{\transparency}{\diagbox[linewidth=0.2pt, width=\dimexpr \textwidth/4+2\tabcolsep\relax, height=0.8cm]{$a$}{$f$}}
   & \cellcolor[gray]{\transparency}{\emptyset} & \cellcolor[gray]{\transparency}{\{\var{elec}\}}\\
\hline
  \cellcolor[gray]{\transparency}{$\emptyset$}   &  \emptyset\dipaNE   &     \{\var{elec}\}\\
\cellcolor[gray]{\transparency}{$\{\var{H_2O}\}$}       &  \{\var{H_2O}\}\diNE       &  \{\var{H_2O},\var{elec}\}\\ 
\cellcolor[gray]{\transparency}{$\{\var{H_2O},\var{H_2O}\}$}     &  \{\var{H_2O},\var{H_2O}\}     &  \{\var{H_2O},\var{H_2O},\var{elec}\}:\gamma_f\\ 
\cellcolor[gray]{\transparency}{$\{\var{H_2O},\var{H_2O},\var{H_2O}\}$}   &  \{\var{H_2O},\var{H_2O},\var{H_2O}\}   &  \{\var{H_2O},\var{H_2O},\var{H_2O},\var{elec}\}:\gamma_f\\ 
\cellcolor[gray]{\transparency}{$\{\var{drink}\}$}       &  \{\var{drink}\}\diNE       &  \{\var{drink},\var{elec}\}\\ 
\cellcolor[gray]{\transparency}{$\{\var{drink},\var{H_2O}\}$}     &  \{\var{drink},\var{H_2O}\}\diNE     &  \{\var{drink},\var{H_2O},\var{elec}\}\\ 
\cellcolor[gray]{\transparency}{$\{\var{drink},\var{H_2O},\var{H_2O}\}$}   &  \{\var{drink},\var{H_2O},\var{H_2O}\}   &  \{\var{drink},\var{H_2O},\var{H_2O},\var{elec}\}:\gamma_f\\ 
\cellcolor[gray]{\transparency}{$\{\var{drink},\var{H_2O},\var{H_2O},\var{H_2O}\}$} &  \{\var{drink},\var{H_2O},\var{H_2O},\var{H_2O}\} &  \{\var{drink},\var{H_2O},\var{H_2O},\var{H_2O},\var{elec}\}\dipaNE:\gamma_a, \gamma_f\\
\hline
         \end{tabular}
\egroup
\end{center}
\caption{\label{fig:Gaf-epsilonsecond} The game $G^{\epsilon''}_{\mathit{af}}$.}
\end{figure}

When preferences are parsimonious, the profiles $(\emptyset,
\emptyset)$ and $(\{\var{drink},\var{H_2O},\var{H_2O},\var{H_2O}\}, \{\var{elec}\})$ are Nash equilibria
in $G^{\epsilon''}_{\mathit{af}}$ and are the only ones.

Notice that, the redistribution $\epsilon'$ would also effectively construct the profile, although at the price of a more draconian redistribution. It would also eliminate $(\emptyset, \emptyset)$.

\subsection{Ann and Bernard get a divorce}\label{sec:divorce}

We formalize Example~\ref{ex:divorce}. We will only consider parsimonious preferences.
We also assume that \logica is Affine MLL.
%This time, we will look at the distinctions between the choice Affine MLL and (Linear) MLL for \logica.
We introduce the resources involved in the example.
\begin{itemize}
\item the lease agreement: $\var{shop}$ 
\item the resource of flour for a year: $\var{flour}$
\item the resource of one year worth of bread: $\var{bread}$
\item the bread making equipment is the resource transformation process: $\var{flour} \multimap \var{bread}$
\end{itemize}

Using these as basic resources, we formalize Example~\ref{ex:divorce} as the game $G^{\epsilon}_{\mathit{ab}}$.

\paragraph{Game definition.}
Let $G^\epsilon_{\mathit{ab}} = (\{a,b\},\gamma_a, \gamma_b, \epsilon_a,
\epsilon_b)$ be the individual resource game with two players, Ann $a$ and
Bernard $b$. Ann wants enough bread for a year: $\gamma_a = \var{bread}$. Bernard wants the lease agreement: $\gamma_b = \var{shop}$.
In the game $G^\epsilon_{\mathit{ab}}$, Ann is endowed with the lease agreement: $\epsilon_a = \{\var{shop}\}$. Bernard is endowed with enough flour to make bread for two years, and with the bread making equipment: $\epsilon_b = \{\var{flour}, \var{flour}, \var{\var{flour} \multimap \var{bread}}\}$.

The game $G^{\epsilon}_{\mathit{ab}}$ and the realized objectives can be depicted as
on Figure~\ref{fig:Gab-epsilon}.
\begin{figure}[ht]
\begin{center}
\resizebox{\textwidth}{!}{
\bgroup
\def\arraystretch{1.5}
        \begin{tabular}{|c|*{2}{>{$}c<{$}}|}
\hline
           \cellcolor[gray]{\transparency}{\diagbox[linewidth=0.2pt, width=\dimexpr \textwidth/3+2\tabcolsep\relax, height=0.8cm]{$b$}{$a$}}
   & \cellcolor[gray]{\transparency}{\emptyset} & \cellcolor[gray]{\transparency}{\{\var{shop}\}}\\
           \hline
           \cellcolor[gray]{\transparency}{$\emptyset$} & \emptyset \dipaNE & \{\var{shop}\}:\gamma_b\\
           \cellcolor[gray]{\transparency}{$\{\var{flour}\}$} & \{\var{flour}\} & \{\var{flour}, \var{shop}\}:\gamma_b\\
           \cellcolor[gray]{\transparency}{$\{\var{flour} \multimap \var{bread}\}$} & \{\var{flour} \multimap \var{bread}\} & \{\var{flour} \multimap \var{bread}, \var{shop}\}:\gamma_b\\
           \cellcolor[gray]{\transparency}{$\{\var{flour}, \var{flour}\}$} & \{\var{flour}, \var{flour}\} & \{\var{flour}, \var{flour}, \var{shop}\}:\gamma_b\\
           \cellcolor[gray]{\transparency}{$\{\var{flour}, \var{flour} \multimap \var{bread}\}$} & \{\var{flour}, \var{flour} \multimap \var{bread}\}:\gamma_a & \{\var{flour}, \var{flour} \multimap \var{bread}, \var{shop}\}:\gamma_b,\gamma_a\\
           \cellcolor[gray]{\transparency}{$\{\var{flour}, \var{flour}, \var{flour} \multimap \var{bread}\}$} & \{\var{flour}, \var{flour}, \var{flour} \multimap \var{bread}\}:\gamma_a & \{\var{flour}, \var{flour}, \var{flour} \multimap \var{bread}, \var{shop}\}:\gamma_b,\gamma_a\\\hline
         \end{tabular}
\egroup
}
\end{center}
\caption{\label{fig:Gab-epsilon} The game $G^{\epsilon}_{\mathit{ab}}$. Bernard plays rows, and Ann plays columns. The profile $(\emptyset,\emptyset) \in \ch_b \times \ch_a$ is the only Nash equilibrium in presence of parsimonious preferences.}
\end{figure}
All the formal proofs of the realized objectives are trivial.

%% \paragraph{Formal proofs of the realized objectives.}
%% All the formal proofs of the realized objectives are trivial.
%% The proof that
%% $\var{shop} \vdash \gamma_b$, indicating that Bernard can realize his objective with the only resource of the lease agreement for the shop, is simply:
%% \[
%% \AxiomC{\phantom{PHA}} \RightLabel{ax}
%% \UnaryInfC{$\var{shop} \vdash \var{shop}$} \DisplayProof
%% \]
%% The proof that Ann can realize her objective with resources of enough flour for a year and the bread making equipment, viz., that $\var{flour}, \var{flour} \multimap \var{bread} \vdash \gamma_a$, is also very simple.
%% \[
%% \AxiomC{\phantom{PHA}}
%% \RightLabel{ax}
%% \UnaryInfC{$\var{flour} \vdash \var{flour}$}
%% \AxiomC{\phantom{PHA}}
%% \RightLabel{ax}
%% \UnaryInfC{$\var{bread} \vdash \var{bread}$} \RightLabel{$\implies$L}
%% \BinaryInfC{$\var{flour}, \var{flour} \multimap \var{bread} \vdash \var{bread}$} \DisplayProof
%% \]
%% These two elementary proofs can easily be extended to a proof for every other realized objective by using the weakening rule ($W$).

\paragraph{An undesirable equilibrium.}
One can see on Figure~\ref{fig:Gab-epsilon}, that the profiles $(\{\var{flour}, \var{flour} \multimap \var{bread}\},\linebreak \{\var{shop}\})$ and $(\{\var{flour}, \var{flour}, \var{flour} \multimap \var{bread}\}, \{\var{shop}\})$ in $\ch_b \times \ch_a$ would satisfy both Ann and Bernard. However, in both of them, Bernard has an incentive to provide less resources from his endowment, and to deviate to $\emptyset \in \ch_b$. In turn, in $(\emptyset, \{\var{shop}\}) \in \ch_b \times \ch_a$, Ann is not satisfied, and so has an incentive to retain her resources as well, deviating to her choice $\emptyset \in \ch_a$. The profiles $(\{\var{flour}, \var{flour} \multimap \var{bread}\}, \emptyset)$ and $(\{\var{flour}, \var{flour}, \var{flour} \multimap \var{bread}\}, \emptyset\})$  in $\ch_b \times \ch_a$ satisfy Ann's objective but do not satisfy Bernard's. Hence, Bernard has an incentive to deviate to  $\emptyset \in \ch_b$. 

The profile $(\emptyset, \emptyset)$ is the only Nash equilibrium of $G^{\epsilon}_{\mathit{ab}}$, but it satisfies neither Ann's objective, nor Bernard's. On the other hand, the outcome of the profile $(\{\var{flour}, \var{flour} \multimap \var{bread}\}, \{\var{shop}\}) \in  \ch_b \times \ch_a$ would satisfy them both.

\paragraph{A desirable redistribution.}
So the arbitrator redistributes the resources that are available. He assigns the bread making equipment and half the flour to Ann. He assigns the lease agreement and half the flour to Bernard. That is, $\epsilon'_a = \{\var{flour}, \var{flour} \multimap \var{bread}\}$ and $\epsilon'_b = \{\var{flour}, \var{shop}\}$. This redistribution yields the game $G^{\epsilon'}_{\mathit{ab}}$. It can be depicted as on Figure~\ref{fig:Gab-epsilonprime}.
\begin{figure}[ht]
\begin{center}
\resizebox{\textwidth}{!}{
\bgroup
\def\arraystretch{1.5}
        \begin{tabular}{|c|*{4}{>{$}c<{$}}|}
\hline
           \cellcolor[gray]{\transparency}{\diagbox[linewidth=0.2pt, width=\dimexpr \textwidth/7+2\tabcolsep\relax, height=0.8cm]{$b$}{$a$}}
   & \cellcolor[gray]{\transparency}{\emptyset} & \cellcolor[gray]{\transparency}{\{\var{flour}\}} & \cellcolor[gray]{\transparency}{\{\var{flour} \multimap \var{bread}\}} & \cellcolor[gray]{\transparency}{\{\var{flour}, \var{flour} \multimap \var{bread}\}}\\
           \hline
           \cellcolor[gray]{\transparency}{$\emptyset$} & \emptyset & \{\var{flour}\} & \{\var{flour} \multimap \var{bread}\} & \{\var{flour}, \var{flour} \multimap \var{bread}\}:\gamma_a\\
           \cellcolor[gray]{\transparency}{$\{\var{flour}\}$} & \{\var{flour}\} & \{\var{flour}, \var{flour}\} & \{\var{flour}, \var{flour} \multimap \var{bread}\}:\gamma_a & \{\var{flour}, \var{flour}, \var{flour} \multimap \var{bread}\}:\gamma_a\\
           \cellcolor[gray]{\transparency}{$\{\var{shop}\}$} & \{\var{shop}\}:\gamma_b & \{ \var{flour}, \var{shop}\}:\gamma_b & \{ \var{flour} \multimap \var{bread}, \var{shop}\}:\gamma_b & \{\var{flour}, \var{flour} \multimap \var{bread}, \var{shop}\}\dipaNE :\gamma_b,\gamma_a\\
           \cellcolor[gray]{\transparency}{$\{\var{flour}, \var{shop}\}$} & \{\var{flour}, \var{shop}\}:\gamma_b & \{ \var{flour}, \var{flour}, \var{shop}\}:\gamma_b & \{\var{flour}, \var{flour} \multimap \var{bread},  \var{shop}\}:\gamma_b,\gamma_a & \{\var{flour}, \var{flour}, \var{flour} \multimap \var{bread}, \var{shop}\}:\gamma_b,\gamma_a\\\hline
        \end{tabular}
\egroup
}
\end{center}
\caption{\label{fig:Gab-epsilonprime} The game $G^{\epsilon'}_{\mathit{ab}}$. The profile $(\{\var{shop}\}, \{\var{flour}, \var{flour} \multimap \var{bread}\}) \in  \ch_b \times \ch_a$ is the only Nash equilibrium in presence of parsimonious preferences.}
\end{figure}
In $G^{\epsilon'}_{\mathit{ab}}$, the profile $(\emptyset, \emptyset)$ is not a Nash equilibrium, and so has been eliminated from $G^{\epsilon}_{\mathit{ab}}$. Indeed, it does not satisfy Bernard, and he has an incentive to deviate to the profile $(\{\var{shop}\}, \emptyset) \in  \ch_b \times \ch_a$ in which his objective is satisfied. But $(\{\var{shop}\}, \emptyset)$ is not a Nash equilibrium either. Indeed, it does not satisfy Ann, and she has an incentive to deviate to the profile $(\{\var{shop}\}, \{\var{flour}, \var{flour} \multimap \var{bread}\}) \in  \ch_b \times \ch_a$. From here, nobody has an incentive to deviate, and it is a Nash equilibrium. It is in fact the only Nash equilibrium in $G^{\epsilon'}_{\mathit{ab}}$.

One can readily see that the profile $(\{\var{flour}, \var{shop}\}, \{\var{flour}, \var{flour} \multimap \var{bread}\}) \in  \ch_b \times \ch_a$, even though it satisfies both Ann and Bernard, is not a Nash equilibrium. Bernard has an incentive to provide less resources. The same can be said about the profile $(\{\var{flour}, \var{shop}\}, \{\var{flour} \multimap \var{bread}\}) \in  \ch_b \times \ch_a$.

\subsection{An interconnected economy}
\label{sec:inter-econ}

We present a three-player variant of Example~\ref{example:telecom1}.
The setting, which we remind briefly, is analogous.
In a local telecom industry, three companies must by regulation accept traffic from each other's customers. Moreover, Activating a network at some capacity has a cost, and companies can privately activate and deactivate networks on the fly.

Company~$A$ manages a 3G network of comprised capacity~$3$ (bundled as capacities~$1$, and $2$). 
Company~$B$ manages a 4G network of capacity~$3$ (bundled as capacities~$1$, and $2$).
Company~$A$ need to offer their customers 3G at capacity~$2$ and 4G at capacity~$1$. Company~$B$ need to offer their customers 3G at capacity $2$ and 4G at capacity~$2$.

A new company, Company~$C$ is entering in this interconnected economy. It has some capital, say, two token of an arbitrary unit; one token being fair price for a mobile network antenna. However, Company~$C$ does not manage any network. Company~$C$ needs to offer their customers 3G at capacity $1$ and 4G at capacity~$1$.

Again, we will only consider parsimonious preferences and assume that \logica is MULT. Since we are using this modest fragment, we trust that formal proofs would be more than superfluous and will be omitted.

We introduce the resources involved in the scenario:
\begin{itemize}
\item the resource of one capacity of 3G network: $\var{3G}$
\item the resource of one capacity of 4G network: $\var{4G}$
\item the resource of one token of capital: $\var{cap}$
\end{itemize}

\paragraph{Game definition.}
Let $G^\epsilon_{\mathit{ie}} = (\{a,b,c\},\gamma_a, \gamma_b, \gamma_c, \epsilon_a, \epsilon_b, \epsilon_c)$ be the individual resource game with three players, Company~$A$, $B$, and $C$ being represented by $a$, $b$, and $c$, respectively .

In the game $G^\epsilon_{\mathit{ie}}$, we have $\epsilon_a = \{ \var{3G}, \var{3G} \otimes \var{3G}\}$. $\epsilon_b = \{ \var{4G}, \var{4G} \otimes \var{4G}\}$, and $\epsilon_c = \{\var{cap}, \var{cap}\}$ for endowments.
The objectives are as follows: $\gamma_a = \var{3G} \otimes\var{3G} \otimes \var{4G}$, $\gamma_b =  \var{3G} \otimes \var{3G} \otimes \var{4G} \otimes \var{4G}$, and $\gamma_c = \var{3G} \otimes \var{4G}$.

\paragraph{Two equilibria.}
%

%% \begin{verbatim}
%% Game: ["(1, Player: Name: Company A. 
%% Endowment: {3G, ('3G', '3G')}. Goal: {3G, 3G, 4G}.)", 
%% "(2, Player: Name: Company B. 
%% Endowment: {4G, ('4G', '4G')}. Goal: {3G, 3G, 4G, 4G}.)", 
%% '(3, Player: Name: Company C. 
%% Endowment: {$$, $$}. Goal: {3G, 4G}.)']
%% **** pct
%% Nash: Profile: [(1, '{}'), (2, '{}'), (3, '{}')]
%% Nash: Profile: [(1, "{('3G', '3G')}"), (2, "{('4G', '4G')}"), (3, '{}')]
%% \end{verbatim}

The game $G^{\epsilon}_{\mathit{ie}}$ and the realized objectives can be depicted as on Figure~\ref{fig:Gie-epsilon}.
Company~$B$, Player~$b$, plays rows, Company~$A$, Player~$a$, plays column. For simplicity we do not represent all Company~$C$'s choices because they do not bear on the players' objectives. We only represent Player~$c$'s choice $\emptyset$. With other choices different from $\emptyset$, the realized objectives are exactly the same. Assuming parsimonious preferences, no profile where Company~$C$'s action is different from $\emptyset$ is a Nash equilibrium.

There are two Nash equilibria in the IRG $G^{\epsilon}_{\mathit{ie}}$, namely, $(\emptyset, \emptyset, \emptyset)$ and  $(\{\var{3G} \otimes \var{3G}\}, \{\var{4G} \otimes \var{4G}\}, \emptyset)$.
In the latter, all agents realize their objective. In the former, none of them do.

\begin{figure}[ht]
\begin{center}
\resizebox{\textwidth}{!}{
\bgroup
\def\arraystretch{1.5}
        \begin{tabular}{|c|*{4}{>{$}c<{$}}|}
\hline
\cellcolor[gray]{\transparency}{$c$} & \multicolumn{4}{c|}{\cellcolor[gray]{\transparency}{$\emptyset$}}\\
\hline
           \cellcolor[gray]{\transparency}{\diagbox[linewidth=0.2pt, width=\dimexpr \textwidth/7+2\tabcolsep\relax, height=0.8cm]{$b$}{$a$}}
   & \cellcolor[gray]{\transparency}{\emptyset} & \cellcolor[gray]{\transparency}{\{\var{3G}\}} & \cellcolor[gray]{\transparency}{\{\var{3G} \otimes \var{3G}\}} & \cellcolor[gray]{\transparency}{\{\var{3G}, \var{3G} \otimes \var{3G}\}}\\
           \hline
           \cellcolor[gray]{\transparency}{$\emptyset$} & \emptyset  \dipaNE & \{\var{3G}\} & \{\var{3G} \otimes \var{3G}\} & \{\var{3G}, \var{3G} \otimes \var{3G}\}\\
           \cellcolor[gray]{\transparency}{$\{\var{4G}\}$} & \{\var{4G}\} & \{\var{4G}, \var{3G}\}:\gamma_c & \{\var{4G}, \var{3G} \otimes \var{3G}\}:\gamma_a,\gamma_c & \{\var{4G}, \var{3G} , \var{3G} \otimes \var{3G}\}:\gamma_a, \gamma_c \\
           \cellcolor[gray]{\transparency}{$\{\var{4G} \otimes \var{4G}\}$} & \{\var{4G} \otimes \var{4G}\} & \{\var{4G} \otimes \var{4G}, \var{3G}\}:\gamma_c & \{\var{4G} \otimes \var{4G}, \var{3G} \otimes \var{3G}\}:\gamma_a,\gamma_b,\gamma_c \dipaNE & \{\var{4G} \otimes \var{4G}, \var{3G} , \var{3G} \otimes \var{3G}\}:\gamma_a, \gamma_b,\gamma_c \\
           \cellcolor[gray]{\transparency}{$\{\var{4G}, \var{4G} \otimes \var{4G}\}$} & \{\var{4G}, \var{4G} \otimes \var{4G}\} & \{\var{4G}, \var{4G} \otimes \var{4G}, \var{3G}\}:\gamma_c & \{\var{4G}, \var{4G} \otimes \var{4G}, \var{3G} \otimes \var{3G}\}:\gamma_a,\gamma_b,\gamma_c & \{\var{4G}, \var{4G} \otimes \var{4G}, \var{3G} , \var{3G} \otimes \var{3G}\}:\gamma_a, \gamma_b,\gamma_c \\
           \hline
        \end{tabular}
\egroup
}
\end{center}

  \caption{\label{fig:Gie-epsilon} Partial representation of the game $G^{\epsilon}_{\mathit{ie}}$.}
\end{figure}

\paragraph{Eliminating the bad equilibrium.}

In the IRG $G^{\epsilon}_{\mathit{ie}}$, the profile $(\emptyset, \emptyset, \emptyset)$ is an arguably undesirable equilibrium.
An arbitrator could however advise the three companies to redistribute their endowments to eliminate $(\emptyset, \emptyset, \emptyset)$.
The arbitrator could propose the redistribution $\epsilon'$ of $\epsilon$, where
$\epsilon'_a = \{ \var{3G} \otimes \var{3G}, \var{cap}\}$, $\epsilon'_b = \{ \var{4G} \otimes \var{4G}, \var{cap}\}$, and $\epsilon'_c = \{\var{3G}, \var{4G}\}$.

%% We can eliminate the undesirable profile by redistributing the resources.
%% \begin{verbatim}
%% Game: ["(1, Player: Name: Company A. 
%% Endowment: {$$, ('3G', '3G')}. Goal: {3G, 3G, 4G}.)", 
%% "(2, Player: Name: Company B. 
%% Endowment: {$$, ('4G', '4G')}. Goal: {3G, 3G, 4G, 4G}.)", 
%% '(3, Player: Name: Company C. 
%% Endowment: {3G, 4G}. Goal: {3G, 4G}.)']
%% **** pct
%% Nash: Profile: [(1, "{('3G', '3G')}"), (2, "{('4G', '4G')}"), (3, '{}')]
%% \end{verbatim}

The game $G^{\epsilon'}_{\mathit{ie}}$ and the realized objectives can be depicted as on Figure~\ref{fig:Gie-epsilonprime}, when Player~$c$'s choice is $\emptyset$. The choices containing the resource $\var{cap}$ are not represented. The resource $\var{cap}$ has no bearing on the player's objectives, and no profile containing it is a Nash equilibrium.

\begin{figure}[ht]
\begin{center}
%\resizebox{\textwidth}{!}{
\bgroup
\def\arraystretch{1.5}
        \begin{tabular}{|c|*{3}{>{$}c<{$}}|}
\hline
\cellcolor[gray]{\transparency}{$c$} & \multicolumn{3}{c|}{\cellcolor[gray]{\transparency}{$\emptyset$}}\\
\hline
           \cellcolor[gray]{\transparency}{\diagbox[linewidth=0.2pt, width=\dimexpr \textwidth/7+2\tabcolsep\relax, height=0.8cm]{$b$}{$a$}}
           & \cellcolor[gray]{\transparency}{\emptyset} & \cellcolor[gray]{\transparency}{\{\var{3G} \otimes \var{3G} \}} & \cellcolor[gray]{\transparency}{\cdots} \\
           \hline

           \cellcolor[gray]{\transparency}{$\emptyset$} & \emptyset & \{\var{3G} \otimes \var{3G}\} & \cdots \\
           \cellcolor[gray]{\transparency}{$\{\var{4G} \otimes \var{4G}\}$} & \{\var{4G} \otimes \var{4G}\} & \{\var{4G} \otimes \var{4G}, \var{3G} \otimes \var{3G}\}:\gamma_a,\gamma_b,\gamma_c \dipaNE & \cdots \\
           \cellcolor[gray]{\transparency}{$\vdots$} & \cdots & \cdots & \ddots \\
           \hline
        \end{tabular}
\egroup
%}
\end{center}

  \caption{\label{fig:Gie-epsilonprime} Partial representation of the game $G^{\epsilon'}_{\mathit{ie}}$.}
\end{figure}

After the redistribution, Company~$C$ manages a 3G and a 4G network, both at capacity~$1$. Activating both of them would be enough to satisfy Company~$C$'s objective. In $G^{\epsilon'}_{\mathit{ie}}$, Player~$c$ thus has an incentive to deviate from 
$(\emptyset, \emptyset, \emptyset)$. 
Hence, the arbitrator's advice permits the elimination of the bad equilibrium:
$(\emptyset, \emptyset, \emptyset)$ is not a Nash equilibrium
in $G^{\epsilon'}_{\mathit{ie}}$.

In the profile $(\emptyset, \emptyset, \{\var{3G}, \var{4G}\})$, Player~$1$ has an incentive to deviate and play $\{\var{3G} \otimes \var{3G}\}$, in order to realize its objective.

In turn, in the profile $(\{\var{3G} \otimes \var{3G}\}, \emptyset, \{\var{3G}, \var{4G}\})$, Player~$2$ has an incentive to deviate
and play $\{\var{4G} \otimes \var{4G}\}$ to realize its objective. (Player~$3$, by parsimony, has also an incentive to withdraw the resource $\var{3G}$.)

In the profile $(\{\var{3G} \otimes \var{3G}\}, \{\var{4G} \otimes \var{4G}\}, \{\var{3G}, \var{4G}\})$, by parsimony, Player~$3$ has an incentive to deviate to the choice $\emptyset$.

Every player is satisfied in $(\{\var{3G} \otimes \var{3G}\}, \{\var{4G} \otimes \var{4G}\}, \emptyset)$, and none of them have an incentive to withdraw any resources. Hence, the good equilibrium of $G^{\epsilon}_{\mathit{ie}}$,
$(\{\var{3G} \otimes \var{3G}\}, \{\var{4G} \otimes \var{4G}\}, \emptyset)$, is still a Nash equilibrium in $G^{\epsilon'}_{\mathit{ie}}$. In addition, this is the unique Nash equilibrium in $G^{\epsilon'}_{\mathit{ie}}$.

\section{Conclusions}\label{sec:conclusions}

We presented a class of games of resources that exploits the formalisms and reasoning methods for resource-sensitive logics.
The language of Linear Logic allows us to represent in an harmonious way simultaneous resources, deterministic and non-deterministic choice, and resource-transforming capacities.

In individual resource games, each player of a game is endowed with a multiset of resources and has an objective represented by a resource. In this context, we studied three decision problems, the first of which is to decide whether a profile is a Nash equilibrium.
Some profiles that are not equilibria can have desirable outcomes from the point of view of an external authority. Some equilibria can have outcomes that are undesirable.
We thus studied redistribution schemes which can be used by a central authority to enforce some behavior, either by disincentivizing a behavior or incentivizing a behavior. This yielded two related decision problems: rational elimination and rational construction of profiles.
We illustrated the models and the decision problems with two examples.

We considered dichotomous or parsimonious preferences, and showed striking algorithmic differences when the logic employed admits or not the weakening rule. 

\paragraph{Summary of the complexity results.}
For all decision problems, for both types of preferences, we have
studied six cases where \logica can have
the following properties along two dimensions: (1)~affine vs.\ linear, and
(2)~in $\mathsf{PTIME}$ vs.\ in $\mathsf{NP}$ vs.\
in $\mathsf{PSPACE}$.

When \logica is $\mathsf{NP}$-complete, we sum up precisely the
results in Table~\ref{tab:compl-summary-NP}.
\begin{table}
\begin{center}
\bgroup
\def\arraystretch{1.2}
\begin{tabular}{l@{\quad} c @{\quad}l@{\quad} @{\quad}l@{\quad}}\toprule
%$\prec$ & {\sf d.p.} 
&
& {\sf linear} & {\sf affine} \\\midrule
\multirow{6}{*}{\rotatebox[origin=c]{90}{\sf dichotomous}} & 
\multirow{2}{*}{\iNE} &
  $\mathsf{NP}$-hard (Prop.~\ref{prop:iNE-hardness})
    &
  $\mathsf{NP}$-hard (Prop.~\ref{prop:iNE-hardness})\\

& & in %$\mathsf{coNP^{B\var{H_2}}} \subseteq$
$\mathsf{\Pi_2^p}$ (Prop.~\ref{prop:NP-Pi2P}) & in $\mathsf{P^{NP||}}$ %$\mathsf{\Delta_2^p}$
  (Prop.~\ref{prop:ine-dicho-weakening-easy})\\\cmidrule{2-4}

 & \multirow{2}{*}{{\sf RE}} & $\mathsf{coNP}$-hard (Prop.~\ref{prop:hard-rat-elim-dicho}) & $\mathsf{coNP}$-hard (Prop.~\ref{prop:hard-rat-elim-dicho})\\

& & in %$\mathsf{NP^{B\var{H_2}}} \subseteq$
$\mathsf{\Sigma_2^p}$
(Prop.~\ref{prop:easyness-rat-elim-linear}) & in $\mathsf{P^{NP||}}$ %$\mathsf{\Delta_2^p}$
(Prop.~\ref{prop:easyness-rat-elim-affine})\\\cmidrule{2-4}

 & \multirow{2}{*}{{\sf RC}} &  $\mathsf{NP}$-hard (Prop.~\ref{prop:hard-constr-dicho}) & $\mathsf{NP}$-hard (Prop.~\ref{prop:hard-constr-dicho})\\
& & in $\mathsf{\Sigma_3^p}$ (Prop.~\ref{prop:rc-dicho-linear}) & in $\mathsf{\Sigma_2^p}$ (Prop.~\ref{prop:rc-dicho-affine}) \\

\midrule

\multirow{6}{*}{\rotatebox[origin=c]{90}{\sf parsimonious}} & 
\multirow{2}{*}{\iNE} & $\mathsf{coNP}$-hard (Prop.~\ref{prop:iNE-hardness-parsimonious})
 & $\mathsf{coNP}$-hard (Prop.~\ref{prop:iNE-hardness-parsimonious})\\
&  & in $\mathsf{\Pi_2^p}$ (Prop.~\ref{prop:ine-easy-NP-pars-linear}) & 
in $\mathsf{P^{NP||}}$ %$\mathsf{\Delta_2^p}$
(Prop.~\ref{prop:ine-nondicho-NP})\\\cmidrule{2-4}

 & \multirow{2}{*}{{\sf RE}}  & $\mathsf{coNP}$-hard (Prop.~\ref{prop:hardness-elimination-parsimonious}) & $\mathsf{coNP}$-hard (Prop.~\ref{prop:hardness-elimination-parsimonious})\\

& &  in $\mathsf{\Sigma_2^p}$ (Prop.~\ref{prop:easyness-rat-elim-linear-pars}) &  in $\mathsf{P^{NP||}}$ %$\mathsf{\Delta_2^p}$
(Prop~\ref{prop:easyness-rat-elim-affine-pars})\\\cmidrule{2-4}

 & \multirow{2}{*}{{\sf RC}} & $\mathsf{coNP}$-hard (Prop.~\ref{prop:RC-parsim-hard}) & $\mathsf{coNP}$-hard (Prop.~\ref{prop:RC-parsim-hard})\\
& & in $\mathsf{\Sigma_3^p}$ (Prop.~\ref{prop:rc-parsi-linear}) & in $\mathsf{\Sigma_2^p}$ (Prop.~\ref{prop:rc-parsi-affine}) \\
\bottomrule
\end{tabular}
\egroup
\end{center}
\caption{\label{tab:compl-summary-NP} Complexity results when the problem of provability in \logica is in $\mathsf{NP}$.}
\end{table}
For instance, one can quickly gather that when \logica is Affine MLL (whose sequent provability is 
$\mathsf{NP}$-complete) and we consider parsimonious preferences, {\sf RATIONAL ELIMINATION} is in $\mathsf{P^{NP||}}$. We proved the same problem to be in $\mathsf{\Sigma_2^p}$ when \logica is Linear MLL.
When \logica is in $\mathsf{PTIME}$, we sum up precisely the
results in Table~\ref{tab:compl-summary-PTIME}.
\begin{table}
\begin{center}
\bgroup
\def\arraystretch{1.4}
\begin{tabular}{l@{\quad} c @{\quad}l@{\quad} @{\quad}l@{\quad}}\toprule
%$\prec$ & {\sf d.p.} 
& & {\sf linear} & {\sf affine} \\\midrule
\multirow{3}{*}{\rotatebox[origin=c]{90}{\sf dichotomous}} & \iNE
& in $\mathsf{coNP}$ (Prop.~\ref{prop:NP-Pi2P}) & in $\mathsf{PTIME}$
  (Prop.~\ref{prop:ine-dicho-weakening-easy})\\\cmidrule{2-4}
& {\sf RE} & in $\mathsf{NP}$
(Prop.~\ref{prop:easyness-rat-elim-linear}) & in $\mathsf{PTIME}$
(Prop.~\ref{prop:easyness-rat-elim-affine})\\\cmidrule{2-4}
& {\sf RC} & in $\mathsf{\Sigma_2^p}$ (Prop.~\ref{prop:rc-dicho-linear}) & in $\mathsf{NP}$ (Prop.~\ref{prop:rc-dicho-affine}) \\

\midrule

\multirow{3}{*}{\rotatebox[origin=c]{90}{\sf parsimonious}} & \iNE
& in $\mathsf{coNP}$ (Prop.~\ref{prop:ine-easy-NP-pars-linear}) & 
in $\mathsf{PTIME}$
(Prop.~\ref{prop:ine-nondicho-NP})\\\cmidrule{2-4}
& {\sf RE} &  in $\mathsf{NP}$ (Prop.~\ref{prop:easyness-rat-elim-linear-pars}) &  in $\mathsf{PTIME}$
(Prop~\ref{prop:easyness-rat-elim-affine-pars})\\\cmidrule{2-4}
& {\sf RC} & in $\mathsf{\Sigma_2^p}$ (Prop.~\ref{prop:rc-parsi-linear}) & in $\mathsf{NP}$ (Prop.~\ref{prop:rc-parsi-affine}) \\
\bottomrule
\end{tabular}
\egroup
\end{center}
\caption{\label{tab:compl-summary-PTIME} Complexity results when the problem of provability in \logica is in $\mathsf{PTIME}$.}
\end{table}
We thus obtained some positive results when the resources are expressed in the fragment MULT, which is suitable to represent and reason about multisets of resources.%\footnote{An implementation of IRGs with Affine MULT is available at \url{https://bitbucket.org/troquard/irgpy/}.}
\begin{theorem}
  When \logica is Affine MULT, with dichotomous or parsimonious preferences, the problems {\sf NASH EQUILIBRIUM} and {\sf RATIONAL ELIMINATION} can be solved in polynomial time.
\end{theorem}

It is interesting to note that, although weakening  usually does not change the complexity of the problem of sequent provability of the logics we considered,\footnote{We did not consider \emph{full} propositional Linear Logic, which also contains so-called `exponentials'. Weakening does make a difference: sequent provability in full propositional Linear Logic is undecidable~\cite{LMSS93}, while sequent provability in full propositional Affine Logic is decidable~\cite{KOPYLOV2001173}.} we have always been able to capitalize on its presence to simplify our solutions to the problems we studied here.

Putting the results of this paper together, it is also easy to see that we have this theorem.
\begin{theorem}
When \logica is MALL, linear or affine, with
dichotomous or with parsimonious preferences, all three decision
problems are $\mathsf{PSPACE}$-complete.
\end{theorem}
First-Order MLL is one of these logics whose complexity of sequent provability is in $\mathsf{NP}$. On the other hand, sequent provability for First-Order MALL is $\mathsf{NEXPTIME}$-complete. It is routine to adapt our proofs to show this theorem.
\begin{theorem}
When \logica is First-Order MALL, linear or affine, with
dichotomous or with parsimonious preferences, all three decision
problems are $\mathsf{NEXPTIME}$-complete.
\end{theorem}

\paragraph{Comparison with the related literature.}
%Oriented application: \cite{DBLP:conf/aaai/ChitnisHL11, Levit:2013:TSB:2484920.2484952}

The research in artificial intelligence, multiagent systems, and computer science has shown some interest in the formal and computational aspects of resource-conscious agents (e.g., \cite{DBLP:conf/atal/HarlandW02, Wooldridge2006835,DBLP:conf/kr/PorelloE10,Dunne201020,harrenstein2015aamas,DBLP:journals/iandc/VelnerC0HRR15,DBLP:journals/jancl/PorelloT15,DBLP:journals/ai/AlechinaBLN17,Tr18aaai}).

Boolean games~\cite{Harrenstein:2001:BG:1028128.1028159,DBLP:conf/ecai/BonzonLLZ06} are games based on classical logic.
Each player controls a set of Boolean variables and produces truth values which can be used without restriction towards the Boolean goals, expressed as classical propositional formula. 
Somehow, also in Boolean games do the players produce and consume `resources'. But there are no immediate natural correspondences between IRGs and Boolean games. As in Boolean games, we could force the endowments to be non-overlapping (for exclusive control over a resource). 
Moreover, we could allow the players in our games to have preferences about the absence of a resource. 
Under these conditions, and using classical propositional logic as \logica, a connection would then exist.

Electric Boolean games~\cite{harrenstein2015aamas} are an extension of Boolean games where playing a certain action has a numeric cost, and agents are endowed with a certain amount of `energy'.
Deciding whether a profile is a Nash equilibrium in a Boolean game is $\mathsf{coNP}$-complete~\cite{DBLP:conf/ecai/BonzonLLZ06}. In electric Boolean games, deciding whether a profile is rationally eliminable is $\mathsf{NP}$-complete, while deciding whether a profile is rationally constructible is $\mathsf{coNP}$-hard and in $\mathsf{\Delta_2^p}$.

In Boolean games, goals of players are expressed as classical
propositional formulas. Moreover, game outcomes or profiles are in
fact models of classical propositional logic, i.e.,
valuations. Checking whether the goal of a player is satisfied in a
game profile is thus an easy problem in Boolean games. This is also true in
electric Boolean games. In contrast in resource games, checking
whether the goal of a player is satisfied in a game profile is as hard
as provability in \logica.

Unsurprisingly, when working with the fragments MLL or MALL, the trend is that the complexity of decision problems in individual
resource games is higher than for their counterparts in electric Boolean
games. An obvious exception is the problem to decide whether an individual
resource game admits a Nash equilibrium when \logica is affine and we
consider dichotomous preferences. The problem is trivial by Proposition~\ref{prop:nonempty} (there is always a Nash equilibrium), while it is $\mathsf{\Sigma_2^p}$-complete in Boolean games~\cite{DBLP:conf/ecai/BonzonLLZ06}.

Moreover, in individual resource games, there is no one-to-one correspondence between profiles and outcomes. This is another difference with electric
Boolean game. As a consequence, the notions of elimination and
construction in individual resource games add a bit of complexity by having
to consider a set of profiles with the same outcomes.

On the other hand, the fragment  MULT is one instance of \logica in which it is easy to check whether a goal of a player is satisfied in a game profile (Proposition~\ref{prop:mult-easy}). In this context, and as shown on Table~\ref{tab:compl-summary-PTIME} and compared to the realm of Boolean games, reasoning about IRGs remains a relatively easy task. It can even be tractable if one considers Affine MULT. Affine logic should be used when we can assume that a player satisfied with an outcome would be satisfied with a more sizeable outcome, which is often a very acceptable assumption.

\medskip

Congestion games (CGs)~\cite{Rosenthal1973} (see also Potential Games~\cite{MONDERER1996124}; exact potential games correspond to CGs up to an isomorphism) are a celebrated class of games where the players interact in resource-sensitive environments.
Despite some apparent similarities between IRGs and CGs, they are rather superficial. Players in CGs do not have endowments \emph{per se}. Players' actions in CGs consist in choosing a subset of an already available common pool of resources to use.
In CGs, the players are only consumers.
In IRGs, players are consumers but also producers of resources; their actions consist in making resources available in the common pool.
%Players' actions consist in choosing a set of resources in CGs, and a multiset of resources in our games.
%
In CGs, these resources are exclusively atomic resources while in IRGs they can be any logical formula in \logica.

\medskip

With the decision problems of rational elimination and
rational construction, there is a dimension of social choice theory
and mechanism design. Formal frameworks concerned with redistribution
schemes and economic policies can be found for instance in
\cite{harrenstein2015aamas} again, or
\cite{Endriss:2011:DIB:2030470.2030482,Levit:2013:TSB:2484920.2484952,naumovmodal}.

% ABOUT COMBINATORIAL AUCTIONS
Our games bear some resemblance with combinatorial exchanges~\cite{DBLP:conf/atal/KothariSS04} and with mixed multi-unit combinatorial auctions (MMUCAs)~\cite{Cerquides:2007:BLW:1625275.1625473,Giovannucci2010}, where the agents can be both sellers and buyers.
Interestingly, in MMUCAs, sets of goods can be transformed into different sets of goods. Resource-transforming capacities are central, as the agents are allowed to bid on \emph{transformation services}.
Determining the sequences of bids to be accepted by an auctioneer is generally intractable in MMUCAs; \cite{FIONDA20131} identifies tractable classes for the winner determination problem.

%% ABOUT COOPERATIVE GAMES:
Finally, we focused on individual games and looked at Nash equilibria. 
Nonetheless, the setting allows one to easily build classes of coalition
games, reminiscent of Coalitional Resource Games~\cite{Wooldridge2006835,Dunne201020} and of Coalition Skill Games~\cite{Bachrach:2013:CCS:2537189.2537704}.
In~\cite{Tr18aaai}, we have started the study of what we called Rich Coalitional Resource Games (RCRGs). Individual Resource Games are essentially one-goal RCRGs.

\paragraph{Perspectives.}
We have obtained tight complexity results when \logica is $\mathsf{PSPACE}$-complete. However, this is lacking when \logica is in $\mathsf{NP}$ and in $\mathsf{PTIME}$. We suspect that the complexity of the diverse decision problems generally lie above the lower bounds we have obtained. It is more likely that some proposed upper bounds are tight.
One perspective will thus be to investigate whether some decision problems could be proven hard for some complexity class in the polynomial or Boolean hierarchy, for instance using the techniques from~\cite{wagner87} of raising $\mathsf{NP}$ lower bounds to lower bounds for classes above $\mathsf{NP}$.

Resource games based on resource-sensitive logics become all the more significant when the resources are subject to transforming activities.
We can exploit the existing research on these resource-sensitive logics about their proof theory. In particular, through the Curry-Howard correspondence between proofs and programs (see, e.g., \cite{DBLP:journals/jsyml/GabbayQ92}), an exciting perspective is the possibility to interpret the logical proofs as rigorous programs to be executed by the players. We can expect to obtain some results for the automated  generation of plans, where the resources can be subjected to a series of transforming activities by the agents. Similar ideas have already been defended in multiagent systems (see, e.g., \cite{10.1007/978-3-540-30200-1_5}).

Our models are agnostic about how the contributed resources are distributed.
Instead of having preferences about a raw profile $P$, the player's preferences could be raised over the result of the (fair, envy-free, efficient, etc) allocation of the resources~\cite{bouveretChap12} contributed in $\out(P)$.
These are interesting extensions that are just a step away to get the models more fit for application, although at the expense of mathematical simplicity.

We are interested in using resource games in problems of gamification.
Gamification refers to the broad application of game-design techniques in contexts that do not otherwise present game-like features \cite{Deterding:2011:GUG:1979742.1979575,gamificationEduBusiness}.
Gamification aims at incentivizing an intended behavior by introducing rewards for specific tasks.
Rewards often present themselves as virtual resources such as achievement badges. Formally, they might be nothing more than distinguished tokens of resources.
In Example~\ref{example:telecom1}, we saw that the profile where all companies refrain from  providing any resources, $(\emptyset, \emptyset)$, is a Nash equilibrium. This is an undesirable behaviour that policy makers might be able to anticipate by using the analytical tools defined in this paper, and to avoid by using advanced gamification methods.

\section*{Acknowledgments}
I thank an anonymous reviewer for making a number of suggestions that greatly helped to improve the presentation.
I am grateful to Jamie Gabbay for his enthusiasm and his encouragements.

\newpage
\appendix
\section{Sequent rules of Affine MALL}
\label{sec:rules}

We present the sequent rules for Affine MALL.
In what follows,
$A$, $B$, $A_0$, and $A_1$ are formulas. $\Gamma$, $\Gamma'$, $\Delta$, and $\Delta'$ are sequences of zero or more formulas.
A sequent rule has an upper and a lower part. The upper part is composed of zero, one, or two sequents. The lower part is composed of one sequent. If there is a proof of all the sequents of the upper part, then the rule can be used to obtain a proof of the sequent of the lower part.

\bigskip

  \resizebox{0.85\textwidth}{!}{
    \begin{minipage}[c]{\linewidth}
    \textit{Identities}

\begin{center}
\begin{tabular}{cc}
& \\
\AxiomC{\phantom{PHA}} \RightLabel{ax} \UnaryInfC{$A \vdash A$} \DisplayProof  
& \AxiomC{$\Gamma, A \vdash \Delta$} \AxiomC{$\Gamma' \vdash A, \Delta'$} \RightLabel{cut}
\BinaryInfC{$\Gamma,
\Gamma' \vdash \Delta, \Delta'$} \DisplayProof
\\
 & \\
\end{tabular}
\end{center}

\textit{Structural Rules}

\begin{center}
\begin{tabular}[]{cc}
&  \\
\AxiomC{$\Gamma, A, B, \Gamma' \vdash \Delta$} \RightLabel{E} \UnaryInfC{$\Gamma, B, A, \Gamma' \vdash \Delta$}
\DisplayProof & \AxiomC{$\Gamma \vdash \Delta, A, B, \Delta'$} \RightLabel{E} \UnaryInfC{$\Gamma \vdash \Delta, B, A, \Delta'$}
\DisplayProof\\
&  \\

%% \AxiomC{$\Gamma, A, A, \vdash \Delta$} \RightLabel{C} \UnaryInfC{$\Gamma, A \vdash \Delta$} \DisplayProof & \AxiomC{$\Gamma \vdash \Delta, A, A$} \RightLabel{C} \UnaryInfC{$\Gamma \vdash \Delta, A$} \DisplayProof\\
%% &  \\
\AxiomC{$\Gamma \vdash \Delta$} \RightLabel{W} \UnaryInfC{$\Gamma, A \vdash \Delta$} \DisplayProof  & \AxiomC{$\Gamma \vdash \Delta$} \RightLabel{W} \UnaryInfC{$\Gamma  \vdash \Delta, A$} \DisplayProof\\
& \\
\end{tabular}
\end{center}

\textit{Negation}

\begin{center}
\begin{tabular}[]{cc}
& \\
\AxiomC{$\Gamma \vdash A, \Delta$} \RightLabel{$L\lineg$} \UnaryInfC{$\Gamma, \lineg A \vdash \Delta$}
\DisplayProof   &
\AxiomC{$\Gamma, A\vdash \Delta$} \RightLabel{$R\lineg$} \UnaryInfC{$\Gamma \vdash \lineg A, \Delta$}
\DisplayProof
\\
 & \\
\end{tabular}
\end{center}
  
\textit{Multiplicatives}

\begin{center}
\begin{tabular}{ll}
&  \\
\AxiomC{$\Gamma \vdash A, \Delta$} \AxiomC{$\Gamma' \vdash B, \Delta'$}
 \LeftLabel{$\otimes$R} \BinaryInfC{$\Gamma, \Gamma' \vdash A\otimes B, \Delta, \Delta'$} \DisplayProof  & 
 \AxiomC{$\Gamma, A, B \vdash \Delta$} \RightLabel{$\otimes$L} \UnaryInfC{$\Gamma, A \otimes B \vdash \Delta$}
\DisplayProof      \\
& \\
\AxiomC{$\Gamma, A \vdash \Delta$} \AxiomC{$\Gamma', B \vdash \Delta'$} \LeftLabel{$\parr$L}
\BinaryInfC{$\Gamma, \Gamma', A \parr B \vdash \Delta, \Delta'$} \DisplayProof  &

\AxiomC{$\Gamma \vdash A, B, \Delta$} \RightLabel{$\parr$R}
\UnaryInfC{$\Gamma \vdash A \parr B, \Delta$}   \DisplayProof \\
& \\

\AxiomC{$\Gamma \vdash A, \Delta$} \AxiomC{$\Gamma', B \vdash \Delta'$} \LeftLabel{$\implies$L}
\BinaryInfC{$\Gamma, \Gamma', A \implies B, \Delta \vdash  \Delta'$} \DisplayProof  &

\AxiomC{$\Gamma, A  \vdash B, \Delta$} \RightLabel{$\implies$R}
\UnaryInfC{$\Gamma \vdash A \implies B, \Delta$}   \DisplayProof \\
& \\
\end{tabular}
\end{center}

\begin{center}
\begin{tabular}{cccc}
\AxiomC{$\Gamma \vdash \Delta$}
\RightLabel{$\mathbf{1}$L}
\UnaryInfC{$\Gamma, \mathbf{1} \vdash \Delta$}\DisplayProof 
& 
\AxiomC{\phantom{PHA}}
\RightLabel{$\mathbf{1}$R}
\UnaryInfC{$\vdash \mathbf{1}$}\DisplayProof 

&
\AxiomC{\phantom{PHA}}
\RightLabel{$\bot$L}
\UnaryInfC{$\bot \vdash$}
\DisplayProof

&
\AxiomC{$\Gamma \vdash \Delta$}
\RightLabel{$\bot$R}
\UnaryInfC{$\Gamma \vdash \Delta, \bot$}
\DisplayProof
\\
& \\
\end{tabular}
\end{center}

\textit{Additives} (In $\oplus$R, and $\with$L, $i$ stands for either $0$ or $1$.)

\begin{center}
\begin{tabular}{cc}
&  \\
\AxiomC{$\Gamma \vdash A, \Delta$} \AxiomC{$\Gamma \vdash B, \Delta $} \LeftLabel{$\with$R}
\BinaryInfC{$\Gamma \vdash A \with B, \Delta$} \DisplayProof & \AxiomC{$\Gamma, A_{i} \vdash \Delta$} \RightLabel{$\with$L} \UnaryInfC{$\Gamma, A_{0} \with
A_{1} \vdash \Delta$} \DisplayProof\\
& \\

\AxiomC{$\Gamma, A \vdash \Delta$} \AxiomC{$\Gamma, B \vdash \Delta$} \LeftLabel{$\oplus$L}
\BinaryInfC{$\Gamma, A\oplus B \vdash \Delta$} \DisplayProof &  \AxiomC{$\Gamma \vdash A_{i}, \Delta$}
\RightLabel{$\oplus$R}
\UnaryInfC{$\Gamma \vdash A_{0} \oplus A_{1},\Delta$} \DisplayProof\\
& \\

%% no rule ($\top$L) 
%% & 
\AxiomC{\phantom{PHA}}
\RightLabel{$\top$R}
\UnaryInfC{$\Gamma \vdash \top, \Delta$}\DisplayProof 

&
\AxiomC{\phantom{PHA}}
\RightLabel{$\mathbf{0}$L}
\UnaryInfC{$\Gamma, \mathbf{0}\vdash \Delta$}
\DisplayProof
\end{tabular}
\end{center}
\end{minipage}
} % end resizebox
%\caption{Sequent rules for Affine MALL. (In $\oplus$R, and $\with$L, $i$ stands for either $0$ or $1$.)\label{table:rules}}
%\end{table}

\normalsize

\newpage
\section{Elements of computational complexity}\label{sec:complexity}

We need to assume some familiarity with computational complexity.
This appendix only introduces some elements of terminology and some definitions about complexity theory. 
The reader familiar with these notions can use this section for quick reference.
Another reader can use it as a starting point and move to a more complete introduction.
A classic introduction to computational complexity is~\cite{papadimitriou}.
All elementary complexity classes used in this paper are presented in~\cite{rothe}. 

\medskip

A \emph{decision problem} (or problem for short) is a problem that is posed as `yes'/`no' question of the values of the input.

The class $\mathsf{PTIME}$, also noted $\mathsf{P}$, is the class of decision problems that can be solved in deterministic polynomial time (wrt.\ the size of the input). The class $\mathsf{NP}$ is the class of problems that can be solved in non-deterministic polynomial time. The class $\mathsf{PSPACE}$ is the class of problems that can be solved using a polynomial amount of space.
The \emph{complement} of a decision problem is the decision problem resulting from reversing the `yes' and `no' answers.
For every class of complexity $\mathsf{C}$, we denote $\mathsf{coC}$ the class populated with the complements of the problems in $\mathsf{C}$.
Given two classes of complexity $\mathsf{C_1}$ and $\mathsf{C_2}$, the class $\mathsf{C_1^{C_2}}$ is the class of problems that are in $\mathsf{C_1}$ if we assume the availability of an oracle to solve the problems in $\mathsf{C_2}$. An \emph{oracle} for $\mathsf{C_2}$ is a black box capable to solve every problem in $\mathsf{C_2}$ in a single operation.
Queries to an oracle can be \emph{adaptive} (also called serial), or \emph{non-adaptive} (also called parallel). A query is adaptive when it depends on the answer of a previous query.
Non-adaptive queries on the other hand, can be chosen in advance and computed from the start and are asked in parallel.

For every class of complexity $\mathsf{C}$, we denote $\mathsf{P^C}$ (resp.\ $\mathsf{NP^C}$) the class of problems solvable on a deterministic (resp.\ non-deterministic) polynomial-time bounded oracle Turing machine using an oracle set $\mathsf{C}$. We denote $\mathsf{P^{C[k]}}$ and $\mathsf{NP^{C[k]}}$ when at most $k$ adaptive queries to $\mathsf{C}$ can be used. We denote $\mathsf{P^{C||[k]}}$ and $\mathsf{NP^{C||[k]}}$ when at most $k$ non-adaptive queries to $\mathsf{C}$ can be used.

We denote $\mathsf{P^{C||}}$ (resp.\ $\mathsf{NP^{C||}}$) the class of problems solvable on a deterministic (resp.\ non-deterministic) polynomial-time bounded oracle Turing machine with non-adaptive queries to $\mathsf{C}$. The class $\mathsf{P^{NP||}}$ is also referred to as $\mathsf{\Theta_2^p}$.

\paragraph{The polynomial hierarchy.} The polynomial hierarchy contains a family of complexity classes that are smaller than $\mathsf{PSPACE}$. The class $\mathsf{P}$ lies at the bottom of the polynomial hierarchy. Then, for every positive integer~$i$, we can define $\mathsf{\Delta_{i}^p}$, $\mathsf{\Sigma_{i}^p}$, and $\mathsf{\Pi_{i}^p}$ recursively as follows:
\begin{itemize}
\item $\mathsf{\Delta_{0}^p} = \mathsf{\Sigma_{0}^p} = \mathsf{\Pi_{0}^p} = \mathsf{P}$;
\item $\mathsf{\Delta_{i+1}^p} = \mathsf{P^{\Sigma_{i}^p}}$;
\item $\mathsf{\Sigma_{i+1}^p} = \mathsf{NP^{\Sigma_{i}^p}}$;
\item $\mathsf{\Pi_{i}^p} = \mathsf{co\Sigma_{i}^p}$.
\end{itemize}

\paragraph{The Boolean hierarchy over $\mathsf{NP}$.} The Boolean hierarchy has been studied in~\cite{Wechsung85,kobler1987,wagner87}. 
The Boolean hierarchy over $\mathsf{NP}$ contains a family of complexity classes that are smaller than $\mathsf{\Delta_{2}^p}$. The class $\mathsf{NP}$ lies at the bottom of the Boolean hierarchy over $\mathsf{NP}$.
Here, we are better off looking at complexity classes not as classes of decision problems, but as classes of languages. A \emph{language} is the formal realization of a decision problem. Let $p$ be a decision problem with $k$ inputs. A language of $p$ is the language $L_p = \{ (a_1, \ldots, a_k) \mid p \text{ answers `yes' of the input } (a_1, \ldots, a_k) \}$. Given a class of complexity $\mathsf{C}$, we say that $L_p \in \mathsf{C}$ iff $p \in \mathsf{C}$.
Then, given two classes of complexity $\mathsf{C_1}$ and $\mathsf{C_2}$, each representing a set of languages and the decision problems they formalize, we define $\mathsf{C_1} \land \mathsf{C_2} = \{ L_1 \cap L_2 \mid L_1 \in  \mathsf{C_1} \text{ and } L_2 \in \mathsf{C_2} \}$ and $\mathsf{C_1} \lor \mathsf{C_2} = \{ L_1 \cup L_2 \mid L_1 \in \mathsf{C_1} \text{ and } L_2 \in \mathsf{C_2} \}$.
In this context, the class $\mathsf{NP}$ is the class of languages that can be recognised in non-deterministic polynomial time.
Then, for every positive integer~$i$, we can define $\mathsf{BH_i}$ recursively as follows:
\begin{itemize}
\item $\mathsf{BH_0} = \mathsf{NP}$;
\item $\mathsf{BH_{2k}} = \mathsf{coNP} \land \mathsf{BH_{2k-1}}$;
\item $\mathsf{BH_{2k+1}} = \mathsf{NP} \lor \mathsf{BH_{2k}}$.
\end{itemize}
The class $\mathsf{B\var{H_2}} = \mathsf{NP} \land \mathsf{coNP}$ is the ``difference class'' $\mathsf{D^P}$ presented in~\cite{PAPADIMITRIOU1984244}.

\paragraph{Useful properties.}
Besides the definitions, the following properties are useful:
\begin{itemize}
\item $\mathsf{C_1^{coC_2}} = \mathsf{C_1^{C_2}}$ (for all two classes $\mathsf{C_1}$ and $\mathsf{C_2}$); 
\item $\mathsf{NP^{\Sigma_{i}^p}} = \mathsf{\Sigma_{i+1}^p}$; 
\item $\mathsf{co{\Sigma_{i}^p}} = \mathsf{\Pi_{i}^p}$; 
\item $\mathsf{NP^{\Delta_{i}^p}} = \mathsf{\Sigma_{i}^p}$;
\item $\mathsf{P^{\Delta_{i}^p}} = \mathsf{\Delta_{i}^p}$;
\item $\mathsf{\Sigma_{i}^p} \subseteq \mathsf{PSPACE}$;
\item $\mathsf{PSPACE} = \mathsf{coPSPACE} = \mathsf{P^{PSPACE}} = \mathsf{NP^{PSPACE}}$;
\item $\mathsf{BH_i} \subseteq \mathsf{\Delta_{2}^p}$;
\item $\mathsf{P^{NP||[k]}} \subseteq \mathsf{BH_{k+1}} \subseteq \mathsf{P^{NP||[k+1]}}$.
\end{itemize}

\section{Proofs of the realized objectives in the Example of Section~\ref{sec:alan-fish}}
\label{sec:proofs-realized-obj}

%\paragraph{Formal proofs of the realized objectives.}

The proof of $\var{H_2O},\var{H_2O},\var{elec} \vdash \gamma_f$ will be instrumental for the subsequent proofs. We label it Proof~$\star$ for reuse.
\[
%%%%%%
%%% BRANCH LEFT
\AxiomC{\phantom{PHA}}
\RightLabel{ax} 
\UnaryInfC{$\var{H_2O} \vdash \var{H_2O}$}
\AxiomC{\phantom{PHA}}
\RightLabel{ax} 
\UnaryInfC{$\var{H_2O} \vdash \var{H_2O}$}
\RightLabel{$\otimes$R}
\BinaryInfC{$\var{H_2O}, \var{H_2O} \vdash \var{H_2O} \otimes \var{H_2O}$} \RightLabel{$\otimes$R}
%%%%%%
%%% BRANCH RIGHT
\AxiomC{\phantom{PHA}}
\RightLabel{ax} 
\UnaryInfC{$\var{O_2} \vdash \var{O_2}$}
\RightLabel{W} 
\UnaryInfC{$\var{O_2},\var{H_2} \otimes \var{H_2} \vdash \var{O_2}$}
\RightLabel{E} 
\UnaryInfC{$\var{H_2} \otimes \var{H_2}, \var{O_2} \vdash \var{O_2}$} 
\RightLabel{$\otimes$L} 
\UnaryInfC{$\var{H_2} \otimes \var{H_2} \otimes \var{O_2} \vdash \var{O_2}$} 
%%%%%%
\RightLabel{$\multimap$L}
\BinaryInfC{$\var{H_2O}, \var{H_2O}, \var{H_2O} \otimes \var{H_2} O \multimap \var{H_2} \otimes \var{H_2} \otimes \var{O_2} \vdash \var{O_2}$} \RightLabel{definition} 
\UnaryInfC{$\var{H_2O},\var{H_2O},\var{elec} \vdash \gamma_f$} 
\RightLabel{Proof~$\star$}
\UnaryInfC{}
\DisplayProof  
\]
The other realized objectives of the fish are immediate using
Proof~$\star$ and the weakening rule.  We prove that $\var{H_2O},\var{H_2O},\var{H_2O},\var{elec}
\vdash \gamma_f$, $\var{drink},\var{H_2O},\var{H_2O},\var{elec} \vdash \gamma_f$, and
$\var{drink},\var{H_2O},\var{H_2O},\var{H_2O},\var{elec} \vdash \gamma_f$.

\begin{center}
\resizebox{\textwidth}{!}{
$
\begin{array}{ccc}

\AxiomC{$\vdots$} \RightLabel{Proof~$\star$} 
\UnaryInfC{$\var{H_2O},\var{H_2O},\var{elec} \vdash \gamma_f$}
\RightLabel{W} 
\UnaryInfC{$\var{H_2O},\var{H_2O},\var{elec},\var{H_2O} \vdash \gamma_f$} 
\RightLabel{E} 
\UnaryInfC{$\var{H_2O},\var{H_2O},\var{H_2O},\var{elec} \vdash \gamma_f$} 
\DisplayProof  

&

\AxiomC{$\vdots$} \RightLabel{Proof~$\star$} 
\UnaryInfC{$\var{H_2O},\var{H_2O},\var{elec} \vdash \gamma_f$}
\RightLabel{W}
\UnaryInfC{$\var{H_2O},\var{H_2O},\var{elec},\var{drink} \vdash \gamma_f$} 
\RightLabel{E*}
\UnaryInfC{$\var{drink},\var{H_2O},\var{H_2O},\var{elec} \vdash \gamma_f$} 
\DisplayProof  

&

\AxiomC{$\vdots$} \RightLabel{Proof~$\star$} 
\UnaryInfC{$\var{H_2O},\var{H_2O},\var{elec} \vdash \gamma_f$}
\RightLabel{W} 
\UnaryInfC{$\var{H_2O},\var{H_2O},\var{elec},\var{drink} \vdash \gamma_f$}
\RightLabel{E*} 
\UnaryInfC{$\var{drink},\var{H_2O},\var{H_2O},\var{elec} \vdash \gamma_f$}
\RightLabel{W}
\UnaryInfC{$\var{drink},\var{H_2O},\var{H_2O},\var{elec},\var{H_2O} \vdash \gamma_f$}
\RightLabel{E*} 
\UnaryInfC{$\var{drink},\var{H_2O},\var{H_2O},\var{H_2O},\var{elec} \vdash \gamma_f$} 
\DisplayProof  
\end{array}
$
}
\end{center}

Finally, we prove $\var{drink},\var{elec},\var{H_2O},\var{H_2O},\var{H_2O} \vdash \gamma_a$. The proof
also uses Proof~$\star$.

\[
%% BRANCH LEFT
\AxiomC{$\vdots$}
\RightLabel{Proof~$\star$}
\UnaryInfC{$\var{H_2O},\var{H_2O},\var{elec} \vdash \gamma_f$}
%% BRANCH RIGHT
\AxiomC{\phantom{PHA}}
\RightLabel{ax}
\UnaryInfC{$\var{H_2O} \vdash \var{H_2O}$}
\AxiomC{\phantom{PHA}}
\RightLabel{ax}
\UnaryInfC{$\lineg \var{thirst} \vdash \lineg \var{thirst}$}
\RightLabel{ax}
\RightLabel{$\multimap$L}
\BinaryInfC{$\var{H_2O} \multimap \lineg \var{thirst},\var{H_2O} \vdash \lineg \var{thirst}$}
%%%%%%%%%%%
\RightLabel{$\otimes$R}
\BinaryInfC{$\var{H_2O},\var{H_2O},\var{elec},\var{H_2O} \multimap \lineg \var{thirst},\var{H_2O} \vdash \gamma_f \otimes \lineg \var{thirst}$} \RightLabel{definition} 
\UnaryInfC{$\var{H_2O},\var{H_2O},\var{elec},\var{drink},\var{H_2O} \vdash \gamma_a$} \RightLabel{E*} 
\UnaryInfC{$\var{drink},\var{H_2O},\var{H_2O},\var{H_2O},\var{elec} \vdash \gamma_a$} 
\DisplayProof  
\]

\newpage
\bibliographystyle{plain}
\bibliography{JLC-18-58.bib} 

\begin{thebibliography}{10}

\bibitem{DBLP:journals/ai/AlechinaBLN17}
Natasha Alechina, Nils Bulling, Brian Logan, and Hoang~Nga Nguyen.
\newblock The virtues of idleness: {A} decidable fragment of resource agent
  logic.
\newblock {\em Artificial Intelligence}, 245:56--85, 2017.

\bibitem{ambaketal1994interconnection}
Jens Ambak, Bridger Mitchell, Werner Neu, Karl-Heinz Neumann, Ingo Vogelsang,
  Godefroy~Dang N'Guyen, and Bernd Ickenroth.
\newblock {Network interconnection in the domain of ONP. Final report, November
  1994}, 1994.

\bibitem{Bachrach:2013:CCS:2537189.2537704}
Yoram Bachrach, David~C. Parkes, and Jeffrey~S. Rosenschein.
\newblock {Computing Cooperative Solution Concepts in Coalitional Skill Games}.
\newblock {\em Artificial Intelligence}, 204:1--21, 2013.

\bibitem{DBLP:conf/ecai/BonzonLLZ06}
Elise Bonzon, Marie{-}Christine Lagasquie{-}Schiex, J{\'{e}}r{\^{o}}me Lang,
  and Bruno Zanuttini.
\newblock {Boolean Games Revisited}.
\newblock In {\em 17th European Conference on Artificial Intelligence
  (ECAI'06)}, volume 141 of {\em Frontiers in Artificial Intelligence and
  Applications}, pages 265--269. {IOS} Press, 2006.

\bibitem{bouveretChap12}
Sylvain Bouveret, Yann Chevaleyre, and Nicolas Maudet.
\newblock {\em Fair Allocation of Indivisible Goods}, chapter~12, pages
  284--310.
\newblock In Brandt et~al. \cite{Brandt:2016:HCS:3033138}, 2016.

\bibitem{Brandt:2016:HCS:3033138}
Felix Brandt, Vincent Conitzer, Ulle Endriss, J{\'e}r\^{o}me Lang, and Ariel~D.
  Procaccia, editors.
\newblock {\em Handbook of Computational Social Choice}.
\newblock Cambridge University Press, New York, NY, USA, 2016.

\bibitem{brock1995interconnection}
W.~Gerald Brock.
\newblock {The Economics of Interconnection}.
\newblock Technical report, Teleport Communication Group, 1995.
\newblock Prepared for Teleport Communications Group.

\bibitem{Cerquides:2007:BLW:1625275.1625473}
Jes\'{u}s Cerquides, Ulle Endriss, Andrea Giovannucci, and Juan~A.
  Rodr\'{\i}guez-Aguilar.
\newblock Bidding languages and winner determination for mixed multi-unit
  combinatorial auctions.
\newblock In {\em Proceedings of the 20th International Joint Conference on
  Artifical Intelligence}, IJCAI'07, pages 1221--1226, San Francisco, CA, USA,
  2007. Morgan Kaufmann Publishers Inc.

\bibitem{Deterding:2011:GUG:1979742.1979575}
Sebastian Deterding, Miguel Sicart, Lennart Nacke, Kenton O'Hara, and Dan
  Dixon.
\newblock Gamification. using game-design elements in non-gaming contexts.
\newblock In {\em CHI'11 Extended Abstracts on Human Factors in Computing
  Systems}, CHI EA'11, pages 2425--2428, New York, NY, USA, 2011. ACM.

\bibitem{Dunne201020}
Paul~E. Dunne, Sarit Kraus, Efrat Manisterski, and Michael Wooldridge.
\newblock {Solving coalitional resource games}.
\newblock {\em Artificial Intelligence}, 174(1):20--50, 2010.

\bibitem{Endriss:2011:DIB:2030470.2030482}
Ulle Endriss, Sarit Kraus, J{\'e}r\^{o}me Lang, and Michael Wooldridge.
\newblock {Designing Incentives for Boolean Games}.
\newblock In {\em 10th International Conference on Autonomous Agents and
  Multiagent Systems}, AAMAS'11, pages 79--86. International Foundation for
  Autonomous Agents and Multiagent Systems, 2011.

\bibitem{DBLP:books/cu/FKW2014}
Shaheen Fatima, Sarit Kraus, and Michael~J. Wooldridge.
\newblock {\em Principles of Automated Negotiation}.
\newblock Cambridge University Press, 2014.

\bibitem{FIONDA20131}
Valeria Fionda and Gianluigi Greco.
\newblock The complexity of mixed multi-unit combinatorial auctions:
  Tractability under structural and qualitative restrictions.
\newblock {\em Artificial Intelligence}, 196:1 -- 25, 2013.

\bibitem{DBLP:journals/jsyml/GabbayQ92}
Dov~M. Gabbay and Ruy J. G.~B. de~Queiroz.
\newblock {Extending the Curry-Howard Interpretation to Linear, Relevant and
  Other Resource Logics}.
\newblock {\em Journal of Symbolic Logic}, 57(4):1319--1365, 1992.

\bibitem{Giovannucci2010}
Andrea Giovannucci, Jes{\'u}s Cerquides, Ulle Endriss, and Juan~A.
  Rodr{\'i}guez-Aguilar.
\newblock A graphical formalism for mixed multi-unit combinatorial auctions.
\newblock {\em Autonomous Agents and Multi-Agent Systems}, 20(3):342--368, May
  2010.

\bibitem{girard1987}
Jean-Yves Girard.
\newblock Linear logic.
\newblock {\em Theoretical Computer Science}, 50(1):1--101, 1987.

\bibitem{DBLP:conf/atal/HarlandW02}
James Harland and Michael Winikoff.
\newblock Agent negotiation as proof search in linear logic.
\newblock In {\em The First International Joint Conference on Autonomous Agents
  {\&} Multiagent Systems, {AAMAS} 2002, Proceedings}, pages 938--939, 2002.

\bibitem{harrenstein2015aamas}
Paul Harrenstein, Paolo Turrini, and Michael Wooldridge.
\newblock {Electric Boolean Games: Redistribution Schemes for Resource-Bounded
  Agents}.
\newblock In {\em 14th International Conference on Autonomous Agents and
  Multi-agent Systems}, AAMAS'15, pages 655--663, 2015.

\bibitem{Harrenstein:2001:BG:1028128.1028159}
Paul Harrenstein, Wiebe van~der Hoek, John-Jules Meyer, and Cees Witteveen.
\newblock Boolean games.
\newblock In {\em Proceedings of the 8th Conference on Theoretical Aspects of
  Rationality and Knowledge}, TARK'01, pages 287--298, San Francisco, CA, USA,
  2001. Morgan Kaufmann Publishers Inc.

\bibitem{K94}
Max~I. Kanovich.
\newblock The complexity of {H}orn fragments of {L}inear {L}ogic.
\newblock {\em Annals of Pure and Applied Logic}, 69(2-3):195--241, 1994.

\bibitem{kobler1987}
Johannes K\"obler, Uwe Sch\"oning, and Klaus~W. Wagner.
\newblock {The difference and truth-table hierarchies for {NP}}.
\newblock {\em Theoretical Informatics and Applications}, 21:419--435, 1987.

\bibitem{KOPYLOV2001173}
A.~P. Kopylov.
\newblock {Decidability of Linear Affine Logic}.
\newblock {\em Information and Computation}, 164(1):173--198, 2001.

\bibitem{DBLP:conf/atal/KothariSS04}
Anshul Kothari, Tuomas Sandholm, and Subhash Suri.
\newblock Solving combinatorial exchanges: Optimality via a few partial bids.
\newblock In {\em 3rd International Joint Conference on Autonomous Agents and
  Multiagent Systems {(AAMAS} 2004), 19-23 August 2004, New York, NY, {USA}},
  pages 1418--1419. {IEEE} Computer Society, 2004.

\bibitem{10.1007/978-3-540-30200-1_5}
Peep K{\"u}ngas and Mihhail Matskin.
\newblock Symbolic negotiation with linear logic.
\newblock In J{\"u}rgen Dix and Jo{\~a}o Leite, editors, {\em Computational
  Logic in Multi-Agent Systems}, pages 71--88, Berlin, Heidelberg, 2005.
  Springer Berlin Heidelberg.

\bibitem{Levit:2013:TSB:2484920.2484952}
Vadim Levit, Tal Grinshpoun, Amnon Meisels, and Ana~L.C. Bazzan.
\newblock {Taxation Search in Boolean Games}.
\newblock In {\em 12th International Conference on Autonomous Agents and
  Multi-agent Systems}, AAMAS'13, pages 183--190. International Foundation for
  Autonomous Agents and Multiagent Systems, 2013.

\bibitem{LMSS93}
Patrick Lincoln, John Mitchell, Andre Scedrov, and Natarajan Shankar.
\newblock Decision problems for propositional linear logic.
\newblock {\em Annals of Pure and Applied Logic}, 56(1-3):239--311, 1992.

\bibitem{MONDERER1996124}
Dov Monderer and Lloyd~S. Shapley.
\newblock Potential games.
\newblock {\em Games and Economic Behavior}, 14(1):124 -- 143, 1996.

\bibitem{naumovmodal}
Pavel~G. Naumov and Jia Tao.
\newblock A modal logic for reasoning about economic policies.
\newblock {\em Journal of Logic and Computation}, 27(1):395--412, 2017.

\bibitem{DBLP:journals/bsl/OHearnP99}
Peter~W. O'Hearn and David~J. Pym.
\newblock {The logic of Bunched Implications}.
\newblock {\em Bulletin of Symbolic Logic}, 5(2):215--244, 1999.

\bibitem{Osborne1994}
Martin~J. Osborne and Ariel Rubinstein.
\newblock {\em A course in game theory}.
\newblock The MIT Press, Cambridge, USA, 1994.

\bibitem{papadimitriou}
Christos Papadimitriou.
\newblock {\em Computational Complexity}.
\newblock Addison Wesley, 1994.

\bibitem{PAPADIMITRIOU1984244}
Christos Papadimitriou and Mihalis Yannakakis.
\newblock {The complexity of facets (and some facets of complexity)}.
\newblock {\em Journal of Computer and System Sciences}, 28(2):244--259, 1984.

\bibitem{DBLP:conf/kr/PorelloE10}
Daniele Porello and Ulle Endriss.
\newblock {Modelling Combinatorial Auctions in Linear Logic}.
\newblock In Fangzhen Lin, Ulrike Sattler, and Miroslaw Truszczynski, editors,
  {\em Principles of Knowledge Representation and Reasoning: Proceedings of the
  Twelfth International Conference, {KR} 2010}, pages 71--78. {AAAI} Press,
  2010.

\bibitem{DBLP:conf/ecai/PorelloE10}
Daniele Porello and Ulle Endriss.
\newblock {Modelling Multilateral Negotiation in Linear Logic}.
\newblock In {\em {ECAI} 2010 - 19th European Conference on Artificial
  Intelligence, Proceedings}, volume 215 of {\em Frontiers in Artificial
  Intelligence and Applications}, pages 381--386. {IOS} Press, 2010.

\bibitem{DBLP:journals/jancl/PorelloT15}
Daniele Porello and Nicolas Troquard.
\newblock Non-normal modalities in variants of linear logic.
\newblock {\em Journal of Applied Non-Classical Logics}, 25(3):229--255, 2015.

\bibitem{gamificationEduBusiness}
Torsten Reiners and Lincoln~C. Wood, editors.
\newblock {\em Gamification in Education and Business}.
\newblock Springer, Cham, 2015.

\bibitem{restall94thesis}
Greg Restall.
\newblock {\em {On Logics Without Contraction}}.
\newblock PhD thesis, The University of Queensland, 1994.

\bibitem{Reynolds:2002:SLL:645683.664578}
John~C. Reynolds.
\newblock {Separation Logic: A Logic for Shared Mutable Data Structures}.
\newblock In {\em Proceedings of the 17th Annual IEEE Symposium on Logic in
  Computer Science}, LICS'02, pages 55--74, Washington, DC, USA, 2002. IEEE
  Computer Society.

\bibitem{Rosenthal1973}
Robert~W. Rosenthal.
\newblock {A class of games possessing pure-strategy Nash equilibria}.
\newblock {\em International Journal of Game Theory}, 2(1):65--67, Dec 1973.

\bibitem{rothe}
J\"org Rothe.
\newblock {\em {Complexity Theory and Cryptology}}.
\newblock Springer-Verlag Berlin Heidelberg, 2005.

\bibitem{Troelstra1992}
Anne~S. Troelstra.
\newblock {\em {Lectures on Linear Logic}}.
\newblock CSLI Publications, 1992.

\bibitem{Tr16ijcai}
Nicolas Troquard.
\newblock Nash equilibria and their elimination in resource games.
\newblock In {\em Proceedings of the Twenty-Fifth International Joint
  Conference on Artificial Intelligence, {IJCAI} 2016}, pages 503--509. {AAAI}
  Press, 2016.

\bibitem{Tr18aaai}
Nicolas Troquard.
\newblock Rich coalitional resource games.
\newblock In {\em Proceedings of the Thirty-Second {AAAI} Conference on
  Artificial Intelligence, {AAAI 2018}}, pages 1242--1249. {AAAI} Press, 2018.

\bibitem{DBLP:journals/iandc/VelnerC0HRR15}
Yaron Velner, Krishnendu Chatterjee, Laurent Doyen, Thomas~A. Henzinger,
  Alexander~Moshe Rabinovich, and Jean{-}Fran{\c{c}}ois Raskin.
\newblock The complexity of multi-mean-payoff and multi-energy games.
\newblock {\em Inf. Comput.}, 241:177--196, 2015.

\bibitem{wagner87}
Klaus~W. Wagner.
\newblock {More Complicated Questions About Maxima and Minima, and Some
  Closures of {NP}}.
\newblock {\em Theor. Comput. Sci.}, 51(1-2):53--80, March 1987.

\bibitem{Wechsung85}
Gerd Wechsung.
\newblock {On the boolean closure of {NP}}.
\newblock In Lothar Budach, editor, {\em Fundamentals of Computation Theory},
  volume 199 of {\em Lecture Notes in Computer Science}, pages 485--493.
  Springer Berlin Heidelberg, 1985.

\bibitem{Wooldridge2006835}
Michael Wooldridge and Paul~E. Dunne.
\newblock {On the computational complexity of coalitional resource games}.
\newblock {\em Artificial Intelligence}, 170(10):835--871, 2006.

\end{thebibliography}

\end{document}